%% file: main_lottery.tex
\def\fullversion{1} 

\ifnum\fullversion=0 
	\documentclass[sigconf]{acmart}
	\setcopyright{acmlicensed}
	\copyrightyear{2018}
	\acmYear{2018}
	\acmDOI{XXXXXXX.XXXXXXX}
	\acmConference[Conference acronym 'XX]{Make sure to enter the correct conference title from your rights confirmation email}{June 03--05, 2018}{Woodstock, NY}
	\acmISBN{978-1-4503-XXXX-X/2018/06}
\else
    \documentclass[runningheads]{llncs}

	\usepackage{fullpage}

\fi

\newcommand{\pnote}[1]{\textcolor{magenta}{[[#1 \textbf{--Paul}]]}}
\newcommand{\anote}[1]{\textcolor{teal}{[[#1 \textbf{--Aravind}]]}}

\renewcommand{\pnote}[1]{}
\renewcommand{\anote}[1]{}

\input{packages.tex}

\input{0-macros.tex}
\usepackage[capitalize,noabbrev,nameinlink]{cleveref}
\crefname{subgamehop}{Exp.}{Experiments}
\Crefname{subgamehop}{Exp.}{Experiments}
\crefname{gamehop}{Exp.}{Experiments}
\Crefname{gamehop}{Exp.}{Experiments}

\crefname{definition}{Definition}{Definitions}
\Crefname{definition}{Definition}{Definitions}
\crefname{lemma}{Lemma}{Lemmas}
\Crefname{lemma}{Lemma}{Lemmas}

\AddToHook{env/definition/begin}{\crefalias{theorem}{definition}}
\AddToHook{env/definition/end}{\crefalias{theorem}{theorem}}
\AddToHook{env/lemma/begin}{\crefalias{theorem}{lemma}}
\AddToHook{env/lemma/end}{\crefalias{theorem}{theorem}}

\begin{document}

\ifnum\fullversion=1
\author{Paul Gerhart\inst{1}\orcidlink{0000-0002-0164-0187} \and Jay Taylor\inst{2}\orcidlink{0009-0004-8671-0771} \and  
Sri Aravinda Krishnan Thyagarajan\inst{2}\orcidlink{0000-0003-0114-7672}}
\institute{TU Wien, Austria \and University of Sydney, Australia}
\fi

\ifnum\fullversion=0 
\fi

\ifnum\fullversion=0 
	\title{Probabilistic Atomic Swaps for Bitcoin and Friends}

	\begin{abstract}
		\input{abstract.tex}
	\end{abstract}
	\maketitle
\else
	\title{Probabilistic Atomic Swaps for Bitcoin and Friends}
    
	\maketitle

	\begin{abstract}
		\input{abstract.tex}
	\end{abstract}
    

\fi

	\input{intro.tex}
	\input{related.tex}
	\input{prelims.tex}

	\input{overview.tex}

	\input{fair_lottery.tex}
	\input{protocol.tex}
	\input{extensions.tex}

	\input{prototype.tex}

\ifnum\fullversion=0 
    \bibliographystyle{ACM-Reference-Format}
	\bibliography{abbrev3.bib,crypto.bib,extrarefs.bib}
\fi

\appendix		
\crefalias{section}{appendix}

\ifnum\fullversion=0

\section*{Ethical Considerations}
This work introduces probabilistic atomic swaps, a cryptographic protocol for trustless probabilistic exchange of digital assets.
Our contribution is purely theoretical and systems-oriented: we design, formalize, and implement a new cryptographic primitive.
The work does not involve human subjects, user data, surveys, or real-world vulnerability analysis, and therefore does not require IRB/ERB approval.

We briefly discuss the potential societal implications of our work.
Probabilistic swaps enable new forms of decentralized exchange, including wager-style transactions and randomized cross-chain trades.
Like all financial primitives, these mechanisms could in principle be used for gambling applications.
However, we note that such use cases are already possible through existing on-chain lottery protocols~\cite{ESPW:BilBen17,FC:BarZun17}, which our protocol does not meaningfully extend in terms of accessibility.
On the contrary, our construction improves privacy and fungibility compared to existing approaches by making probabilistic swaps indistinguishable from ordinary on-chain transactions, thereby reducing third parties' ability to surveil or censor specific transaction types. 

Our prototype is evaluated on public Bitcoin and Litecoin testnets using only test coins of no monetary value. 
No real funds were used or placed at risk during our experiments. 
We do not release any exploit code, attack tooling, or adversarial artifacts.

\section*{Open Science}
Our prototype implementation of probabilistic atomic swaps, including all components described in \cref{sec:benchmarks}, is available for review.
Specifically, the artifact includes:

\begin{itemize}[leftmargin=*,nosep]
    \item The full protocol implementation in Rust, covering two-party DKG, two-party Schnorr pre-signing, the OPRF-based claim mechanism, the Bulletproof-based well-formedness proof, the cut-and-choose-based well-formedness proof with batched Schnorr proofs, and the Lightning linking proofs.
    \item Scripts to reproduce all benchmark results reported in \cref{sec:benchmarks}.
    \item Transaction IDs for the public Bitcoin testnet4 and Litecoin testnet runs listed in \cref{tab:txids}, which can be independently verified on the respective public block explorers.
    \item Documentation describing how to build and run the implementation.
\end{itemize}
The artifact is available at the following anonymous repository for double-blind review:

\begin{center}
    \url{https://anonymous.4open.science/r/probabilistic-swaps-62D8}
\end{center}

The implementation does not depend on any proprietary datasets or third-party services beyond publicly accessible Bitcoin and Litecoin testnets.
No artifact relevant to the core contributions is withheld.

\fi
\section*{Acknowledgements}
\ifnum\fullversion=1
	Paul Gerhart's research has been supported by the Google PhD Fellowship in Privacy, Safety, and Security.
\else
This paper was edited for grammar using ChatGPT.
\fi

\ifnum\fullversion=1
    \bibliographystyle{splncs04}
	\bibliography{abbrev3.bib,crypto.bib,extrarefs.bib}
\fi

	\input{appdx_prelims.tex}

	\input{instantiation.tex}	
	\ifnum \fullversion=0
	\input{appdx_proof.tex}
	\input{appdx_figures.tex}
	\fi
\end{document}

%% file: packages.tex
\usepackage[T1]{fontenc}
\usepackage{graphicx}

\usepackage{xcolor}
\usepackage{graphicx}
\usepackage{hyperref}
\hypersetup{
    colorlinks,
    linkcolor={blue!50!black},
    citecolor={blue!50!black},
    urlcolor={blue!80!black}
}
\usepackage{float}
\usepackage{xspace}
\usepackage{nicematrix}
\usepackage{resizegather}
\usepackage{wrapfig}
\usepackage{xcolor}
\definecolor{accent}{RGB}{69,90,131}
\usepackage{orcidlink}

\usepackage{enumitem}
\usepackage{booktabs}
\usepackage{array}
\usepackage{tikz}
\usetikzlibrary{fit} 
\usepackage{tikzpeople}
\usepackage{placeins}

\usepackage{xspace}

\usepackage[most]{tcolorbox}
\tcbuselibrary{breakable}

\newtcolorbox{textdef}{
  breakable,
  colback=white,
  colframe=black,
  boxrule=0.5pt,
  arc=0pt,
  breakable=false,
  left=1ex,
  right=1ex,
  top=1ex,
  bottom=1ex,
}

\ifnum\fullversion=0 

\fi

\usepackage{xspace}

\xspaceaddexceptions{.,;:!?)}

  \usepackage{makecell}
\usepackage{graphicx}
\usepackage{pgfplots}
\usepgfplotslibrary{groupplots}
\usepackage{pgfplotstable}
\pgfplotsset{compat=1.18}
\usepackage{bbm}
\usepackage{pifont}
\newcommand{\cmark}{\ding{51}} 
\newcommand{\xmark}{\ding{55}} 

\newcounter{gamehop}

\providecommand{\currgamelabel}{gamehop:0}
\providecommand{\lastgamelabel}{gamehop:0}

\newcommand{\expinline}[1]{%
  \noindent\textbf{Exp.~#1.}\enspace\ignorespaces%
}

\newcommand{\initialgame}{%
  \setcounter{gamehop}{-1}%
  \gamehop%
}

\newcommand{\gamehop}{%
  \refstepcounter{gamehop}%
  \edef\currgamelabel{gamehop:\thegamehop}%
  \edef\lastgamelabel{gamehop:\number\numexpr\value{gamehop}-1\relax}%
  \label{\currgamelabel}%
  \expinline{\thegamehop}%
}

\newcommand{\thisgame}{%
  \hyperref[\currgamelabel]{Exp.~\thegamehop}\xspace%
}

\newcommand{\lastgame}{%
  \ifnum\value{gamehop}>0
    \hyperref[\lastgamelabel]{Exp.~\number\numexpr\value{gamehop}-1\relax}\xspace%
  \else
    previous experiment\xspace%
  \fi
}

\newcounter{subgamehop}[gamehop]
\renewcommand{\thesubgamehop}{\thegamehop.\arabic{subgamehop}}

\newcommand{\currsubgamelabel}{}
\newcommand{\lastsubgamelabel}{}

\newcommand{\lastsubgame}{%
  \ifnum\value{subgamehop}>1
    \hyperref[\lastsubgamelabel]{Exp.~\thegamehop.\number\numexpr\value{subgamehop}-1\relax}\xspace%
  \else
    \thisgame\xspace%
  \fi
}

\usepackage{hyperref}

\usepackage[
	lambda,
	operators,
	advantage,
	sets,
	adversary,
	landau,
	probability,
	notions,
	logic,
	ff,
	mm,
	primitives,
	events,
	complexity,
	asymptotics,
	keys]{cryptocode}

\newcommand{\fullinline}[1]{%
\ifnum\fullversion=0
$#1$
\else
\[#1\]
\fi
}

%% file: 0-macros.tex
\newcommand{\Prim}{\ensuremath{\mathsf{Prim}}\xspace}

\newcommand{\comrand}{\ensuremath{\omega}\xspace}
\newcommand{\com}{\ensuremath{\mathsf{C}}\xspace}

\newcommand{\proofstatement}{\ensuremath{\mathbbm{x}\xspace}}
\newcommand{\proofwitness}{\ensuremath{\mathbbm{w}\xspace}}

\newcommand{\func}{\ensuremath{\mathcal{F}}\xspace}
\newcommand{\blockchainfunc}{\ensuremath{\func_\blockchain}\xspace}
\newcommand{\smtfunc}{\ensuremath{\func_\mathsf{smt}}\xspace}

\newcommand{\lotteryfunc}{\ensuremath{\func_{\lottery}^{\blockchain}}\xspace}
\newcommand{\realview}{\ensuremath{\mathtt{REAL}}\xspace}
\newcommand{\idealview}{\ensuremath{\mathtt{IDEAL}}\xspace}
\newcommand{\simdealer}{\ensuremath{\simulator_\dealer}\xspace}

\newcommand{\pkT}{\ensuremath{\pk_\mathsf{tmp}}\xspace}
\newcommand{\skT}{\ensuremath{\sk_\mathsf{tmp}}\xspace}

\newcommand{\skTD}{\ensuremath{\sk_{\mathsf{tmp}, \dealer}}\xspace}
\newcommand{\skTP}{\ensuremath{\sk_{\mathsf{tmp}, \player}}\xspace}

\newcommand{\BlindEval}{\ensuremath{\mathsf{BlindEval}}\xspace}
\newcommand{\Finalize}{\ensuremath{\mathsf{Finalize}}\xspace}

\newcommand{\request}{\ensuremath{\mathsf{req}}\xspace}
\newcommand{\response}{\ensuremath{\mathsf{res}}\xspace}

\providecommand{\sign}{\pcalgostyle{Sign}}
\providecommand{\vrfy}{\pcalgostyle{Vrfy}}
\providecommand{\pSign}{\pcalgostyle{pSign}}

\providecommand{\pVrfy}{\pcalgostyle{pVrfy}}
\providecommand{\adapt}{\pcalgostyle{Adapt}}
\providecommand{\extract}{\pcalgostyle{Extract}}

\providecommand{\pSig}{\presig}

\providecommand{\adaptorScheme}{\mathsf{AS}}

\providecommand{\relation}{\mathcal{R}}
\providecommand{\genRelation}{\mathsf{genR}}

\providecommand{\fext}{\mathsf{Ext}}
\providecommand{\newY}{\mathsf{NewS}}
\providecommand{\setS}{\mathcal{S}}
\providecommand{\setT}{\mathcal{T}}

\newcommand{\queue}{\mathcal{Q}}

\newcommand{\setup}{\mathsf{Setup}}
\newcommand{\keygen}{\mathsf{KeyGen}}

\NewDocumentCommand\nonce{o}{\ensuremath{\rho\IfNoValueF{#1}{_{#1}}}}
\NewDocumentCommand\prepstate{o}{\ensuremath{\mathit{state}\IfNoValueF{#1}{_{#1}}}}
\newcommand{\indexSet}{\ensuremath{S}}
\NewDocumentCommand\fml{mO{i}O{\indexSet}}{\ensuremath{\{#1_{#2}\}_{#2\in #3}}}

\newcommand{\buyer}{{\sf B}}
\newcommand{\seller}{{\sf S}}

\newcommand{\buyerSet}{\mathcal{B}}
\newcommand{\sellerSet}{\mathcal{S}}

\NewDocumentCommand\skB{o}{\ensuremath{\sk_{\buyer\IfNoValueF{#1}{_{#1}}}}}
\NewDocumentCommand\skS{o}{\ensuremath{\sk_{\seller\IfNoValueF{#1}{_{#1}}}}}

\NewDocumentCommand\pkB{o}{\ensuremath{\pk_{\buyer\IfNoValueF{#1}{_{#1}}}}}
\NewDocumentCommand\pkS{o}{\ensuremath{\pk_{\seller\IfNoValueF{#1}{_{#1}}}}}

\NewDocumentCommand\fmlB{mO{i}O{\buyerSet}}{\ensuremath{\{#1_{#2}\}_{#2\in #3}}}
\NewDocumentCommand\fmlS{mO{j}O{\sellerSet}}{\ensuremath{\{#1_{#2}\}_{#2\in #3}}}

\NewDocumentCommand\stateB{o}{\ensuremath{\state_{\buyer\IfNoValueF{#1}{_{#1}}}}}
\NewDocumentCommand\stateS{o}{\ensuremath{\state_{\seller\IfNoValueF{#1}{_{#1}}}}}

\providecommand{\eval}{\mathsf{Evaluate}}

\newcommand{\timet}{\ensuremath{\mathbf{T}}\xspace}
\newcommand{\blockchain}{\ensuremath{\mathbb{B}}\xspace}

\providecommand{\tx}{\mathsf{tx}}

\newcommand{\statement}{Y}
\newcommand{\witness}{y}
\newcommand{\presig}{\widetilde{\sigma}}

\providecommand{\signerSet}{\mathcal{S}}

\NewDocumentCommand\cm{o}{\ensuremath{\mathsf{cm}\IfNoValueF{#1}{_{#1}}}}
\RenewDocumentCommand\st{o}{\ensuremath{\mathsf{st}\IfNoValueF{#1}{_{#1}}}}
\NewDocumentCommand\setOfSigners{mO{i}O{\signerSet}}{\ensuremath{\{#1_{#2}\}_{#2\in #3}}}
\NewDocumentCommand\setOfOtherSigners{mO{j}O{\signerSet \setminus i}}{\ensuremath{\{#1_{#2}\}_{#2\in #3}}}

\renewcommand{\nizk}{\ensuremath{\mathsf{NIZK}}}

\newcommand{\ciphertext}{\ensuremath{ct}}

\newcommand{\finalize}{\ensuremath{\mathsf{Finalize}}\xspace}

\newcommand{\lockscript}{\ensuremath{\Lambda}\xspace}

\newcommand{\environment}{\ensuremath{\mathcal{E}}\xspace}
\newcommand{\dkg}{\ensuremath{\mathsf{DKG}}\xspace}

\newcommand{\refund}{\ensuremath{\mathsf{rfnd}}\xspace}
\newcommand{\compclose}[2]{
\noindent\underline{\textbf{Exp.~\ref{#1} $\approx$ Exp.~\ref{#2}.}}} 

\newcommand{\nizkProof}{\pi}
\newcommand{\nizkpp}{\mathsf{pp}\xspace}
\newcommand{\proofsys}{\mathsf{\Pi}}
\newcommand{\proofsetup}{\mathsf{Setup}\xspace}
\newcommand{\proofprove}{\mathcal{P}\xspace}

\newcommand{\proofverify}{\ensuremath{\mathcal{V}}\xspace}

\newcommand{\dealer}{\ensuremath{\mathsf{D}}\xspace}
\newcommand{\player}{\ensuremath{\mathsf{P}}\xspace}

\newcommand{\winningProb}{\ensuremath{p}\xspace}

\newcommand{\freeze}{\ensuremath{\mathsf{Freeze}}\xspace}

\newcommand{\transfer}{\ensuremath{\mathsf{Transfer}}\xspace}

\newcommand{\abort}{\ensuremath{\mathsf{abort}}\xspace}
\newcommand{\play}{\ensuremath{\mathsf{claim}}\xspace}
\newcommand{\lottery}{\ensuremath{\mathsf{ProSwap}}\xspace}

\newcommand{\buyout}{\ensuremath{\nu_\dealer}\xspace}

\newcommand{\target}{\ensuremath{\mathsf{tgt}}\xspace}
\newcommand{\win}{\ensuremath{\mathsf{win}}\xspace}
\newcommand{\guess}{\ensuremath{\mathsf{gss}}\xspace}

\newcommand{\witnessTarget}{\ensuremath{\witness_\target}\xspace}
\newcommand{\witnessGuess}{\ensuremath{\witness_\guess}\xspace}

\newcommand{\paymentWit}{\ensuremath{\witness_\win}\xspace}
\newcommand{\paymentStmt}{\ensuremath{\statement_\win}\xspace}

\newcommand{\mypar}[1]{\paragraph{#1}}
\newcommand{\lightpar}[1]{\smallskip\noindent\emph{#1}}

\newcommand{\Request}{\ensuremath{\mathsf{Req}}\xspace}
\newcommand{\BlindEv}{\ensuremath{\mathsf{BlindEv}}\xspace}

\newcommand{\oprf}{\ensuremath{\mathsf{OPRF}}\xspace}

%% file: abstract.tex
Atomic swaps are a fundamental primitive for the trustless exchange of digital assets across blockchains: they guarantee that either both parties receive the agreed assets or neither party transfers. While this all-or-nothing guarantee is powerful, it also imposes an inherent determinism that rules out exchanges whose intended outcome is probabilistic. As a result, existing atomic swaps cannot realize trustless exchanges in which one party pays for a fixed chance of receiving a larger asset or reward, as in lotteries, randomized allocation mechanisms, and probabilistic cross-chain trades.

We introduce \emph{probabilistic swaps}, a new cryptographic primitive that extends atomic swaps to the probabilistic setting. 
In a probabilistic swap, one party's transfer is executed with a fixed, publicly specified probability embedded in the protocol and cannot be biased by either party. 
This yields a trustless mechanism for randomized exchange with verifiable odds and no trusted intermediary.

Our construction combines adaptor signatures with oblivious pseudorandom functions (OPRFs) to realize the desired probabilistic outcome while ensuring that neither party can predict or bias it in advance. Along the way, we introduce a new mechanism for the atomic exchange of OPRF evaluations for payments, which may be of independent interest.
A key feature of our approach is that it preserves the minimal on-chain footprint of modern atomic-swap protocols. The protocol relies only on standard Bitcoin scripts, such as digital signatures and timelocks, and is deployable on any blockchain that already supports atomic swaps. Consequently, probabilistic swaps are indistinguishable from ordinary on-chain transactions, which helps preserve privacy and fungibility. 
We provide formal security foundations and demonstrate practicality through a probabilistic swap between the Bitcoin and Litecoin testnets, as well as in the Lightning Network. 

%% file: intro.tex
\section{Introduction}
Atomic swaps of digital assets are a fundamental economic primitive: they allow two parties to exchange assets with the guarantee that either both receive their desired assets or neither does~\cite{nolan2013atomic}. 
One of the key motivations behind blockchain systems is precisely to enable such exchanges without a trusted intermediary, and this has attracted significant interest from both academia and industry~\cite{NDSS:HABSG17,SP:TaiMorMaf21,CCS:GMMMTT22,SP:QPMSES23,SP:ThyMalMor22,SP:HLTW24,AC:AEEFHM21}, with practical relevance for cross-chain interoperability and decentralized finance~\cite{NDSS:MMSKM19,FC:MBBKM19,SP:ThyMalMor22,NDSS:MTVFMM23,CCS:GMMMTT22,SP:QPMSES23,FC:TayGerThy26}.

Classical atomic swaps enforce a deterministic exchange relation: if one party obtains the counterparty's asset, then the counterparty obtains the agreed asset as well. 
This all-or-nothing guarantee is powerful, but it cannot directly express exchanges whose intended output is a probability distribution rather than a fixed transfer. 
In this work, we ask whether one can realize a \emph{probabilistic swap}: a contract-minimal exchange in which one party pays for a publicly specified chance of receiving an asset, while the protocol guarantees that the realized outcome has exactly the advertised distribution and cannot be biased by either party.

One motivating example is \emph{wagers}. In such a setting, a participant commits a fixed amount in exchange for a chance to receive a larger reward (from another party) according to predefined odds. This application is also popularly referenced as \texttt{OP\_RAND} in the Bitcoin community~\cite{kurbatov2025emulatingoprandbitcoin}, and has been identified as a promising tool for handling tiny, otherwise ``uneconomical'' payments on the Lightning Network~\cite{btforum2025}. 
Our work realizes exactly this vision, as we demonstrate in~\cref{sec:benchmarks}. 
 
Another natural example is \textit{probabilistic cross-chain exchange}. 
A user holding a long-tail asset may prefer a small, explicitly priced chance of receiving a highly liquid asset, such as Bitcoin, over a guaranteed but unfavorable exchange rate. 
A probabilistic swap lets the parties encode this trade directly: the user transfers the source asset in exchange for a fixed probability of receiving the target asset, with the odds and payoffs agreed in advance.

The key requirement in both these examples is that, once the payment is made, the outcome is realized with exactly the \emph{advertised} probability, and neither party can bias or influence the result. Unlike traditional scenarios that rely on a trusted house or intermediary, the setting here is decentralized, and the parties obtain cryptographic guarantees about the fairness of the outcome and the correctness of the winning probability.

\noindent\textbf{Existing work.} Prior approaches have explored blockchain-based randomized protocols, including lottery-style mechanisms on Bitcoin, showing that probabilistic outcomes can in principle be realized in a trustless setting~\cite{CCS:KumMorBen15,ESPW:BilBen17}. 
However, these constructions typically rely on several assumptions. Firstly, they require complex smart contracts or intricate script-level mechanisms, such as carefully structured spending conditions, to enforce correctness~\cite{ESPW:BilBen17,SP:ADMM14,C:BenKum14}. 
This complexity has concrete costs: larger transactions and higher fees, more distinguishable on-chain behavior, and reduced portability to chains with limited scripting support.
This complexity comes at huge costs: (1) larger transactions and higher {fees}, (2) more distinguishable on-chain behavior, weakening  {on-chain privacy} and {fungibility}, and (3) {limited compatibility} across blockchains with limited scripting/contract support. 
These downsides of excessive script-usage have already been highlighted in several prior works~\cite{Thyagarajan2021LockableSF,CCS:ATMMM22,SP:ThyMalMor22,SP:HLTW24,CCS:VanSonThy24} in the context of other payment applications.
Secondly, existing solutions like~\cite{ESPW:BilBen17,FC:BarZun17,SP:ADMM14} for lotteries require the set of participants to be fixed and known to everyone ahead of time, which makes them poorly suited to open, permissionless participation.
As a result, while such approaches demonstrate theoretical feasibility, they do not provide a rigorous abstraction for decentralised probabilistic exchange analogous to atomic swaps.
We therefore ask the question: 
\begin{center}
{\em Can we extend the atomic swap of blockchain tokens into the probabilistic setting without modifying the existing minimal on-chain footprint?} 
\end{center}

We emphasize that this question is not merely of theoretical interest as it addresses a fundamental limitation of classical atomic swaps and opens up new benefits for  applications with randomized outcomes on blockchain systems like Bitcoin. 

\noindent\textbf{A new primitive.} 
We answer this question affirmatively by introducing \emph{probabilistic swaps}, a cryptographic primitive that extends atomic swaps from deterministic exchanges to exchanges with a prescribed outcome distribution. 
In a probabilistic swap between a dealer $\dealer$ and a party $\player$, $\player$ transfers an asset in exchange for a fixed probability $p$ of receiving $\dealer$'s asset; with probability $1-p$, $\player$ receives nothing. 
The probability $p$ is fixed in advance and enforced by the protocol: after the parties commit, neither party can bias the outcome, selectively abort after learning it, or obtain value outside the prescribed distribution. 
If the swap does not proceed, then neither side's asset is transferred. 
Thus, probabilistic swaps retain the trustless-exchange guarantee of atomic swaps while replacing deterministic delivery with fair probabilistic delivery. 
We give a high-level overview in~\Cref{fig:overviewNocrypto}.

\begin{figure}[h]
  \centering
  \begin{tikzpicture}[>=Stealth, scale=0.8, transform shape]
    \node[draw, thick, rounded corners=10pt,
          minimum width=1.5cm, minimum height=3.5cm, align=center]
      (D) at (0,0) {{\Large Dealer}};

    \node[draw, thick, rounded corners=10pt,
          minimum width=1.5cm, minimum height=3.5cm, align=center]
      (F) at (4.5,0) {{\Large $\lotteryfunc$}\\[0.4em]
        $b \sample
        \left\{
        \begin{array}{ll}
          1 : p\\[0.3em]
          0 : 1{-}p
        \end{array}
        \right.$};

    \node[draw, thick, rounded corners=10pt,
          minimum width=1.5cm, minimum height=3.5cm, align=center]
      (P) at (9,0) {{\Large Party}};

    \draw[->, thick, shorten <=5pt, shorten >=5pt]
      ($(D.north east)!0.33!(D.south east)$)
      -- node[above] {$\$\nu$}
      ($(F.north west)!0.33!(F.south west)$);

    \draw[->, thick, shorten <=5pt, shorten >=5pt]
      ($(F.north west)!0.67!(F.south west)$)
      -- node[above] {$\$1\ \ \ \ : b=1$}
         node[below] {$\$\nu{+}1 : b=0$}
      ($(D.north east)!0.67!(D.south east)$);

    \draw[->, thick, shorten <=5pt, shorten >=5pt]
      ($(P.north west)!0.33!(P.south west)$)
      -- node[above] {$\$1$}
      ($(F.north east)!0.33!(F.south east)$);

    \draw[->, thick, shorten <=5pt, shorten >=5pt]
      ($(F.north east)!0.67!(F.south east)$)
      -- node[above] {$\$\nu : b=1$}
         node[below] {$\$0: b=0$}
      ($(P.north west)!0.67!(P.south west)$);

  \end{tikzpicture}
  \caption{
    Ideal functionality $\lotteryfunc$ for probabilistic swaps. The dealer inputs
    $\nu$ coins and the party inputs one coin. The functionality samples
    $b \sample \mathrm{Bernoulli}(p)$ and always transfers one coin to the
    dealer. If $b=1$, the party receives $\nu$ coins; otherwise, the dealer
    is refunded its $\nu$ coins. Probabilistic swaps subsume atomic swaps~\cite{SP:ThyMalMor22} when setting $\winningProb=1=\nu$.
  }
  \label{fig:overviewNocrypto}
\end{figure}
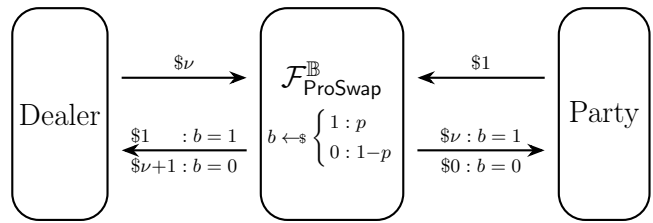

%
%

A key advantage of our approach is its minimal reliance on blockchain functionality. Our protocol can be implemented on \emph{any blockchain} that supports atomic swaps~\cite{SP:ThyMalMor22}, \emph{without any} additional on-chain overhead and requiring only lightweight off-chain cryptographic techniques.
This makes probabilistic swaps immediately deployable on a wide range of existing systems, including script-limited and cross-chain environments. 

\subsection{Our contributions}
Our main contributions can be summarised as follows:
\begin{itemize}[leftmargin=*,nosep]
    \item
    We introduce \emph{probabilistic atomic swaps}, a new cryptographic primitive that extends classical atomic swaps to probabilistic settings. We formalize this notion via an ideal functionality (c.f.,~\cref{sec:idealfunc}) that captures fair exchange with an a priori specified probability $p$. Our new functionality is a generalization of the regular atomic swap functionality~\cite{SP:ThyMalMor22} simply by setting $p=1$. Importantly, our formulation preserves the lightweight nature of atomic swaps by ensuring that probabilistic swaps incur the same on-chain footprint as their deployed deterministic counterparts.

    \item  We construct an efficient protocol for probabilistic atomic swaps using only minimal script support: transaction signature verification and timeouts. 
    These features are available on Bitcoin-like chains and are compatible with payment-channel frameworks. 
    We give a concrete instantiation from Schnorr signatures, and more generally identify the compatibility requirements needed for other adaptor signature relations, and cryptographic tools like encryption keys, oblivious pseudorandom function (OPRF) output spaces, and well-formedness proofs~\cite{AC:AEEFHM21,FC:TaiMorMaf21,EC:GSST24,CCS:GerRauSch26}. 
    Our techniques can also be extended to chains without native transaction timeouts using cryptographic timed-payment methods~\cite{SP:ThyMalMor22}.


    \item
    We develop, as a central technical ingredient and contribution of independent interest, a mechanism for \emph{atomic exchange of oblivious PRF evaluation for payments}. 
    An OPRF~\cite{TCC:FIPR05} allows a client to obtain the evaluation of a PRF on an input using a server’s secret key, without the server learning the input and without the client learning anything about the key beyond the output.
	In this mechanism, the dealer acts as an OPRF server, and the party obtains OPRF evaluations if and only if it makes the required payment. Beyond enabling probabilistic atomic swaps, this tool may be broadly useful in OPRF-based applications, including payment-driven rate limiting in password authentication, or encrypted backup systems~\cite{CCS:CamLehNev15,USENIX:ECSJR15,EC:JarKraXu18,CCS:AMMM18,EC:HJKW23,C:DFGHHH23}. 
	In particular, we envision payment-based rate limiting to replace hard cut-offs after a fixed number of attempts with a small fee per authentication attempt. This imposes negligible cost on honest users, especially since our protocol can be instantiated entirely on layer-2 networks, thereby minimizing fee overhead. At the same time, it renders large-scale brute-force attacks economically infeasible. We leave the exploration of this idea as interesting future work.
    
	\item We implement and evaluate our protocol on the Bitcoin and Litecoin testnets, demonstrating practical feasibility and low on-chain overhead. 
    A complete on-chain swap uses exactly four standard Taproot transactions; in our testnet experiments, these transactions were about $575$WU each. 
    We also demonstrate a Bitcoin--Litecoin cross-chain swap. 
    For layer-2 deployment, PTLC-style systems support the adaptor signature version of our protocol; since the deployed Lightning Network is currently \emph{hash timelock contract} (HTLC)-based, our artifact also includes an HTLC/preimage variant with explicit hash/discrete-log link proofs. 
    This provides a concrete realization of an \texttt{OP\_RAND}-style functionality for Lightning-like payment networks~\cite{btforum2025}.

\end{itemize}

%% file: related.tex
\subsection{Related Work} \label{sec:related}
Prior work demonstrates that probabilistic exchanges can be realized on blockchains, for example, in decentralized poker~\cite{CCS:KumMorBen15} and lottery protocols~\cite{ESPW:BilBen17,FC:BarZun17}, or general MPC~\cite{SP:ADMM14,C:BenKum14}. 
However, these protocols either require additional opcodes not supported by Bitcoin~\cite{CCS:KumMorBen15,FC:BarZun17}, achieve guarantees only via security deposits as least as big as all coins to be exchanged~\cite{SP:ADMM14,C:BenKum14}, or rely on advanced scripting capabilities that publicly reveal both the execution of a lottery and its conditions~\cite{ESPW:BilBen17,FC:BarZun17}. 
We compare these approaches in \cref{tab:related}.

(Deterministic) Atomic swaps enable trustless exchange across blockchains. Early constructions such as TumbleBit~\cite{NDSS:HABSG17} and later protocols including $A^2L$~\cite{SP:TaiMorMaf21,CCS:GMMMTT22} and BlindHub~\cite{SP:QPMSES23} focus on practical, privacy-preserving implementations. 
More recent works, such as Universal Swaps~\cite{SP:ThyMalMor22} and Sweep UC~\cite{SP:HLTW24}, provide more general and formally secure frameworks. 
Atomic swaps are also used to build payment channels~\cite{AC:AEEFHM21}.
Our formalization builds upon the security framework of~\cite{SP:ThyMalMor22} but is enhanced to support probabilistic exchanges. 
Along the way, we identified minor formalization gaps in the functionality of~\cite{SP:ThyMalMor22}, in particular a mismatch between the ideal functionality and real-world behavior: while the ideal functionality refunds both parties atomically, an adversary in the real world may refuse to reclaim locked funds, leading to executions that cannot be simulated. 
We refer to~\cref{sec:functionality} for more discussion.

\ifnum\fullversion=1
  \setlength{\tabcolsep}{8pt} 
\fi

\begin{table}

\caption{Comparison of on-chain protocols. 
With private protocol semantics, we denote whether the protocol execution can be inferred from on-chain data. 
}

\centering
\footnotesize
\begin{tabular}{l|c c c c c c}

 & \rotatebox{90}{\textbf{This work}}
 & \rotatebox{90}{\textbf{OP\_Rand}~\cite{kurbatov2025emulatingoprandbitcoin}}
 & \rotatebox{90}{\textbf{Lottery}~\cite{ESPW:BilBen17,FC:BarZun17}}
 & \rotatebox{90}{\textbf{UnivSwps}~\cite{SP:ThyMalMor22}}
 & \rotatebox{90}{\textbf{Poker}~\cite{CCS:KumMorBen15}}
 & \rotatebox{90}{\textbf{MPC}~\cite{SP:ADMM14,C:BenKum14}}\\
\midrule
Probabilistic output distribution & \cmark & \cmark & \cmark & \xmark & \cmark & \cmark   \\
On-Chain privacy & \cmark & \cmark & \xmark & \cmark & \xmark & \xmark  \\
Off-chain/cross-chain compatible & \cmark & \cmark & \xmark & \cmark & \xmark & \xmark  \\
As-of-today Bitcoin compatible & \cmark & \cmark & \cmark & \cmark & \xmark & \cmark   \\
On-chain cost: moderate $(\rightarrow$)/ low ($\downarrow$)
& $\downarrow$ 
& $\downarrow$ 
& $\rightarrow$ 
& $\downarrow$ 
& $\rightarrow$ 
& $\rightarrow$  \\
Deposit-based (no deposits are better) & \xmark & \xmark & \xmark & \xmark & \cmark & \cmark \\
Proven secure & \cmark & \xmark & \cmark & \cmark & \cmark & \cmark \\
\end{tabular}

\label{tab:related}
\end{table}

Most directly related to our protocol are universal swaps~\cite{SP:ThyMalMor22} and the Bitcoin lotteries of~\cite{ESPW:BilBen17,FC:BarZun17}. 
Our protocol, as well as the lotteries, supports probabilistic exchanges, whereas universal swaps focus on deterministic exchanges.
On-chain, our protocol behaves similarly to universal swaps, revealing only freezing and exchange transactions without exposing additional protocol-specific information. This results in an on-chain footprint of four standard transactions: two that each transfer coins from a party's address to a temporary address to freeze them. And two transactions from the temporary address to either party, (or both to the dealer, if the party lost). 
In contrast, the Bitcoin lottery protocols~\cite{ESPW:BilBen17,FC:BarZun17} require encoding the full protocol logic on-chain, including the outcome probabilities. 
Consequently, on-chain costs involve evaluating complex scripts, whereas our protocol relies solely on standard transaction signatures and timelock verification. 
The overhead of timelocks can be further reduced by instantiating our protocol with verifiable timed signatures (VTS)~\cite{CCS:TBMDKS20}, similar to~\cite{SP:ThyMalMor22}. 
Finally, since the atomicity of the exchange in our protocol does not rely on expressive on-chain scripting, it can be integrated into existing exchange frameworks, such as layer2 networks, and naturally supports cross-chain execution. This is not possible for the lottery protocol~\cite{ESPW:BilBen17}.

Independently in his note, Kurbatov~\cite{kurbatov2025emulatingoprandbitcoin} proposes a protocol for emulating an \texttt{OP\_RAND} opcode on Bitcoin via a two-party interactive protocol.
The overall structure is similar to ours: a party makes a guess, and if correct, learns a witness that allows it to spend coins. However, the two works differ substantially in two ways.
First, in~\cite{kurbatov2025emulatingoprandbitcoin} \emph{both} the dealer and the party must compute NIZKs over hash evaluations, with the dealer's proof growing linearly in $m$: for winning probability $1/m$, the dealer proves $m+1$ hash instances including a \texttt{hash160} evaluation, while the party proves one \texttt{hash160} evaluation.
In contrast, our protocol requires the dealer only to compute a single well-formedness proof, while the party computes no NIZK at all; our Bulletproof instantiation achieves near-constant proving time regardless of $m$.
Therefore, for probabilities other than $1/2$, the approach of~\cite{kurbatov2025emulatingoprandbitcoin} quickly becomes impractical, whereas our protocol remains efficient even for probabilities as small as $2^{-20}$.
This matters in practice: encoding a one-satoshi payment on the Lightning Network as a probabilistic payment requires winning probabilities on the order of $2^{-17}$ or smaller, a regime where~\cite{kurbatov2025emulatingoprandbitcoin} becomes prohibitively expensive. 
Second, \cite{kurbatov2025emulatingoprandbitcoin} provides no formal security definitions or proofs. Therefore, while the protocol appears intuitively sound, it remains unclear whether it achieves the atomicity and bias-resistance guarantees we formalize in \cref{sec:functionality}.

%% file: prelims.tex
\section{Preliminaries} \label{sec:prelims}
We denote by $\secpar$ the security parameter. We say an algorithm is efficient if it can be computed in probabilistic polynomial time (PPT). We say a function $\mu(\cdot)$ is negligible if it decreases faster than the inverse of any polynomial. 
By $\mathcal{X} \approx_c \mathcal{Y}$, we denote computational indistinguishability of the two distributions $\mathcal{X}$ and $\mathcal{Y}$.
We defer formal treatments of OPRFs, adaptor signatures, and NIZKS to~\cref{sec:appdx:prelims}.

\mypar{Security definition.}
We use a standard simulation-based notion of security~\cite{EPRINT:Lindell16}. 
%
Let $\Pi$ be a protocol, $\func$ be an ideal functionality, $\adv$ be a probabilistic polynomial-time adversary, and $\simulator$ be a probabilistic polynomial-time simulator. 
We define the random variables $\idealview_{\func, \simulator}$ consisting of the identity of the corrupted party, the outputs of the honest parties, and the output of $\simulator$, and $\realview_{\Pi, \adv}$ consisting of the identity of the corrupted party, the outputs of the honest parties, and the view of $\adv$ at the end of the execution of~$\Pi$.

\mypar{Digital signatures.}
A digital signature scheme $\Sigma$ consists of three algorithms: a key generation algorithm $\kgen(\secpar)$ that reads a security parameter and outputs a key-pair $(\pk, \sk)$, a signing algorithm $\sign(\sk, m)$ that reads the signing key $\sk$, and a message $m$ and outputs a signature $\sigma$, and a verification algorithm $\vrfy(\pk, m, \sigma)$ that reads the public key $\pk$, a message $m$, and signature $\sigma$ and outputs $1$ if $\sigma$ is valid and $0$ otherwise. We require $\Sigma$ to be existentially unforgeable under chosen-message attacks (EUF-CMA)~\cite{KatLin14}.

\mypar{Hard relations}
A relation $\relation \subseteq \bin^* \times \bin^*$ is associated with the language
$
\mathcal{L}_\relation := \{\proofstatement \mid \exists \proofwitness : (\proofstatement,\proofwitness) \in \relation \}.
$
We say that $\relation$ is a hard relation~\cite{AC:AEEFHM21} if:
(i) there exists a PPT algorithm $\genRelation(\secpar)$ that outputs $(\proofstatement,\proofwitness) \in \relation$,
(ii) membership in $\relation$ can be decided in PPT, and
(iii) for every PPT adversary $\adv$, given $\proofstatement$, it is infeasible to compute a corresponding witness $\proofwitness$.

\mypar{Adaptor signatures.}
Let $\relation$ be a hard relation and let $\Sigma$ be a signature scheme. An adaptor signature scheme is a tuple $\adaptorScheme = (\pSign, \pVrfy, \adapt, \extract)$, where $\pSign(\sk, m, Y)$ outputs a pre-signature $\presig$, $\pVrfy(\pk, m, Y, \presig)$ verifies its validity, $\adapt(\presig, y)$ outputs a signature $\sigma$, and $\extract(\presig, \sigma)$ outputs a witness $y$.

We require $\adaptorScheme$ to be \emph{adaptable}, ensuring that any valid pre-signature can be adapted into a valid signature given a valid witness, and \emph{extractable}, guaranteeing that any valid signature was either issued by the signer or enables extraction of a valid witness with respect to a pre-signature issued by the signer~\cite{EC:GSST24}. 

\mypar{Oblivious PRFs.}
An oblivious pseudorandom function (OPRF)~\cite{TCC:FIPR05} allows a client to obtain an evaluation of a PRF on input $x$ under a server-held key $\sk$, without revealing $x$ to the server and without learning $\sk$. 
OPRFs consist of four algorithms $(\Request, \BlindEval, \allowbreak\Finalize, \eval)$. 
We instantiate OPRFs using the 2-Hash-DH construction~\cite{AC:JarKiaKra14}. Let $\GG$ be a cyclic group of prime order $p$, and let $\hash_p$ and $\hash_\GG$ be hash functions mapping to $\ZZ_p$ and $\GG$, respectively. To form a request for input $x$, the client runs $\Request(x)$ by sampling $r \sample \ZZ_p$ and computing $\request = \hash_\GG(x)^r$. Upon receiving $\request$, the server computes $\response = \BlindEval(\sk, \request) = \request^\sk$. The client then runs $\Finalize(x, r, \response)$ to remove the blinding and obtain $\eval(\sk, x) = \hash_p(x, \response^{1/r}) = \hash_p(x, \hash_\GG(x)^\sk)$. 

\mypar{NIZKs}
A non-interactive argument system (NARG) for a relation $\relation$ in the random oracle model is a tuple $\proofsys = (\proofprove, \proofverify)$, where $\proofsetup(\secpar)$ outputs public parameters, $\proofprove^\hash(\proofstatement,\proofwitness)$ outputs a proof $\pi$ for $(\proofstatement,\proofwitness)\in\relation$, and $\proofverify^\hash(\proofstatement,\pi)$ outputs $1$ or $0$. We require completeness, i.e., valid statements admit accepting proofs, and soundness, i.e., no efficient adversary can produce a proof for a false statement. The system is zero-knowledge if there exists a simulator that can generate proofs for true statements without access to the witness such that no efficient adversary can distinguish simulated proofs from real ones. A NARG with statistical soundness is called a non-interactive zero-knowledge proof (NIZK). 

\mypar{Communication model.}
We consider a synchronous communication model in which parties interact in rounds, and messages sent in a round are delivered within that same round. We assume static corruptions, i.e., the adversary chooses the corrupted party before the protocol execution and cannot change it thereafter. This model is standard in the context of atomic swap protocols~\cite{SP:ThyMalMor22}.

\mypar{Blockchain.}
We follow the blockchain model of~\cite{SP:ThyMalMor22} and assume an ideal ledger functionality $\blockchainfunc$ that we model as a trusted append-only bulletin board maintaining the balances associated with each address. 
Parties can register addresses and transfer coins via a $\mathsf{Post}(A, B, v)$ interface, which moves $v$ coins from address $A$ to $B$, and can query the current state of the ledger via a $\mathsf{Read}$ interface. 
Following~\cite{AC:AEEFHM21,SP:ThyMalMor22}, the ideal blockchain functionality accepts any posted transaction without enforcing authorization by a specific identity. In a real-world instantiation, however, transactions must be authorized via digital signatures. 
We address this gap by having the simulator forward $\mathsf{Post}$ requests only upon receiving a valid signature on the transaction, and by including these signatures in the adversary’s view upon $\mathsf{Read}$ queries.

\mypar{Timelocks.}
We use absolute timelocks as specified in BIP~65 via \texttt{OP\_CHECKLOCKTIMEVERIFY} (CLTV) using two spending branches. One branch allows spending by $\pkT$ at any time. The other branch allows spending by $\pk$, provided that time $\timet$ has been reached.
For convenience, we denote this script by
$
\lockscript(\pkT, \timet, \pk),
$
which expands to the locking script shown in \ifnum\fullversion=0 Figure~\ref{fig:lockingscripts} in~\cref{sec:appdx:figures}. \else \cref{fig:lockingscripts}.\fi
Instead of using timelocks, our protocol could also be realized using verifiable timed payments as shown in~\cite{SP:ThyMalMor22} (c.f.,~\cref{sec:discussion}). 
\ifnum\fullversion=1
\input{bip65fig.tex}
\fi%

%% file: bip65fig.tex
\begin{figure}	%
\begin{textdef}
$\texttt{OP\_IF } \pkT\ \texttt{OP\_CHECKSIG}$\\
$\texttt{OP\_ELSE}$\\
\strut \quad $\timet\ \texttt{OP\_CHECKLOCKTIMEVERIFY OP\_DROP}$\\
\strut \quad $\pk\ \texttt{OP\_CHECKSIG}$\\
$\texttt{OP\_ENDIF}$
\end{textdef}
	\caption{Transaction output using BIP~65. This output can be spent by $\pkT$ at any time, and by $\pk$ after time $\timet$.}
	\label{fig:lockingscripts}
\end{figure}%

%% file: overview.tex
\section{Solution Overview} \label{sec:overview}
Before presenting our solution for probabilistic swap protocols, let us first recall our setting. 
We consider two entities, a dealer and a party, who wish to engage in a coin exchange with asymmetric and probabilistic outcomes.
In contrast to deterministic atomic swaps, the dealer receives a fixed, typically smaller, payment of $1$ coin, while the party receives a higher reward, say $\buyout$ coins, only with some probability $\winningProb \in [0,1]$. 
The outcome is determined by a probabilistic event agreed upon in advance by both parties and cannot be influenced by either party.

\mypar{Deterministic atomic swaps.}
We begin by recalling the deterministic atomic swaps on scriptless blockchains enabled by adaptor signatures from~\cite{SP:ThyMalMor22}. 
On a high level, in a deterministic atomic swap, the dealer samples a statement-witness pair $(Y, y)$ of a hard relation and sends the statement $Y$ to the party together with an adaptor signature \textit{pre-signature} $\presig_{\dealer\player}$ with respect to $Y$. 
This pre-signature corresponds to a transaction transferring coins from the dealer to the party. 
Upon receiving the pre-signature, the party constructs a different transaction transferring coins from the party to the dealer, produces a corresponding pre-signature $\presig_{\player\dealer}$ with respect to the same statement $Y$, and sends it to the dealer. 
If the dealer decides to execute the swap, it uses its knowledge of a valid witness $y$ to adapt the party's pre-signature $\presig_{\player\dealer}$ into a valid signature $\sigma_{\player\dealer}$ (using the interface from adaptor signatures), broadcasts the party-to-dealer payment transaction on-chain, and can then go offline.
The party eventually observes the adapted signature on-chain and extracts the witness $y$ using the extractability property of adaptor signatures. 
Knowing $y$, the party can, in turn, adapt the pre-signature $\presig_{\dealer\player}$ received from the dealer on the dealer-to-payment transaction and claim its payment. 

This exchange is atomic for the dealer, since it reveals the witness $y$ only by completing the transaction that pays it. 
If the dealer does not proceed, the party cannot obtain $y$ and therefore cannot claim the dealer's coins. 
It is also atomic for the party, since whenever the dealer completes the transaction and receives payment, the party can extract $y$ from the on-chain signature and use it to claim its own payment. 
We note that such a protocol requires temporarily locking coins to prevent frontrunning attacks. 
For simplicity of exposition, we omit these locking mechanisms in this overview.

\subsection{\ifnum\fullversion=1 A \fi Strawman Approach}
On a high level, deterministic atomic swaps have a similar setting to the probabilistic version we seek. 
Therefore, as a natural starting point, we try to modify that paradigm to obtain probabilistic behavior. 
A first idea is to make \textit{adapt} probabilistic: upon learning a valid witness $y$, the party can adapt its pre-signature only with probability $\winningProb$. The parties would then exchange pre-signatures exactly as in the deterministic protocol. The dealer would still obtain payment first, while the party would only be able to complete its side of the swap, and hence receive payment, with probability $\winningProb$. 
Similar in flavor, but using another interface of adaptor signatures, one could also think about making witness extraction probabilistic. 

However, neither of these approaches achieves the desired notion of probabilistic exchange. The core problem is that once the dealer receives the party's pre-signature, the dealer essentially has the same information as the party for evaluating whether the party will ultimately succeed. In particular, the dealer can run the same adapt or extraction procedure locally and determine whether the party can claim the reward. Thus, the dealer can predict the protocol outcome before the exchange is completed and abort whenever the outcome is unfavorable. For example, in a wager-style execution, the dealer could learn in advance that the party is going to win and simply refuse to proceed.

This is not merely an artifact of a particular adaptor signature construction. Rather, it reflects a more basic obstacle: in a probabilistic swap, the outcome must remain hidden from both parties until they are already committed to the exchange. Otherwise, the party that learns the outcome early can selectively abort, thereby biasing the protocol's effective outcome. 
We return to this broader issue in~\cref{sec:extensions}, where we also discuss the difficulty of realizing probabilistic swaps in which both parties receive probabilistic outcomes on minimal blockchains. 

\subsection{Towards the Solution} 
Instead of making adapt or extraction itself probabilistic, we encode the uncertainty into some \textit{hidden information} chosen by both the dealer and the party. 
Concretely, the dealer samples a ``target value'' from a sufficiently large domain, and the party succeeds only if it ``guesses'' this value correctly. 
Crucially, the party's guess must remain hidden from the dealer (similar to the dealer's choice from the party) until the dealer completes its part of the exchange irrevocably. 
In this way, neither party can predict the outcome in advance: the party learns whether it can complete the swap only after paying the dealer, while the dealer cannot tell ahead of time whether the party's guess was correct. 
By appropriately choosing the domain size for the dealer's guessing space, we can ensure that the party succeeds with probability exactly $\winningProb$.\footnote{With a single guess, the party has a winning probability of $1/m$, where $m$ is the domain size. However, we can modify our protocol, such that it supports $a$ concurrent guesses, resulting in a probability of $a/m$ (c.f.,~\cref{sec:discussion}).} 

\mypar{A First attempt.}
Our first attempt to realize this idea is based on oblivious transfer~\cite{10.1145/3812.3818}. 
In an oblivious transfer (OT) protocol, a sender holds multiple messages, and a receiver obtains exactly one of them without revealing its choice to the sender, while learning nothing about the remaining messages. 
A natural approach is to adapt OT to our setting as follows: the dealer (the sender in OT) prepares $m$ statement–witness pairs, pre-signs a pre-signature paying the party $\presig_{\dealer\player}$ w.r.t. exactly one of the statements, while the party (the receiver in OT) uses the OT to obtain the witness corresponding to exactly one statement. 
If the party's choice is correct, it can adapt the dealer's pre-signature and win. 
By the security of OT, the party's choice remains oblivious to the dealer, so it cannot anticipate the output, which was one half of our goal. 

However, the other half of our goal, namely that the party cannot anticipate the protocol outcome prior to the exchange, cannot be covered with the OT-based approach. 
This is the case since in OT, the party knows which element it selects. 
Consequently, since the party can identify in advance which of $m$ statements corresponds to the valid adaptor signature (simply by running pre-verification with all dealer-provided statements and the check passes for exactly one statement), it can simply choose the correct witness with certainty, eliminating the intended probabilistic behavior. 
We could fix this gap by letting the dealer provide a uniformly random witness to the party, but this would open the door to an attack in which the dealer deliberately always sends the wrong witness to the party without being detected. 

This mismatch reflects a more fundamental tension. On the one hand, the dealer must be prevented from cheating, meaning it should be guaranteed that the revealed witness corresponds to the party’s actual choice. 
On the other hand, the party must also be prevented from cheating, which requires that it cannot determine in advance which witness is the ``correct'' one to learn. 
These two requirements are inherently at odds in OT: ensuring verifiability of the chosen message tends to leak information about correctness, while hiding choice-correctness from the party prevents enforceability of the choice.
Because of this limitation, we move away from OT and instead look for other cryptographic primitives. 

\mypar{OPRFs for oblivious guessing.}
We seek a primitive, in which the dealer verifiably but obliviously evaluates the party's guess, and at the same time, the party is oblivious about the correctness of its guess until after the interaction with the dealer. 
Oblivious pseudorandom functions (OPRFs)~\cite{TCC:FIPR05} provide exactly this type of asymmetry. 
In an OPRF protocol, a server holding a secret key $\sk$ interacts with a party holding an input $x$. At the end of the protocol, the party learns the value $\oprf.\eval(\sk,x)$, while the server learns nothing about $x$ or the corresponding function value. Moreover, the party obtains the function value on only one input per execution.
We use the OPRF input space as the space of possible guesses. The dealer chooses a hidden target value $y_\target$, while the party chooses a guess $y_\guess$. The party then interacts with the dealer, who acts as the OPRF server, to obtain $\oprf.\eval(\sk_\dealer,y_\guess)$. By OPRF security, the dealer learns neither the party's guess nor whether it was correct.

We can connect this guessing mechanism to adaptor signatures if the dealer encodes the target value into the adaptor statement. Concretely, the dealer constructs the statement $\paymentStmt$ so that there exists exactly one value $y_\target$ satisfying
$
\paymentStmt = g^{\oprf.\eval(\sk_\dealer, y_\target)},
$ 
where $g$ is a generator of the group used by the adaptor signature scheme. The dealer then pre-signs the payment to the party with respect to $\paymentStmt$ and sends the resulting pre-signature $\presig_{\dealer\player}$ together with $\paymentStmt$, and a proof of well-formedness of $\paymentStmt$ that we discuss later in this section. 
By the pseudorandomness of \oprf, the statement $\paymentStmt$ leaks no information about $y_\target$, and the only way for the party to check if $y_\guess$ is correct is to evaluate the \oprf on $y_\guess$. 
If $y_\guess = y_\target$, the resulting evaluation $\oprf.\eval(\sk_\dealer,y_\guess)$ is a valid witness corresponding to $\paymentStmt$, and so the party can adapt the pre-signature $\presig_{\dealer\player}$ and claim the payment. If the guess is incorrect, the OPRF output is useless for adapting the pre-signature. Thus, when the target is sampled uniformly from a domain of size $m$, a random guess of the party succeeds with probability $1/m$. 
This captures the probabilistic aspect of the party's payout, and we visualize this guessing strategy in~\cref{fig:oprf-guessing}.

At this point, however, only one part of the problem has been addressed. 
Our OPRF-based approach safeguards the odds for the party, but they do not yet explain how to make the learning of the OPRF evaluation itself atomic with payment, i.e., make the exchange atomic for the dealer. 
If the dealer simply returns the OPRF response first, the party may obtain the reward (if $y_\guess = y_\target$) without paying. If the party pays first, the dealer may abort without delivering the response. Thus, the remaining challenge is to make \emph{OPRF evaluation itself} the object of an atomic exchange.

\begin{figure}
  \centering
  \begin{tikzpicture}[>=Stealth]

    \node[ thick, rounded corners=3pt,
          fill=white,
          minimum width=1.4cm, minimum height=1.0cm, align=center]
      (stmt) at (-3.5, 0) {$\paymentStmt= g^{\paymentWit}$};

    \foreach \i/\xpos in {0/-1.5, 1/0.5, 2/2.5}{
      \node[draw, very thick, rounded corners=3pt,
            fill=white,
            minimum width=1.4cm, minimum height=1.0cm, align=center]
        (slot\i) at (\xpos, 0) {$\mathsf{OPRF}(\sk, \cdot)$};
    }

    \draw[->, thick]
      ($(slot0.south)+(0,-0.5)$) -- (slot0.south)
      node[pos=0, below] {$y_{\mathsf{gss},0}$};
    \draw[->, thick]
      ($(slot1.south)+(0,-0.5)$) -- (slot1.south)
      node[pos=0, below] {$y_{\mathsf{tgt}}$};
    \draw[->, thick]
      ($(slot2.south)+(0,-0.5)$) -- (slot2.south)
      node[pos=0, below] {$y_{\mathsf{gss},2}$};

    \draw[->, thick, red!70!black, dashed]
      (slot0.north) -- ++(0,0.5) node[above] {$y_0$};
    \draw[->, thick, blue!70!black]
      (slot1.north) -- ++(0,0.5) node[above] {$y_{\mathsf{win}}$};
    \draw[->, thick, red!70!black, dashed]
      (slot2.north) -- ++(0,0.5) node[above] {$y_2$};

  \end{tikzpicture}
  \caption{The dealer acts as the OPRF server. The party submits guesses
  $y_{\mathsf{gss},0}, y_{\mathsf{tgt}}, y_{\mathsf{gss},2}$ and obtains
  the corresponding OPRF outputs. Only one guess $y_\target$ corresponds to a valid witness $\paymentWit$; the remaining outputs (red, dashed) are invalid.}
  \label{fig:oprf-guessing}
\end{figure}
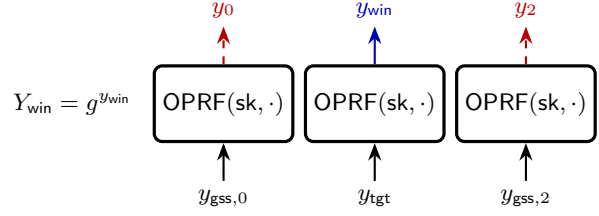

\mypar{The missing piece: OPRF evaluation as a service.}\label{sec:missing_piece}
To atomically exchange OPRF evaluations, our idea is to treat the dealer's OPRF evaluation response as \emph{locked} until the party-to-dealer payment is completed. 
Rather than sending the response to the party directly, the dealer encrypts it under an ephemeral public key $Z$ for which the dealer knows the corresponding decryption key $z$, and proves well-formedness of this ciphertext w.r.t.~the request sent by the party (using a NIZK proof, see~\cref{sec:instances}).
Then, the dealer sends a pre-signature $\presig_{\dealer\player}$, pre-signed under the winning statement $\paymentStmt$, alongside the ephemeral key $Z$, the ciphertext, and the proof of well-formedness. 
Upon receiving this information, the party can verify that the dealer has provided a valid response, which corresponds to evaluating the OPRF on the party's guess $y_\guess$, but cannot yet recover the resulting OPRF value. This simplifies the exchange, since now it is atomic if the party can obtain the decryption key $z$ exactly when the dealer is paid.

This is where adaptor signatures enter again. The party prepares the dealer's payment transaction $\presig_{\player\dealer}$, but using the ephemeral key $Z$ as the adaptor statement and sends the resulting pre-signature to the dealer. Since the dealer knows the ephemeral decryption key (which now also acts as the adaptor witness) $z$, it can adapt the pre-signature into a full signature $\sigma_{\player\dealer}$ and publish the party-to-dealer payment transaction on-chain to claim payment. 
Once this signature is posted, the party can extract the adaptor witness (aka, the ephemeral decryption key) $z$ from the adapted signature. This key $z$ allows the party to decrypt the previously received ciphertext and recover the OPRF evaluation on $y_\guess$.

The protocol ties the two events together: the dealer can obtain payment only by revealing the decryption key, and the party can obtain the OPRF evaluation only after that key is revealed through the on-chain signature. This is the mechanism that upgrades the probabilistic OPRF-based idea into an atomic probabilistic swap. We visualize this OPRF evaluation for a payment in~\cref{fig:overviewPartcrypto}.

\subsection{Combining the Building Blocks} \label{sec:proofsketch}
At a high level, the protocol combines two ingredients: an OPRF-based guessing mechanism that introduces the desired outcome probability for the party, and an atomic OPRF-evaluation service that ensures the dealer is paid.
The dealer samples a target value $y_\target$ from a domain of size $m$, constructs the corresponding payment statement $\paymentStmt$, and sends the party a pre-signature $\presig_{\dealer\player}$ together with the required well-formedness proofs. The party, in turn, selects a guess $y_\guess$ from the same domain and engages in the OPRF evaluation protocol with the dealer. This OPRF evaluation is atomically tied to a payment from the party to the dealer using the encryption-and-adaptor-signature mechanism above.
If $y_\guess = y_\target$, which happens with probability $1/m$, the party obtains the correct OPRF evaluation, derives the witness $\paymentWit$, and successfully adapts the pre-signature $\presig_{\dealer\player}$ to claim the payment. Otherwise, the party learns nothing useful for adaptation, while the dealer is compensated for the OPRF evaluation.
We illustrate the full protocol in~\cref{fig:overviewPartcrypto} and give its formal description in~\cref{sec:proto}.

\mypar{Security intuition.} 
We briefly sketch why the proposed probabilistic swap protocol is secure and defer the formal treatment to~\cref{sec:security}. 
From the party’s perspective, knowledge soundness of the NIZK proofs and correctness of the adaptor signature scheme ensure that the dealer’s pre-signature $\presig_{\dealer\player}$ is well-formed with respect to a target value $y_\target$ from an appropriate domain matching the expected winning probability. 
Hence, if the party’s guess $y_\guess$ matches the dealer’s hidden target $y_\target$, then an OPRF evaluation yields exactly the witness $\paymentWit$, which enables the party to adapt $\presig_{\dealer\player}$ and claim the payment. 
Moreover, by the extractability of adaptor signatures, whenever the dealer completes the payment transaction and obtains compensation, the party can extract the decryption key $z$ from the on-chain signature and thereby recover the OPRF output. Finally, the input privacy of the OPRF guarantees that the dealer learns no information about $y_\guess$, preventing it from biasing the outcome based on the party’s choice.

From the dealer’s perspective, zero-knowledge of the NIZK proofs and the CPA security of the encryption scheme ensure that no information about the target value $y_\target$ or the corresponding OPRF evaluation is revealed prior to payment from the party.  
At the same time, the adaptor signature mechanism guarantees that the ephemeral decryption key is disclosed to the party only if the dealer completes the transaction and receives payment. 
Finally, the pseudorandomness of the OPRF ensures that each evaluation yields a single independent output, so the party learns at most one candidate witness per execution. This implies that the party’s success probability is exactly $\winningProb = 1/m$, as desired. 

\begin{figure}
	\centering

\begin{tikzpicture}[>=Stealth]
  \node[draw,  thick, rounded corners=10pt,
        minimum width=1.5cm, minimum height=3.5cm, align=center]
    (D) at (0,0) {{\Large Dealer}\\ ($\sk_\dealer, y_\target$)};
    
  \draw[ ->, shorten >=5pt]
  (D.south) -- ++(0,-.6)
  node[below] {$\sigma_{\player\dealer}$};

  \node[draw, thick, rounded corners=10pt,
        minimum width=1.5cm, minimum height=3.5cm, align=center]
    (P) at (6,0) {{\Large Party} \\ ($y_\guess$)};

  \draw[ ->, shorten >=5pt]
  (P.south) -- ++(0,-.6) 
  node[below] {$\mathsf{if}\; y_\target = y_\guess: \sigma_{\dealer\player}$};

  \foreach \k/\dir/\lbl in {
    {1/->/$\paymentStmt = g^{\oprf.\eval(\sk_\dealer, y_\target)}$},
    2/<-/$\request \gets \oprf.\Request(y_\guess)$,
    {3/->/$Z,\enc(Z, \response),\presig_{\dealer\player}(\paymentStmt)$},
    4/<-/$\presig_{\player\dealer}(Z)$
  }{
    \pgfmathsetmacro{\t}{\k/5}

    \coordinate (dpt) at ($(D.north east)!\t!(D.south east)$);
    \coordinate (ppt) at ($(P.north west)!\t!(P.south west)$);

    \draw[\dir,shorten <=5pt,shorten >=5pt] (dpt) -- node[above] {\lbl} (ppt);
  }
\end{tikzpicture}
\caption{Protocol overview.
(1)~$\dealer$ sends $\paymentStmt$, binding $y_\target$ under its OPRF key.
(2)~$\player$ sends a blinded OPRF request for its hidden guess $y_\guess$.
(3)~$\dealer$ returns an OPRF response encrypted under ephemeral key $Z$,
and pre-signs $\player$'s reward w.r.t.~$\paymentStmt$.
(4)~$\player$ pre-signs $\dealer$'s payment w.r.t.~$Z$; $\dealer$ posts it
on-chain, revealing $z$. $\player$ extracts $z$, decrypts the response,
and adapts $\presig_{\dealer\player}$ to claim the reward iff $y_\guess = y_\target$.}
\label{fig:overviewPartcrypto}
\end{figure}

\subsection{Proving Well-Formedness of $\paymentStmt$}\label{sec:overview-nizk}
The remaining technical contribution of this work is the proof of well-formedness of the statement $\paymentStmt$ sent by the dealer (the prover) to the party (the verifier). 
This proof must show that this statement admits a discrete logarithm witness that can be computed by evaluating the OPRF on inputs $(\sk_\dealer, y_\target)$, where $y_\target \in [0, m]$. Only if the statement $\paymentStmt$ has this structure, the party can be sure that evaluating the OPRF on its guess $y_\guess$ will result in a probabilistic swap with the desired winning probability $\winningProb = 1/m$. 

Before we talk about this proof of well-formedness, let us first recall the structure of the OPRF we are using, which is the 2HashDH OPRF~\cite{7467360}.
With 2HashDH, the witness $\paymentWit$ has the structure 
$
	\paymentWit = \hash_p\left(y_\target, \hash_\GG(y_\target)^\sk\right).
$ 
Proving well-formedness of this leads to additional challenges, as both $\hash_p$ and $\hash_\GG$ are modeled as random oracles. 
We provide two solutions to this problem.

\mypar{Instantiation with Bulletproofs} 
Our first solution uses a deterministic Groth-style map-to-curve relation for $\hash_\GG$ and a MiMC-based hash-to-scalar relation for $\hash_p$, and shows well-formedness using R1CS in a Bulletproof~\cite{SP:BBBPWM18}.
We show in~\cref{sec:proto} that this instantiation is concretely efficient with proving time below $6$ seconds, verification time below $0.35$s, and proof size of $2.8$KB. However, this approach has the downside of requiring the RO to be instantiated, which is unfavorable~\cite{SP:HLTW24}. 

\mypar{Instantiation from cut-and-choose.}
To avoid instantiating the random oracle in the well-formedness proof, we present a second approach based on the \textit{cut-and-choose} technique~\cite{JC:LinPin12}.
A natural attempt is to have the dealer sample $\lambda$ (statistical security parameter) target values $y_{\target,1}, \ldots, y_{\target,\lambda}$, compute the corresponding statements $Y_{\win,i}$, and open a random half upon challenge.
The party verifies that each opened value lies in the correct domain and that the OPRF component is well-formed with respect to the OPRF public key $\pk = g^\sk$.
Cut-and-choose then guarantees that, except with negligible probability, at least one unopened instance is well-formed. 
As we argue next, this property is not strong enough to achieve fairness. 

We observe first that the dealer must sample distinct target values $y_{\target,i}$, as opening an instance reveals information about the corresponding witness.  
If the values are not distinct, the party could recover a valid witness without guessing if the same target value is used in both an opened and an unopened instance. 
As a result, only $m - \lambda/2$ target values remain as possible winning candidates.  
In addition, since the party cannot decide which of the unopened statements corresponds to a valid OPRF execution, the dealer must provide a pre-signature for every unopened statement $Y_{\win,j}$. 

However, this breaks fairness: the party now holds $\lambda/2$ pre-signatures, each corresponding to a different target value, and can attempt to guess a valid witness for any of them. 
Since a dishonest dealer needs to make only one of the $\lambda/2$ unopened instances well-formed, the winning probability for an honest party is guaranteed to be only $\winningProb = 1/(m-\lambda /2 )$. 
However, if the dealer behaves honestly, a possibly malicious party has a winning probability of $\winningProb = (\lambda/2)/(m - \lambda/2)$, a factor-$\secpar/2$ gap ($\lambda=128$ in our case). 
This violates the core security requirement of probabilistic swaps that the winning probability $p = 1/m$ holds unconditionally, regardless of the dealer/party's behavior. 

We resolve this issue by modifying the used OPRF. 
The difficulty is that in the 2HashDH construction, opening any instance requires revealing $y_\target$ in the clear (as it is the first part of the input to $H_p$), which would immediately leak the target to the party.
We aim to allow opening instances \emph{without} revealing $y_\target$. 
Along this way, we first use the intermediate
$
 \oprf(\sk, y_\target) =: \hash_p\!\left(g^\sk,\, \hash_\GG(y_\target)^{\sk} \right),
$ 
whose only modification compared to 2HashDH is the changed first input, moving from $y_\target$ to the OPRF key $g^\sk$. 
This intermediate OPRF allows the dealer to open the outer hash $\hash_p$ without revealing $y_\target$ as the first input. 
However, the value $\hash_\GG(y_\target)^{\sk}$ is still deterministic w.r.t.~$y_\target$. Hence, we make an additional modification, which is replacing 
$\hash_\GG(y_\target)^{\sk}$ by $\hash_\GG(y_\target)^{\sk\cdot \alpha_i}$, where the values $\{g^{\alpha_i}\}_{1 \leq i \leq \secpar}$ are additional (ephemeral) public keys for a single probabilistic swap instance. 
These two modifications turn out to be sufficient for our goal: 
For all opened instances, the prover publishes $\hash_\GG(y_\target)^{\sk\cdot \alpha_i}$, and a proof of well-formedness of this value with respect to the OPRF key $g^\sk$, the ephemeral key $g^{\alpha_i}$, and the public set $\{\hash_\GG(y) : y \in [0, \ldots, m]\}$ using a Chaum--Pedersen OR proof, without revealing $y_\target$. 
For unopened instances, the dealer reveals $\alpha_j$, allowing the party to locally compute and verify the candidate witness after obtaining the OPRF blind evaluation $\hash_\GG(y_\target)^{\sk}$. 

The resulting OPRF inherits the pseudorandomness and request privacy of the baseline 2HashDH construction, as shown in~\cref{sec:appdx:prelims}. Since the ephemeral public keys are uniformly random, the prover can reuse the same $y_\target$ across all cut-and-choose instances. This proof strategy yields the correct winning probability of $\winningProb = 1/m$. 
We show in~\cref{sec:benchmarks} that this cut-and-choose-based approach not only avoids instantiating the random oracle, but is also more efficient than the Bulletproof instantiation for winning probabilities $p \geq 1/2^{12}$.
See~\cref{fig:proofs:cutandchoose,fig:proofs:orSchnorr,fig:proofs:orWF,fig:proofs:hiddenbaseWF} for a formal treatment.

%% file: fair_lottery.tex
\newcommand{\swapfunc}{\ensuremath{\mathcal{F}_\mathsf{Swap}^{\blockchain}}}

\section{Defining Probabilistic Swaps}
\label{sec:idealfunc}
\label{sec:functionality}
We formalize the security requirements of probabilistic swap protocols. 
The central property we aim to capture is a notion we call \emph{probabilistic atomicity}, which can be seen as a probabilistic analog of atomic swaps~\cite{SP:ThyMalMor22}. 
%
%
We begin by recalling the functionality of deterministic atomic swaps~\cite{SP:ThyMalMor22}, which is a simplified sub-case of probabilistic swaps (by simply setting the winning probability $\winningProb=1$). 
Then, we use this baseline to explain the modifications required to obtain our probabilistic functionality.

\mypar{Deterministic atomic swaps.}
The deterministic atomic functionality interacts with two parties over multiple rounds.
In the first two rounds, each party commits to the swap by submitting their matching public inputs to the functionality. 
Once both parties have submitted their values and agreed to run the swap, the functionality uses their inputs to lock the respective funds under a temporary functionality-controlled address. 
After the funds are frozen, each party retains the ability to abort the protocol.
If neither party aborts, the functionality atomically executes both transfers, paying each party. 
If any party aborts, the functionality releases all locked funds back to their respective owners.
This mechanism ensures \emph{deterministic} atomicity: either both transfers occur on-chain, or all funds are returned to their original owners.

\mypar{Our functionality.}
Our functionality follows the overall structure of the deterministic case with respect to setup and funding. 
Both parties agree on public  parameters, and the functionality locks their respective funds under a temporary, functionality-controlled address.
After funding, the dealer decides whether to participate by issuing a claim transaction. 
At this point, the dealer does not yet know the protocol's outcome; hence, the dealer needs to make this decision \emph{independently} of the protocol's outcome. 
Once the dealer claims its payment, the functionality samples the probabilistic outcome and reveals it to the party.
If the party receives a winning outcome (i.e., $b=1$), it can claim the dealer's locked funds. 
Otherwise, the dealer can later reclaim its funds. 
If either party remains inactive, the respective other party can recover their respective funds after predefined timeouts.
To retain the correct output probability for the swap, the protocol's outcome is revealed only after the dealer has committed to execution by claiming its payment. 
The functionality has two atomic outcomes: 
\begin{itemize}[leftmargin=*,nosep]
    \item The dealer is paid. The counterparty then can claim the locked funds with probability $\winningProb$. 
    \item The dealer is not paid. In this case, no exchange occurs, and both parties remain unaware of the probabilistic outcome. 
\end{itemize}
In both cases, unclaimed funds can be refunded by the respective parties after an a priori fixed period.
We provide a formal definition of our functionality in~\cref{fig:idealfunc:lottery}.
\begin{figure}
\centering

\begin{textdef}
Functionality $\lotteryfunc$  interacts with a dealer $\dealer$, and a party $\player$. 

\textbf{Setup:} Upon receiving $(\buyout, \winningProb, \pk_\dealer, \pk_\player, \timet_\dealer, \timet_\player)$ from $\dealer$, store the inputs and proceed to \texttt{Freeze}. 
	
\textbf{Fund:} Upon receiving $(\buyout, \winningProb, \pk_\dealer, \pk_\player, \timet_\dealer, \timet_\player)$ from $\player$, verify that $\player$'s inputs match those provided by $\dealer$. 
	Call $\freeze(\pk_\player, 1, \timet_\player)$ and $\freeze(\pk_\dealer, \buyout, \timet_\dealer)$ and set $b=0$. 
	If the check and freezing succeed, send $\mathtt{funded}$ to the dealer.
	Otherwise, send $\mathtt{fail}$ to the dealer and the player.
	In both cases, store the inputs and $b$, and proceed to \texttt{ClaimDealer}. 
	
\textbf{Claim Dealer:} Upon receiving $\play$ from $\dealer$ before $\timet_\player$ is up, set $b=1$ with probability $\winningProb$, call the subroutine $\transfer(1, \pk_\dealer)$, and send $b$ to $\player$. Continue with \texttt{ClaimParty}.

\textbf{Claim Party:} Upon receiving $\play$ from $\player$ before $\timet_\dealer$ is up, if $b=1$, call the subroutine $\transfer(\buyout, \pk_\player)$ to pay $\player$.				   

\textbf{Refund:} Upon receiving $\abort$ from $\dealer$ after $\timet_\dealer$ is up ($\player$, after $\timet_\player$), call the subroutine $\mathsf{Unfreeze}(\buyout, \pk_\dealer)$ ($\mathsf{Unfreeze}(1, \pk_\player)$).
					   
\vspace{0.5em}
\hrule
\vspace{0.5em}
Description of the subroutines. The subroutines are successful if the transaction is accepted by $\blockchainfunc$.

\textbf{Freeze:}
On input a tuple $(\pk, \nu, \timet)$, transfer $\nu$ coins from the address specified by $\pk$ to a script $\lockscript(\pkT, \timet, \pk)$ via $\texttt{Post}(\pk, \lockscript, \nu)$, where $\pkT$ is a temporary address. 

\textbf{Transfer:} On input a tuple $(\nu, \pk)$ transfer $\nu$ frozen coins to the address specified by $\pk$, via $\texttt{Post}(\pkT, \pk, \nu)$. 

\textbf{Unfreeze:}
On input $(\nu, \pk)$, move $\nu$ coins from $\lockscript$ to $\pk$, making them spendable only by $\pk$ via $\texttt{Post}(\pk, \pk, \nu)$.
\end{textdef}

\caption{The ideal functionality $\lotteryfunc$.}

\label{fig:idealfunc:lottery}
\end{figure}

\mypar{Other key differences with the functionality of~\cite{SP:ThyMalMor22}.}
While our functionality follows the high-level structure of~\cite{SP:ThyMalMor22}, we introduce several modifications that facilitate practical realizations. 

\lightpar{Refunds.}
The functionality of~\cite{SP:ThyMalMor22} assumes that funds can be atomically unfrozen. 
In most protocols, however, refunding requires active participation of the party controlling the locked coins (e.g., by issuing a transaction). 
In particular, an adversary in the security experiment could act irrationally by refusing to post a refund transaction, thereby breaking the simulation. 
To capture this more faithfully, our functionality models refunds explicitly: each party must trigger its own refund separately. 

\lightpar{Timing-based locking.}
The functionality of~\cite{SP:ThyMalMor22} does not model time, and funds remain locked until they are explicitly released. 
In contrast, our functionality is parameterized by timeouts $\timet_\dealer$ and $\timet_\player$, after which the dealer and the party, respectively, can reclaim their funds.
The main motivation for this change is that practical atomic swap protocols are typically realized using timelock scripts rather than verifiable timed signatures (VTS) as in~\cite{SP:ThyMalMor22}. 
Explicitly modeling time avoids mismatches between ideal and real protocol executions, since in the real execution, a locking script would be the recipient of the freezing transaction, whereas in the simulated world, a temporary address would be. 
We emphasize that this is primarily a modeling choice: our functionality could equivalently be instantiated using a timeless approach based on VTS. 

\lightpar{Public inputs only.}
In contrast to~\cite{SP:ThyMalMor22}, our functionality does not require parties to provide secret inputs, but only public parameters of the exchange. 
This significantly simplifies simulation, as the simulator no longer needs to extract secret keys from the adversary, which would otherwise require additional proofs of knowledge.
At first glance, this may appear at odds with the fact that blockchain transactions require signatures.
However, both our blockchain functionality and that of~\cite{SP:ThyMalMor22} model the blockchain as an ordered list of transactions, abstracting away transaction validation.
Since the ideal blockchains accept all incoming transactions, the proofs simulator that sits between the adversary and the ideal blockchain must ensure that only authorized transactions are sent to the blockchain. 
We do this by only appending signed transactions to the blockchain. 

\lightpar{Handling of signatures.}
Both our protocol and that of~\cite{SP:ThyMalMor22} rely on signatures in concrete realizations. 
However, signatures are not explicitly modeled by the blockchain functionality. To resolve this issue, we incorporate them at the simulation level.
Concretely, the simulator only appends transactions to the blockchain functionality if the corresponding party provides a valid signature. 
Moreover, when the adversary observes the blockchain, the simulator augments the output with the corresponding signatures before forwarding it to the adversary.
Not handling signatures like this would be an issue in the proof of~\cite{SP:ThyMalMor22} as well, so we conclude that the authors assume this is done implicitly. 

%% file: protocol.tex
\section{Our Protocol} \label{sec:proto}\label{sec:security}
We present our probabilistic swap protocol between a dealer $\dealer$ and a party $\player$. 
As building blocks, we use adaptor signatures and OPRFs, which we recalled in~\cref{sec:prelims}.
Additionally, we rely on secure two-party computations for the adaptor signature key pairs and pre-signing. 
For ease of exposition, we focus on Schnorr pre-signatures with DLog keypairs. However, the respective sub-protocols in~\cref{fig:proto:jointDKG,fig:proto:jointpresign} could simply be replaced by secure two-party protocols for any adaptor signature scheme. 

\mypar{2-party DLog key generation.}
We use a standard two-party protocol for generating a shared discrete logarithm keypair, following~\cite{EC:GJKR99}. 
The protocol ensures that both parties obtain additive secret shares of a jointly generated public key, while proving knowledge of their respective shares.
We depict the corresponding ideal functionality in \ifnum\fullversion=0\cref{fig:func:jointDKG} in~\cref{sec:appdx:figures}. \else \cref{fig:func:jointDKG}.
\input{2pdlogfuncfig.tex}
\fi  
It allows the adversary to fix its secret share while the honest party’s share is sampled uniformly at random. 
In particular, this enables us to learn the adversarial share of the secret key while determining the joint public key, which will be crucial in the security proof. 
We formalize the standard protocol realizing this two-party key generation functionality in~\cref{fig:proto:jointDKG}. 
\begin{figure}
	\begin{textdef}
		Global input: $(\GG, g, p)$ and a hash function $\hash: \bin^* \to \ZZ_p$. Parties have no individual inputs. 
		
		\begin{enumerate}[leftmargin=*]
			\item Party $P_0$ samples a secret key $\sk_0 \sample \ZZ_p$, the corresponding public key $\pk_0 = g^{\sk_0}$, commits to $\pk_0$ via $\com_0 = \hash(\pk_0)$, and sends $\com_0$ to party $P_1$.
			\item Party $P_1$ samples a secret key $\sk_1 \sample \ZZ_p$, the corresponding public key $\pk_1 = g^{\sk_1}$ and proves knowledge of $\sk_1$ via $\pi_1 \gets \proofsys.\proofprove(\pk_1, \sk_1)$. It sends $( \pk_1, \pi_1)$ to $P_0$.
			\item Party $P_0$ receives $( \pk_1, \pi_1)$. If the proof verifies, i.e., $\proofsys.\proofverify(\pk_1, \pi_1) =1$, party $P_0$ also proves knowledge of $\sk_0$ via $\pi_0 \gets \proofsys.\proofprove(\pk_0, \sk_0)$ and sends $(\pk_0, \pi_0)$ to $P_1$. Otherwise, it aborts.
			\item Party $P_1$ asserts, that $pk_0$ is a valid opening for $\com_0$  by checking if $\hash(\pk_0)=\com_0$ and verifies the proof by checking if $\proofsys.\proofverify(\pk_0, \pi_0) =1$. It aborts otherwise. 
		\end{enumerate}
		Party $P_0$ sets its output to $(\pk_0 \cdot \pk_1, \sk_0)$. Party $P_1$ sets its output to $(\pk_0 \cdot \pk_1, \sk_1)$. 
	\end{textdef}

\caption{Two-Party protocol $\Gamma_\mathsf{DL}$ for generating a DLog keypair. C.f. ~\cref{fig:proofs:schnorrDL} for the corresponding proof system.}
\label{fig:proto:jointDKG}
\end{figure}

\mypar{2-party Schnorr pre-signing.}
We use a two-party variant of Schnorr pre-signing as introduced in~\cite{AC:AEEFHM21}, which allows two parties holding additive secret shares to jointly produce a valid Schnorr pre-signature. 
On public input, a public key $\pk$, a transaction $\tx$, and a statement $Y$, and secret inputs being the key shares of two parties corresponding to $\pk$, the protocol outputs a valid pre-signature.   
We formally recall the protocol in~\cref{fig:proto:jointpresign}.
\begin{figure}
	\begin{textdef}
		Global input: $(\pk, \tx, \statement)$. Party $P_0$'s input: $\sk_0$. $P_1$'s input: $\sk_1$. 
		
		\begin{enumerate}[leftmargin=*]
			\item Both parties run the two-party protocol from~\cref{fig:proto:jointDKG} to obtain a joint value $R$, while $P_0$ receives $r_0$ and $P_1$ receives $r_1$, such that $R = g^{r_0 + r_1}$.
			\item Party $P_0$ computes $s_0 = \sk_0\cdot c + r_0$ and sends $s_0$ to $P_1$, where $c = \hash(\pk, R \cdot Y, \tx)$.
			\item Party $P_1$ receives $s_0$ from $P_0$ and computes $s_1 = \sk_1 \cdot c + r_1$, where $c = \hash(\pk, R \cdot Y, \tx)$, and sends $s_1$ to $P_0$.
		\end{enumerate}
		Both parties set their output to $(R\cdot Y, s_0 + s_1)$. $P_1$ receives the output first. 
	\end{textdef}
\caption{Two-Party Schnorr pre-signing protocol $\Gamma_\pSign$.}
	\label{fig:proto:jointpresign}
\end{figure}

\mypar{Protocol description.}
Using the above building blocks, we can describe our probabilistic atomic swap protocol.
The protocol reads as public input a buyout amount $\buyout$, a winning probability $\winningProb = 1/m$, the dealer's public key $\pk_\dealer$, the counterparty's public key $\pk_\player$, and two timeout parameters $\timet_\dealer$ and $\timet_\player$. 
The protocol's output is a reallocation of the dealer's and player's funds, as depicted in~\cref{fig:overviewNocrypto}. 
We require that $\timet_\dealer > \timet_\player$, ensuring that once the dealer has issued its claim, the party has sufficient time $\timet_\dealer - \timet_\player$ to submit its own claim. 
Since our construction relies on timelock scripts on the blockchain rather than verifiable timed signatures (VTS) as in~\cite{SP:ThyMalMor22}, the timeout parameters need not be chosen overly conservatively, which provides the dealer with faster liquidity if the party loses. 
We provide our protocol in~\cref{fig:const}, which realizes probabilistic swaps for probabilities of the form $1/m$. 
However, this baseline protocol can be extended to support arbitrary rational probabilities in $\mathbb{Q}$, as discussed in~\cref{sec:discussion}.
We omit a high-level overview of our protocol here as we provided one in~\cref{sec:overview}.
\begin{figure*}
	\begin{textdef}
	
	Global input: $(\buyout, \winningProb, m, \pk_\dealer, \pk_\player, \timet_\dealer, \timet_\player)$, dealer's input: $\sk_\dealer$, player's input: $\sk_\player$.
	It holds that $\winningProb = 1/m$.
	
			\textbf{Setup:} The dealer verifiably encrypts the winning witness under the (small domain) target witness.
			\begin{enumerate}
				\item Sample the target witness $\witnessTarget \sample [0, m]$, and compute $\paymentWit \gets \oprf.\eval(\sk_\dealer, \witnessTarget)$, $\paymentStmt = g^{\paymentWit}$. 
                
				\item Prove well-formedness of $\paymentStmt$ via 
               $\nizkProof \gets \proofsys.\proofprove({\relation_\win},(\pk_\dealer, \paymentStmt, m),(\sk_\dealer, \witnessTarget))$ and send $(\nizkProof, \paymentStmt)$ to the party. We define the corresponding relation language $\relation_\win$ in~\cref{sec:instances}. 
			\end{enumerate}

			\textbf{Funding Phase:} The dealer and the party lock the coins they want to use in the protocol. 
			\begin{enumerate}
				\item They run the key generation protocol of~\cref{fig:proto:jointDKG}. Joint output: $\pkT$. Private output for dealer: $\skTD$, for the party: $\skTP$.
				\item Both party and dealer create locking transactions $\tx_\dealer = (\pk_\dealer, \buyout, \lockscript(\pkT, \timet_\dealer, \pk_\dealer))$ and $\tx_\player = (\pk_\player, 1, \lockscript(\pkT, \timet_\player, \pk_\player))$ that spend $\buyout$ (one) coins from the address corresponding to $\pk_\dealer$ ($\pk_\player$) and whose outputs are the locking scripts $\Lambda$ (c.f.,~\cref{fig:lockingscripts}). 
				\item The dealer spends these locking transactions by computing $\sigma_\dealer \gets \Sigma.\sign(\sk_\dealer, \tx_\dealer)$ and sending $(\tx_\dealer, \sigma_\dealer)$ to the blockchain $\blockchain$. The party also spends the locking transactions with $(\tx_\player, \sigma_\player)$, where 
				$\sigma_\player \gets \Sigma.\sign(\sk_\player, \tx_\player)$. 
			\end{enumerate}
			When funding is complete, the dealer's funds can be spent solely by $\pkT$ for $\timet_\dealer$ timeout period. Afterwards, both $\pkT$ and $\pk_\dealer$ can spend it. The same holds for the party's funds with $\timet_\player$ and $\pk_\player$. 

			\textbf{Claim Dealer:} After the swap is successfully funded, the dealer and party exchange information until the dealer is paid.
			\begin{enumerate}
				\item The party samples $\witnessGuess \sample [0, m]$, computes an OPRF request $(\st, \request ) \gets \oprf.\Request(\witnessGuess)$, and sends $\request$ to the dealer.
				\item The dealer evaluates the OPRF request $\response \gets \oprf.\BlindEv(\sk, \request)$, samples an encryption key $(Z, z)\gets \kgen(\secpar)$, and verifiably encrypts $\response$ under $Z$ via $\ciphertext_\dealer \gets \Pi_\enc.\enc(Z, \response)$, $\nizkProof_\dealer \gets \proofsys.\proofprove(\relation_\enc, (\pk_\dealer, Z, \request, \ciphertext_\dealer)(\sk_\dealer, z))$. 
				We define the corresponding relation language $\relation_\enc$ in~\cref{sec:instances}. 
				Finally, the dealer sends $(Z, \ciphertext_\dealer, \nizkProof_\dealer)$ to the party. 
				\item Dealer and party jointly pre-sign the transactions $\tx_{\player\dealer} = (\pkT, 1, \pk_\dealer)$ which spends one coin from $\pkT$ to $\pk_\dealer$ and $\tx_{\dealer\player} = (\pkT, \buyout, \pk_\player)$ that spends $\buyout$ coins from $\pkT$ to $\pk_\player$ by running the pre-signing protocol of~\cref{fig:proto:jointpresign} twice. 
				\begin{itemize}
					\item Joint input: $(\tx_{\dealer\player}, \paymentStmt)$. Private input: $\dealer$ inputs $\skTD$, party inputs $\skTP$. The protocol outputs $\presig_{\dealer\player}$ first to $\dealer$ then to $\player$. 
					\item Joint input: $(\tx_{\player\dealer}, Z)$. Private input: $\dealer$ inputs $\skTD$, party inputs $\skTP$. The protocol outputs $\presig_{\player\dealer}$ first to $\player$ then to $\dealer$. 
				\end{itemize}
				\item The dealer adapts $\presig_{\player\dealer}$ using $z$ via $\sigma_{\player\dealer} \gets \adaptorScheme.\adapt(\pk_\player, \presig_{\player\dealer}, z)$ and posts $(\tx_{\player\dealer}, \sigma_{\player\dealer})$ to the blockchain\ifnum\fullversion=0 $\blockchain$\fi
				. 
			\end{enumerate}

			\textbf{Claim Party:} The party claims the $\buyout$ coins if its guess was correct.
			\begin{enumerate}
				\item The party extracts $z$ from $\sigma_{\player\dealer}$ obtained from $\blockchain$ via $z \gets \adaptorScheme.\extract(Z, \presig_{\player\dealer}, \sigma_{\player\dealer})$.
				\item It decrypts the OPRF response via $\response \gets \Pi_\enc.\dec(z, \ciphertext)$ and computes the guessed witness via $\paymentWit' \gets \oprf.\finalize(\response, \st)$. 
				\item If the guess was correct, the party obtains $\sigma_{\dealer\player} \gets \adaptorScheme.\adapt(\pkT, \presig_{\dealer\player}, \paymentWit')$, and posts $(\tx_{\dealer\player}, \sigma_{\dealer\player})$ to the blockchain $\blockchain$. 
			\end{enumerate}

            \textbf{Timeout Phase:} Parties reclaim their frozen funds if the protocol fails or the other party did not claim their coins.
            \begin{enumerate}
            	\item If the dealer fails to post $(\tx_{\player\dealer}, \sigma_{\player\dealer})$ to $\blockchain$ within time $\timet_\player$, the party creates a refund transaction $\tx_{\refund\player} = (\pk_\player, 1, \pk_\player)$ that spends the party's unspent outputs from $\tx_{\player\dealer}$, signs it via $\sigma_{\refund\player} \gets \Sigma.\sign(\sk_\player, \tx_{\refund\player})$ and sends $(\tx_{\refund\player}, \sigma_{\refund\player})$ to the blockchain $\blockchain$.
            	\item Similarly, if the party fails to post $(\tx_{\dealer\player}, \sigma_{\dealer\player})$ to $\blockchain$ within time $\timet_\dealer$, the dealer creates a refund transaction $\tx_{\refund\dealer} = (\pk_\dealer, \buyout, \pk_\dealer)$, signs it via $\sigma_{\refund\dealer} \gets \Sigma.\sign(\sk_\dealer, \tx_{\refund\dealer})$ and sends $(\tx_{\refund\dealer}, \sigma_{\refund\dealer})$ to the blockchain $\blockchain$. 
            \end{enumerate}
	\end{textdef}

\caption{Our probabilistic atomic swap protocol from adaptor signatures and oblivious pseudorandom functions for exchange probabilities in the powers of two.
}
\label{fig:const}
\end{figure*}
\begin{theorem} \label{thm:main}
Let $\adaptorScheme$ be an extractable and adaptable adaptor signature scheme w.r.t. a strongly unforgeable signature scheme $\Sigma$, and a hard relation $\relation$. 
Let $\Gamma_\mathsf{KG}$ and $\Gamma_\pSign$ be secure 2PC pre-signing and key generation protocols for $\adaptorScheme$. 
Let $\Pi_\enc$ be a CPA-secure encryption scheme with key pairs matching the relation of $\adaptorScheme$, and $\proofsys$ a knowledge sound zero-knowledge proof system. 
Let $\oprf$ be a pseudorandom and request-private two-hash OPRF mapping in the witness space of~$\relation$. 
Then, for all $\buyout \in \mathbb{Q}^{+}, m \in \NN$, and $\winningProb = 1/m$, the probabilistic atomic swap protocol described in~\cref{fig:const} with access to $\blockchainfunc$ securely realizes $\lotteryfunc$, i.e., for 
every PPT adversary $\adv$ there exists a PPT simulator $\simulator$ such that
\fullinline{
\idealview_{\lotteryfunc, \simulator} \approx_c
\realview_{\Pi, \adv}.}
\end{theorem} 
\ifnum\fullversion=0
We defer a formal proof to~\cref{sec:appdx:proofs} and refer to~\cref{sec:proofsketch} for an intuitive proof description. 
\else
We refer to~\cref{sec:proofsketch} for an intuitive proof description. 
\fi
The central component of our protocol, which we referred to as \emph{OPRF evaluation as a service} in~\cref{sec:overview}, is implemented in the claim-dealer phase. There, the dealer acts as an OPRF server: it blindly evaluates the OPRF on the party's guess, verifiably encrypts the response, and uses an adaptor pre-signature to tie the release of the decryption key atomically to its own payment. This treatment of OPRF evaluation as an atomically purchasable service is what upgrades the probabilistic guessing mechanism into a probabilistic atomic swap.

\ifnum\fullversion=1
\input{appdx_proof.tex}
\fi

%% file: 2pdlogfuncfig.tex
\begin{figure}
	\begin{textdef}
	Functionality $\mathcal{F}_\mathsf{DLog}$ interacts with an honest party $P$ and an adversary $\adv$. 
	\begin{enumerate}[leftmargin=*]
		\item Receive $\sk_\adv$ from the adversary.
		\item Choose $\sk \sample \ZZ_p$ and set $\pk = g^{\sk}$.
		\item Send $\pk$ to $\adv$. If $\adv$ sends $\texttt{ok}$, send $(\pk, \sk - \sk_\adv)$ to $P$.
	\end{enumerate}
	\end{textdef}
	\caption{Two-party DLog keypair generation functionality.}
	\label{fig:func:jointDKG}

\end{figure}

%% file: appdx_proof.tex
\ifnum\fullversion=0
\section{Proofs} \label{sec:appdx:proofs}
\fi
\begin{proof}[Proof of~\cref{thm:main}]
To prove~\cref{thm:main}, we show that for every PPT adversary $\adv$ corrupting either the dealer or the party, there exists an efficient simulator $\simulator$ that simulates the real-world protocol execution with $\adv$ while interacting with the ideal functionality $\lotteryfunc$.
We consider static corruptions, where the environment $\environment$ specifies the corrupted party at the beginning of the session. The adversary controls the corrupted party, while the simulator is given the corrupted party's internal state, in particular its secret key. The simulator faithfully impersonates the honest party.
As we carry out the proof in the $\smtfunc$-hybrid world, communication reveals only that it occurred, without leaking any information about message contents. Therefore, we omit modelling communication in the simulator.
We proceed in two steps, constructing one simulator for the case of a corrupted dealer and one for the case of a corrupted party.

We begin with the case where the adversary corrupts the party and construct a simulator $\simulator_\dealer$.
We analyze this case via a sequence of hybrid executions, starting from the real-world execution and gradually modifying it until we obtain the simulator’s execution. We then argue that any two consecutive hybrids are indistinguishable. For simplicity, we describe a single session, noting that the argument extends to multiple sessions in a straightforward manner.
	
\initialgame The first experiment is the honest execution of our protocol depicted in~\cref{fig:const}. \label{gamehop:simdealer:honest}
	
\gamehop In \thisgame, we replace honestly computed zero-knowledge proofs by their simulated counterparts using the simulator $\simulator_\nizk$. By the zero-knowledge property of $\proofsys$, there exists such a simulator. 
$\simulator_\nizk$ reads a statement $\proofstatement$ and outputs a simulated proof $\pi$. 
\label{gamehop:simdealer:nizk}
	
\gamehop In \thisgame, we replace the key generation protocol of~\cref{fig:proto:jointDKG} using the 2PC simulator $\simulator_\dkg$. Such a simulator exists for this protocol, as our theorem assumes it is a secure 2PC protocol. 
$\simulator_\dkg$ reads a key $\pkT \in \GG$ and outputs $\skTP$ to $\simulator_\dealer$. 
\label{gamehop:simdealer:jointDKG}

\gamehop In \thisgame, we replace the joint adaptor signing protocols of~\cref{fig:proto:jointpresign} using the 2PC simulator $\simulator_\pSign$. Such a simulator exists for this protocol, as our theorem assumes it is a secure 2PC protocol. 
$\simulator_\pSign$ reads a public key $\pkT \in \GG$, a secret key $\skTP$, and a pre-signature $\presig$ and has no output to $\simulator_\dealer$. 
$\simulator_\pSign$ ensures, that the pre-signature computed in the 2PC equals $\presig$.
\label{gamehop:simdealer:jointpresign}

\gamehop In \thisgame, we replace all OPRF-related operations from the dealer using the simulator $\simulator_\oprf$. Such a simulator exists, as our theorem assumes that the OPRF is pseudorandom. 
$\simulator_\oprf$ reads public parameters, outputs a public key $\pk$, and provides access to $\eval, \BlindEval, \Prim$ oracles. 
\label{gamehop:simdealer:oprf}

\gamehop In \thisgame, the simulator aborts with message $\abort_0$ if the adversary attempts to post a tuple $(\tx^*, \sigma^*)$ to $\blockchain$, in which $\sigma^*$ is valid w.r.t. $\pk_\dealer$ on $\tx^*$, and $\tx^* \notin \{\tx_{\dealer}, \tx_{\refund\dealer}\}$. 
\label{gamehop:simdealer:abort0}

\gamehop In \thisgame, the simulator aborts with message $\abort_1$ if the adversary attempts to post a tuple $(\tx^*, \sigma^*)$ to $\blockchain$, in which $\sigma^*$ is valid w.r.t. $\pkT$ on $\tx^*$, and $\sigma^*$ does not allow extracting a valid witness for $\paymentStmt$ from $\presig_{\dealer\player}$, i.e., $(\paymentStmt, \adaptorScheme.\extract(\pkT, \presig_{\dealer\player}, \sigma^*,\paymentStmt))\notin \relation$. If our protocol did not yet define $\presig_{\dealer\player}$, we set it to $\bot$. 

\label{gamehop:simdealer:abort1}

\gamehop In \thisgame, the simulator aborts with message $\abort_2$ if the adversary attempts to post a tuple $(\tx^*, \sigma^*)$ to $\blockchain$, in which $\sigma^*$ is valid w.r.t. $\pkT$ on $\tx^*$, \emph{before} the simulator forwards $\sigma_{\player\dealer}$ to $\adv$. \label{gamehop:simdealer:aborttwo} 

\gamehop In \thisgame, at the beginning of the experiment, the simulator uses the two-hash structure of the OPRF, and whenever the adversary queries the $\Prim$ oracle for the first hash function $H_1$ on some input $\hat \witnessGuess \in [2^\ell]$, the simulator computes $y = F(\sk, H_1(\hat \witnessGuess)$, and stores $(\hat \witnessGuess, y)$. 
We denote these tuples as \emph{valid OPRF preimages}. 
The simulator aborts with message $\abort_3$ if an input to the $\Prim$ oracle is a valid $\oprf$ preimage, but the simulator did not forward $\sigma_{\player\dealer}$ to $\adv$, or if this output would be the second valid OPRF preimage to the $\Prim$ oracle to the adversary. \label{gamehop:simdealer:abortsinglequery} 

\gamehop In \thisgame, the simulator changes the way it simulates $\Prim$ to $\adv$. When the adversary queries $\Prim$, on input $z$, $\simulator_\dealer$ checks if this is a valid OPRF preimage (defined in~\cref{gamehop:simdealer:abortsinglequery}). 
If this is the case \emph{after} the functionality outputs $b=1$, the simulator changes the output of $\Prim$ to $\Prim(z) = \paymentWit$. \label{gamehop:simdealer:correctprobb}
\label{gamehop:simdealer:last}

\noindent\textbf{Simulator $\simulator_\dealer$.} The execution of the simulator is defined as the execution in ~\cref{gamehop:simdealer:last}. 
The simulator $\simulator_\dealer$ reads as input $(\buyout, \winningProb, \ell, \timet_\dealer, \timet_\player)$ and starts by internally running $\simulator_\oprf$ to receive the OPRF-part $\pk_{\dealer, \oprf}$ of $\pk_\dealer$. 
Second, it runs $(\pk_{\dealer, \Sigma}, \sk_{\dealer, \Sigma})\gets \kgen(\secpar)$ to obtain the signature part of $\pk_\dealer$. 
Then, $\simulator_\dealer$ internally executes $\adv$ on input $(\buyout, \winningProb, \ell, \timet_\dealer, \timet_\player, \pk_\dealer)$. 
Eventually, $\adv$ outputs $\pk_\player$.

\noindent\underline{Setup:} $\simulator_\dealer$ samples $(\paymentStmt, \paymentWit) \sample \genRelation(\secpar)$, and simulates $\pi$ using $\simulator_\nizk$ on input $\proofstatement = (\pk_\dealer, \paymentStmt, \ell)$.
The simulator sends $(\paymentStmt, \pi)$ to $\adv$, and $(\buyout, \winningProb, \pk_\dealer, \pk_\player, \timet_\dealer, \timet_\player)$ to $\lotteryfunc$. 

\noindent\underline{Funding:} $\adv$ eventually attempts to run the DKG protocol. To run this protocol, $\simulator_\dealer$ first samples a keypair $(\pkT, \skT) \gets \keygen(\secpar)$ and interacts with $\adv$ using $\simulator_\dkg$ on input $\pkT$. Eventually, $\adv$ receives back $\skTP$. 
After running this $\dkg$, the adversary eventually attempts to post $(\tx_\player, \sigma_\player)$ on-chain. 
The simulator computes a signature $\sigma_\dealer$ and provides the pair $(\tx_\dealer, \sigma_\dealer)$ to $\adv$. 
In addition, the simulator invokes the \texttt{Freeze} interface of $\lotteryfunc$. 

\noindent\underline{Claim Dealer:}
Eventually, $\adv$ sends a request $\request$ to the simulator. 
$\simulator_\dealer$ then computes the response $\response \gets \simulator_\oprf.\BlindEv(\request)$, samples $(Z, z)\gets \genRelation(\secpar)$, and $\ciphertext_\dealer \gets \Pi_\enc.\enc(Z, \response)$.
Finally, $\simulator_\dealer$ simulates $\nizkProof_\dealer$ using $\simulator_\nizk$ on input $\proofstatement = (\pk_\dealer, Z, \request, \ciphertext_\dealer)$ and forwards $(Z, \ciphertext_\dealer, \nizkProof_\dealer)$ to $\adv$. 
Then, to run both pre-signing protocols, $\simulator_\dealer$ locally pre-signs each $\tx_{\player\dealer}, \tx_{\dealer\player}$ using $\skT$, resulting in $\presig_{\player\dealer}$, $\presig_{\dealer\player}$, and uses $\simulator_\pSign$ on input $(\skTP, \presig_{\player\dealer})$ and $(\skTP, \presig_{\dealer\player})$, to enforce these pre-signatures as outputs of the pre-signing protocols. 
Finally, $\simulator_\dealer$ adapts $\presig_{\player\dealer}$ using $z$, and outputs $\sigma_{\player\dealer}$ to $\adv$. 
In addition, $\simulator_\dealer$ calls the \texttt{ClaimDealer} interface of $\lotteryfunc$, which responds with a bit $b$. 

\noindent\underline{Claim Party:} 
Eventually, $\adv$ queries the $\Prim$ interface on a valid OPRF preimage. 
If this is the case, and $b=1$, $\simulator_\dealer$ returns $\paymentWit$ to $\adv$. If $\adv$ attempts to send $(\tx_{\dealer\player}, \sigma_{\dealer\player})$ to the blockchain, $\simulator_\dealer$ calls the \texttt{ClaimParty} interface of \lotteryfunc. 

\noindent\underline{Timeout:}
If funding was successful, but $\adv$ did not participate in the protocol until $\timet_\dealer$ is up, $\simulator_\dealer$ signs $\tx_{\refund\dealer}$, sends $\abort$ to $\lotteryfunc$, and forwards $(\tx_{\refund\dealer}, \sigma_{\refund\dealer})$ to $\adv$. 

On top of this simulation, $\simulator_\dealer$ aborts if the adversary forges signatures under $\pk_\dealer$ (c.f.,~\cref{gamehop:simdealer:abort0}) or $\pkT$ (c.f.,~\cref{gamehop:simdealer:abort1,gamehop:simdealer:aborttwo}), or obtains unfaithful \oprf evaluations (c.f.,~\cref{gamehop:simdealer:abortsinglequery}).

We next show that the differences between two neighboring experiments are computationally indistinguishable. 
Since we have only polynomially many experiments, this establishes the first part of our theorem, since by this negligible transition, we have 
\[
	\idealview_{\lotteryfunc, \simulator_\dealer} \approx \realview_{\Pi_\lottery^\blockchain, \adv_\player}.
\]
\compclose{gamehop:simdealer:honest}{gamehop:simdealer:nizk} 
	These experiments differ only in the way how zero-knowledge proofs are computed by $\simulator_\dealer$. 
	By the zero-knowledge property of $\proofsys$, indistinguishability follows canonically. 

\compclose{gamehop:simdealer:nizk}{gamehop:simdealer:jointDKG} These experiments differ only in how the two-party DKG protocol is executed. In~\cref{gamehop:simdealer:nizk}, $\simulator_\dealer$ runs the honest DKG protocol of~\cref{fig:proto:jointDKG} with $\adv$. In~\cref{gamehop:simdealer:jointDKG}, it is replaced by the 2PC simulator $\simulator_\dkg$. Since~\cref{thm:main} assumes $\Gamma_\mathsf{KG}$ is a secure two-party protocol, such a simulator exists, and the adversary's view is computationally indistinguishable across the two experiments.

\compclose{gamehop:simdealer:jointDKG}{gamehop:simdealer:jointpresign} These experiments differ only in how the two joint pre-signing protocols of~\cref{fig:proto:jointpresign} are executed. In~\cref{gamehop:simdealer:jointDKG}, $\simulator_\dealer$ runs the honest pre-signing protocol. In~\cref{gamehop:simdealer:jointpresign}, it is replaced by the 2PC simulator $\simulator_\pSign$. Since~\cref{thm:main} assumes $\Gamma_\pSign$ is a secure two-party protocol, the adversary's view is computationally indistinguishable.

\compclose{gamehop:simdealer:jointpresign}{gamehop:simdealer:oprf} These experiments differ in the way how \oprf queries are handled by $\simulator_\dealer$. If an efficient adversary could distinguish this change, it would by definition break the pseudorandomness of $\oprf$. 

\compclose{gamehop:simdealer:oprf}{gamehop:simdealer:abort0} Indistinguishability between these experiments follows canonically from the unforgeability of $\Sigma$. 

\compclose{gamehop:simdealer:abort0}{gamehop:simdealer:abort1} We show the computational indistinguishability between \cref{gamehop:simdealer:abort0} and \cref{gamehop:simdealer:abort1} by building a reduction $\bdv$ that uses an efficient distinguisher of these experiments and breaks the extractability of $\adaptorScheme$. 
As input, $\bdv$ receives a public key $\pk$, and it sets $\pkT = \pk$. In addition, $\bdv$ has access to a pre-signing oracle and a signing oracle that output (pre-)signatures valid on $\pkT$ on messages (and statements) of $\bdv$'s choice. 
Whenever $\simdealer$ needs to pre-sign, it forwards the respective message-statement pair to the pre-sign oracle. During simulation, $\simdealer$ never needs to sign messages using $\skT$. 
With these modifications, $\bdv$ simulates~\cref{gamehop:simdealer:oprf} to $\adv$. 
By assumption, there exists an efficient distinguisher between the experiments. 
As the experiments only differ, if the adversary attempts to post a tuple $(\tx^*, \sigma^*)$ to $\blockchain$, in which $\sigma^*$ is valid w.r.t. $\pkT$ on $\tx^*$, and $\sigma^*$ does not allow extracting a valid witness for $\paymentStmt$ from $\presig_{\dealer\player}$, i.e., $(\paymentStmt, \adaptorScheme.\extract(\pkT, \presig_{\dealer\player}, \sigma^*,\paymentStmt))\notin \relation$, this must be the case. 
Otherwise, the experiments would be perfectly indistinguishable. 
If the adversary outputs such a forgery, $\bdv$ outputs it as well. By definition, this breaks the extractability of $\adaptorScheme$. 
In addition, the simulation of $\bdv$ is perfect, as we replace pre-signing with perfect oracles and efficient, as \cref{gamehop:simdealer:oprf} is efficient. 

\compclose{gamehop:simdealer:abort1}{gamehop:simdealer:aborttwo} 
We demonstrate the computational closeness between these two experiments using a sub-experiment. 
This subexperiment is idential to \cref{gamehop:simdealer:abort1}, but instead of encrypting the response $\response$ in $\ciphertext_\dealer$, $\simdealer$ encrypts the generator $g$. 
We show, that before $\simdealer$ forwards $\sigma_{\dealer\player}$ to $\adv$, \cref{gamehop:simdealer:abort1} and this subexperiment are computationally indistinguishable. 
We do this by building a reduction $\bdv$ from the distinguishability between these two experiments before $\simdealer$ sends $\sigma_{\dealer\player}$, and in which $\adv$ posts such a forgery to $\adv$ to the CPA security of $\Pi_{\enc}$. In other words, by contradiction, we assume the existence of an efficient distinguisher $\ddv$. 
As input $\bdv$ receives a public key $\pk$, and sets $Z = \pk$. 
Since we only consider the experiment before $\simdealer$ sends the signature, it is sufficient for $\bdv$ to not know $z$. 
Then, $\bdv$ simulates \cref{gamehop:simdealer:abort1}. Instead of encrypting the response $\response$, it forwards the pair $(\response, g)$ to its CPA challenger, and receives back a ciphertext $\ciphertext$. 
$\bdv$ continues simulating \cref{gamehop:simdealer:abort1} with $\ciphertext_\dealer = \ciphertext$ and outputs the transcript of the experiment to $\ddv$. 
Eventually, $\ddv$ outputs a bit and $\bdv$ outputs the same bit. 
If the CPA challenger encrypts $\response$, $\bdv$ simulated \cref{gamehop:simdealer:abort1} perfectly to $\ddv$. In the other case, it simulated our subexperiment. 
Therefore, $\ddv$'s advantage in distinguishing the experiments carries over to $\bdv$'s ability to break CPA.

Next, we show that this subexperiment and \cref{gamehop:simdealer:aborttwo} are computationally indistinguishable by the hardness of the relation $\relation$. 
This holds, since in the situation of \cref{gamehop:simdealer:abort1}, the adversary can only compute signature forgeries that allow it to extract a valid witness for $\paymentStmt$. 
In the situation of our subexperiment, the adversary learns no OPRF evaluation by running the protocol. 
Therefore, we can compute a reduction $\bdv$ that receives $\paymentStmt$, simulates our subexperiment to $\adv$ (note that knowledge of \paymentWit is not required to do so, since the abort happens), and if the adversary triggers the additional abort in \cref{gamehop:simdealer:aborttwo}, $\bdv$ learns a valid witness for $\paymentStmt$ and breaks the hardness of $\relation$. 
This reduction is perfect, since a uniformly random $\paymentStmt$ is equally distributed to the way our protocol computes $\paymentStmt$, as in the situation of \cref{gamehop:simdealer:oprf}, the OPRF is a random function. 
Combining these two steps, the closeness between \cref{gamehop:simdealer:abort1,gamehop:simdealer:aborttwo} holds. 

\compclose{gamehop:simdealer:aborttwo}{gamehop:simdealer:abortsinglequery}
Both cases in which \cref{gamehop:simdealer:abortsinglequery} aborts can only happen if the adversary obtains a valid OPRF preimage (and hence a correct \oprf evaluation) without querying the $\BlindEv$ oracle. 
By the pseudorandomness of \oprf, the probability that this happens is negligible. 

\compclose{gamehop:simdealer:abortsinglequery}{gamehop:simdealer:correctprobb}
The oracle $\Prim$ simulates an idealized primitive, and hence, provides uniform outputs by the pseudorandomness of \oprf. 
Therefore, as long as $\Prim$ is not queried on an input, the output cannot be anticipated. 
In addition, our simulator sampled $\paymentWit$ uniformly at random. Therefore, this transition is computationally indistinguishable.

	Next, we consider the setting in which the adversary corrupts the dealer. We again proceed in a series of hybrid executions.

	\initialgame The first experiment is the honest execution of our protocol depicted in~\cref{fig:const}. \label{gamehop:simplayer:honest}
	
	\gamehop In \thisgame, we replace the key generation protocol of~\cref{fig:proto:jointDKG} using the 2PC simulator $\simulator_\dkg$. Such a simulator exists for this protocol, as our theorem assumes it is a secure 2PC protocol. 
		$\simulator_\dkg$ reads a key $\pkT \in \GG$ and outputs $\skTP$ to $\simulator_\dealer$. 
		\label{gamehop:simplayer:jointDKG}

	\gamehop In \thisgame, we replace the joint adaptor signing protocols of~\cref{fig:proto:jointpresign} using the 2PC simulator $\simulator_\pSign$. Such a simulator exists for this protocol, as our theorem assumes it is a secure 2PC protocol. 
		$\simulator_\pSign$ reads a public key $\pkT \in \GG$, a secret key $\skTP$, and a pre-signature $\presig$ and has no output to $\simulator_\dealer$. 
		$\simulator_\pSign$ ensures, that the pre-signature computed in the 2PC equals $\presig$.
		\label{gamehop:simplayer:jointpresign}
	
	\gamehop In \thisgame, the simulator computes the OPRF outputs of all possible guesses $\hat \witnessGuess \in [2^\ell]$. If none of these outputs matches $\paymentWit$, the simulator aborts with the message $\abort_0$.
			 The simulator stores the winning guess $\witnessTarget$ and $\paymentWit$, which are found if the simulator does not abort. \label{gamehop:simplayer:extractwit}
	
	\gamehop In \thisgame, the simulator replaces $\oprf.\Request(\witnessGuess)$ in the claim dealer phase by $\oprf.\Request(\witnessTarget)$. However, only if the functionality returns $b=1$, it adapts $\presig_{\dealer\player}$ with $\paymentWit$. \label{gamehop:simplayer:oprfpriv}
	
	\gamehop In \thisgame, the simulator aborts with message $\abort_1$, if extraction of $z$ fails in the funding phase, i.e., if the adversary posts some signature $\sigma^*$ which is valid under $\pkT$, but it holds that $(Z, \adaptorScheme.\extract(\pkT, \presig_{\player\dealer}, \allowbreak \sigma^*,Z))\notin \relation$. \label{gamehop:simplayer:extractadaptor}
	
	\gamehop In \thisgame, the simulator aborts with message $\abort_2$, if in the claim phase after decrypting $\ciphertext$, the result is not well-formed, i.e.,  $\response \neq \request^{\sk_\dealer}$. \label{gamehop:simplayer:soundnessct}
	
	\gamehop In \thisgame, the simulator aborts with message $\abort_3$, if adapting the pre-signature $\presig_{\dealer\player}$ with $\paymentWit$ fails, i.e., it holds that $\Sigma.\vrfy(\pkT, \tx_{\dealer\player}, \adaptorScheme.\adapt(Y, \presig_{\dealer\player}, \paymentWit))=0$. \label{gamehop:simplayer:adaptability}
	
	\gamehop In \thisgame, the simulator aborts with message $\abort_4$ if the adversary attempts to post a tuple $(\tx^*, \sigma^*)$ to $\blockchain$, in which $\sigma^*$ is valid w.r.t. $\pk_\player$ on $\tx^*$, and $\tx^* \notin \{\tx_{\player}, \tx_{\refund\player}\}$. 
		\label{gamehop:simplayer:abortforgesigma}\label{gamehop:simplayer:last}

\noindent\textbf{Simulator $\simulator_\player$.} The execution of the simulator is defined as the execution in ~\cref{gamehop:simplayer:last}. 
The simulator $\simulator_\player$ reads as input $(\buyout, \winningProb, \ell, \timet_\dealer, \timet_\player)$ and starts by computing a signing keypair $(\pk_{\player}, \sk_{\player})\gets \kgen(\secpar)$. 
Then, $\simulator_\player$ internally executes $\adv$ on input $(\buyout, \winningProb, \ell, \timet_\dealer, \timet_\player, \pk_\player)$. 
Eventually, $\adv$ outputs $(\pk_\dealer, \paymentStmt, \pi)$. 

\noindent\underline{Setup:} $\simulator_\player$ sends $(\buyout, \winningProb, \pk_\dealer, \pk_\player, \timet_\dealer, \timet_\player)$ to $\lotteryfunc$. 

\noindent\underline{Funding:}, $\adv$ eventually attempts to run the DKG protocol. To run this protocol, $\simulator_\player$ first samples a keypair $(\pkT, \skT) \gets \keygen(\secpar)$ and interacts with $\adv$ using $\simulator_\dkg$ on input $\pkT$. Eventually, $\adv$ receives back $\skTD$. 
After running this $\dkg$, the adversary eventually attempts to post $(\tx_\dealer, \sigma_\dealer)$ on-chain. 
The simulator computes a signature $\sigma_\player$ and provides the pair $(\tx_\player, \sigma_\player)$ to $\adv$. 
In addition, the simulator invokes the \texttt{Freeze} interface of $\lotteryfunc$. 

\noindent\underline{Claim Dealer:}
The simulator extracts $(\witnessTarget', \sk_\dealer')$ using the proof system's extractor. 
If this extraction does not yield a valid secret key for $\pk_\dealer$, or $\paymentStmt \neq g^{\oprf.\eval(\sk_\dealer', \witnessTarget')}$, the simulator aborts with the message $\abort_0$.
The simulator stores $\paymentWit$, which is found if the simulator does not abort. 
Then, the simulator computes the request $\request \gets \oprf.\Request(\witnessTarget)$. and sends $\request$ to $\adv$, which eventually responds with the tuple $(Z, \ciphertext_\dealer, \nizkProof_\dealer)$ to $\adv$. 
To run both pre-signing protocols, $\simulator_\player$ locally pre-signs each $\tx_{\player\dealer}, \tx_{\dealer\player}$ using $\skT$, resulting in $\presig_{\player\dealer}$, $\presig_{\dealer\player}$, and uses $\simulator_\pSign$ on input $(\skTP, \presig_{\player\dealer})$ and $(\skTP, \presig_{\dealer\player})$, to enforce these pre-signatures as outputs of the pre-signing protocols. 
Eventually, $\adv$ posts $(\tx_{\player\dealer}, \sigma_{\player\dealer})$ to $\blockchain$. 
If this happens, the simulator calls the \texttt{ClaimDealer} interface of $\lotteryfunc$, which responds with a bit $b$. 

\noindent\underline{Claim Party:} 
If $b=1$, the simulator adapts $\presig_{\dealer\player}$ with $\paymentWit$ into $\sigma_{\dealer\player}$ and calls the \texttt{ClaimParty} interface of \lotteryfunc. 

\noindent\underline{Timeout:}
If funding was successful, but $\adv$ did not participate in the protocol until $\timet_\player$ is up, $\simulator_\player$ signs $\tx_{\refund\player}$, sends $\abort$ to $\lotteryfunc$, and forwards $(\tx_{\refund\player}, \sigma_{\refund\player})$ to $\adv$. 

On top of this simulation, $\simulator_\player$ aborts if the adversary forges signatures under $\pk_\player$ (c.f.,~\cref{gamehop:simplayer:abortforgesigma}) or $\pkT$ (c.f.,~\cref{gamehop:simplayer:extractadaptor}), or computes the encrypted OPRF response unfaithfully (c.f.,~\cref{gamehop:simplayer:soundnessct}), or adapting the adaptor signature fails (c.f.,~\cref{gamehop:simplayer:adaptability}).

We next show that the differences between two neighboring experiments are computationally indistinguishable. 
Since we have only polynomially many experiments, this establishes the first part of our theorem, since by this negligible transition, we have 
\[
	\idealview_{\lotteryfunc, \simulator_\player} \approx \realview_{\Pi_\lottery^\blockchain, \adv_\dealer}.
\]

\compclose{gamehop:simplayer:honest}{gamehop:simplayer:jointDKG} This transition is analogous to~\cref{gamehop:simdealer:nizk}$\approx$\cref{gamehop:simdealer:jointDKG} in the simulation of the dealer: the honest DKG protocol is replaced by the 2PC simulator $\simulator_\dkg$, and indistinguishability follows from the security of $\Gamma_\mathsf{KG}$.

\compclose{gamehop:simplayer:jointDKG}{gamehop:simplayer:jointpresign} This transition is analogous to~\cref{gamehop:simdealer:jointDKG}$\approx$\cref{gamehop:simdealer:jointpresign} in the simulation of the dealer: the honest pre-signing protocol is replaced by the 2PC simulator $\simulator_\pSign$, and indistinguishability follows from the security of $\Gamma_\pSign$.

\compclose{gamehop:simplayer:jointpresign}{gamehop:simplayer:extractwit} This transition is negligible by the knowledge soundness of the proof system $\proofsys$. The experiments differ only when witness extraction fails, violating knowledge soundness.  

\compclose{gamehop:simplayer:extractwit}{gamehop:simplayer:oprfpriv}
The indistinguishability between these experiments follows from the request privacy of $\oprf$, and we build a reduction to show this. 
The challenger for request privacy provides two interfaces $\Request$ and $\Finalize$ to our simulator. 
Instead of computing the request as in \cref{gamehop:simplayer:extractwit}, the simulator calls $\Request(\pk_\dealer, \witnessGuess, \witnessTarget)$, and receives back two requests. 
It forwards the first request to $\adv$ during the claim dealer phase. 
The rest of~\cref{gamehop:simplayer:extractwit} can be simulated perfectly by the simulator. 
Eventually, the adversary finishes, and our reduction provides the transcript to the distinguisher, which outputs a bit $b'$. 
Our reduction outputs the same bit $b'$ as the guess against the request privacy property of $\oprf$. 
If the request privacy challenger set $b=0$, then our reduction sent the request for $\witnessGuess$, and hence simulated~\cref{gamehop:simplayer:extractwit} perfectly. 
In the other case, if $b=1$, it simulated~\cref{gamehop:simplayer:oprfpriv} perfectly. 
Hence, any ability to distinguish these experiments leads to our reduction breaking the request privacy of \oprf. 
Therefore, the gap between the experiments is negligible.

\compclose{gamehop:simplayer:oprfpriv}{gamehop:simplayer:extractadaptor}
Indistinguishability between these experiments follows from the extractability of $\adaptorScheme$: If there is a gap between the experiments, then the adversary outputs a signature valid under $\pkT$ which does not allow extracting with $\presig_{\player\dealer}$. In this situation, we can build a reduction that receives and provides $\pkT$ as input. The simulator computes all pre-signatures using the provided pre-sign oracle, and simulates the experiment as in~\cref{gamehop:simplayer:extractadaptor}. When an abort occurs, the simulator outputs the forged signature. 
By assumption of the abort, this signature forgery breaks extractability, while our reduction is efficient. Therefore, this abort happens only with negligible probability. 

\compclose{gamehop:simplayer:extractadaptor}{gamehop:simplayer:soundnessct}
Indistinguishability between these experiments follows from the soundness of the proof system $\proofsys$: If the decryption of the simulator does not equal the well-formed \oprf evaluation, the adversary managed to compute a valid proof for an invalid statement, and the simulator can use this proof to break soundness. 

\compclose{gamehop:simplayer:soundnessct}{gamehop:simplayer:adaptability}
Indistinguishability between these experiments follows canonically from the adaptability of $\adaptorScheme$. 

\compclose{gamehop:simplayer:adaptability}{gamehop:simplayer:abortforgesigma}
Indistinguishability between these experiments follows canonically from the unforgeability of $\Sigma$. 

\end{proof}

%% file: extensions.tex
\section{Extensions and Limitations} \label{sec:discussion} \label{sec:extensions}

\mypar{Cross-chain swaps.} 
Our proposed probabilistic swap definition and protocols rely solely on the signature verification and locking capabilities of the underlying blockchain, similar to their deterministic counterparts~\cite{SP:ThyMalMor22}. 
Therefore, our protocol is also applicable to settings in which the dealer uses assets on a blockchain different from the party's. 
The only difference is that both parties lock the funds in two joint accounts, one on each blockchain. The remainder of the protocol remains unchanged. In case the two chains employ signature schemes across different curves (e.g., Bitcoin and Cardano), we would require additional cross group equality proofs for discrete log statements which are standard in the literature.

\mypar{Any winning probability.} Our protocol currently only supports a winning probability of $1/m$ for $m \in \NN$. 
This is because the party can evaluate a single OPRF per query, and we used a range proof that shows $y_\target \leq m$. 
We can slightly modify our protocol to enable any winning probability $a/b \in \mathbb{Q}$. 
For the denominator, we simply set $m=b$. 
To modify the numerator from $1$ to $a$, the dealer can allow the party to simultaneously send $a$ requests, which the dealer evaluates and encrypts under the same $z$. 
This allows the party to make $a$ guesses per exchange, enhancing the winning probability to $a/b$, while keeping the on-chain footprint unchanged. 

\mypar{Post-quantum security.} 
Our current instantiation is built on Schnorr signatures and hence relies on the hardness of the discrete logarithm problem. 
We made this choice to be compatible with Bitcoin and related blockchains at the time of writing. 
Nevertheless, we conjecture that our protocol can also be instantiated using post-quantum (PQ) secure building blocks. 
In particular, there exist PQ secure 2Hash-based OPRFs~\cite{EC:BDFH25} as well as post-quantum adaptor signature schemes~\cite{FC:TaiMorMaf21,ESORICS:EsgErsErk20}. 
Corresponding CPA-secure public-key encryption schemes, compatible in the sense that adaptor signatures can be used to exchange decryption keys, include Regev’s LWE-based encryption scheme~\cite{10.1145/1568318.1568324} and isogeny-based PKE~\cite{EC:BasMai25}, respectively. 
Moreover, PQ proof systems exist that can verify the well-formedness of ciphertexts and OPRF evaluations~\cite{C:GLSTW23}. 


\mypar{Avoiding on-chain timelocks.}
One on-chain difference between our probabilistic swap protocol and the deterministic counterpart of~\cite{SP:ThyMalMor22} is that, in addition to requiring on-chain signature verification, our protocol also requires on-chain coin locking. 
We added this to our ideal functionality, since most deployed swap protocols rely on this kind of coin-locking. 
However, we can also use verifiable timed signatures (VTS) from~\cite{SP:ThyMalMor22} to realize this without using on-chain scripts. 
This changes the on-chain footprint of our protocol to just ordinary transactions without any scripting. 
The security of our protocol still holds, but~\cref{thm:main} would require the VTS security on top. 

\mypar{Deployment on the Lightning network.}
Our protocol can also be slightly modified to run on PTLC-based layer2 networks, such as the Lightning Network, once PTLCs become active. 
To do this, the dealer and party run the funding phase only once, locking the funds into the payment channel. 
Then, they can run the protocol several times by just executing the claim dealer and claim party protocol subparts. 
Running these subparts reveals the adaptor witnesses needed to settle the corresponding channel updates. 
Once enough probabilistic exchanges have happened, both parties can close the channel. 
This results in an on-chain footprint identical to the footprint of an ordinary payment channel. 
While PTLCs are not yet deployed for the Lightning Network, they are actively under development~\cite{bitcoinops-ptlc}, so we expect them to be available in the next few years. 
For today's HTLC-based Lightning interface, we can bind the claim dealer and claim party protocols to the release of hash-preimage, instead of adaptor signatures using off-chain link proofs. We demonstrate the feasibility of this in a variant of our prototype as described in~\cref{sec:benchmarks}.

\mypar{Difficulty of two-sided probabilistic swaps.}
We provide intuition why protocols in which both parties obtain probabilistic outputs cannot be realized using minimal blockchain functionality, following the impossibility result of~\cite{EPRINT:BGKS25}.
We model the parties as interactive machines that sequentially execute a probabilistic swap protocol.
At the start of the protocol, neither party knows whether it will receive its payout; otherwise, a party facing an unfavorable outcome would not participate.
At the end of the protocol, however, both parties learn their outcomes.
Therefore, for each party, there must exist a \emph{first} message after which it can determine whether it will receive its payout; following~\cite{EPRINT:BGKS25}, we call such a message \emph{deciding}.
Consider the deciding message for party $j$.
Upon receiving it, party $j$ learns its partial outcome from the protocol and, if it is unfavorable, party $j$ has no incentive to continue and can simply abort.
Crucially, in any sequential protocol on a scriptless blockchain, the deciding messages for the two parties cannot arrive simultaneously~\cite{EPRINT:BGKS25}. 
One party must receive its deciding message strictly before the other. 
At that point, if the outcome is unfavorable, that party can abort before the other party receives its own deciding message. 
This creates an inherent asymmetry: the aborting party can condition its behavior on information that the other party does not yet have.
This argument shows that, in the absence of additional trust assumptions or stronger on-chain functionality such as a randomness beacon~\cite{CCS:CGJKM17,NDSS:KapGreMie19} or enforceable execution~\cite{ESPW:BilBen17,FC:BarZun17}, two-sided probabilistic swaps cannot be realized.
In contrast, our asymmetric formulation avoids this impossibility by ensuring that only one party's outcome is probabilistic, and hence that only one deciding message is required.

%% file: prototype.tex
\section{Prototype and Benchmarks}
\label{sec:benchmarks}

We implement an $\ell$-bit instantiation of our probabilistic swap protocol in Rust, giving winning probability $p=2^{-\ell}$, and evaluate it on both the Bitcoin and Litecoin testnets. 
The artifact implements the protocol equations, transcripts, OPRF-based claim mechanism, and transaction flows%
\ifnum\fullversion=1 and is available at~\cite{prototype}%
\fi. 
We implement both instantiations of the well-formedness proof for $Y_{\mathsf{win}}$: Bulletproofs over an R1CS constraint system using the constraint-friendly map-to-elliptic-curve-group relation of Groth et al.~\cite{AC:GMMZ25} and a MiMC-based hash-to-scalar, and our cut-and-choose construction with $\lambda=80$. We instantiate the adaptor signature scheme with Schnorr signatures over the secp256k1 curve, matching the cryptographic primitives natively supported by Bitcoin and Litecoin Taproot. The implementation uses the Rust crates \texttt{k256} for secp256k1 arithmetic, \texttt{ark-bulletproofs} and \texttt{merlin} for the Bulletproof backend and transcripts, and \texttt{bitcoin}/\texttt{bitcoincore-rpc} for Taproot transaction construction and daemon integration. The OPRF follows the 2HashDH construction from~\cite{7467360}. All benchmarks were obtained on a M3 Max MacBook Pro with 128GB memory.

\mypar{Micro benchmarks.}
We evaluated the local computation time of each aspect of our protocol and show benchmarks for the proof of well-formedness of $\paymentStmt$ in~\cref{tab:benchmarks_proofs}. As expected, all other local cryptographic operations are sub-millisecond, as shown in~\cref{tab:microbenchmarks}.

For the well-formedness proof, cut-and-choose is faster for proving at small $\ell$, with proof size and running time growing exponentially in $\ell$. At $\ell=10$, it proves in $0.27$s and verifies in $0.25$s, with a $296$KB proof. At $\ell=16$, the proof grows to roughly $16$MB and takes roughly $16$s to prove and $15$s to verify. The Bulletproof instantiation has a relatively constant proving time of $5.6$s, verification time of $0.35$s, and proof size of $2.8$KB.
We benchmarked the Bulletproof-based instantiation up to $\ell=128$, which still falls within the same range. 
These measurements were expected, since the cut-and-choose approach has to evaluate a Schnorr OR proof over all $2^\ell$ elements, so it scales linearly in the winning probability, whereas the Bulletproof-based approach runs a rangeproof, and an R1CS proof of correct hash-to-group, group exponentiation, and hash-to-scalar, which all require constant overhead (or logarithmic for the range-proof).
In practice, the choice of instantiation thus depends on the target winning probability and on whether the limiting resource is prover time or proof size. For large $\ell$ (i.e., small winning probability $p \ll 1$), the Bulletproof instantiation is preferable. For small $\ell$ (e.g., $p \geq 1/4096$), the cut-and-choose construction provides a sound and faster-proving alternative, while Bulletproof verification and proof size remain nearly constant.
\begin{table}
\centering
\ifnum\fullversion=0 
\scriptsize
\fi
\caption{Benchmark results for well-formedness proofs at selected values of $\ell$ ($p = 1/2^\ell$). 
We benchmark both a vanilla cut-and-choose, and a version with batched Schnorr proofs.}
\label{tab:benchmarks_proofs}

\begin{tabular}{r rrr rrr rrr}
\toprule
& \multicolumn{3}{c}{\textbf{Bulletproofs}} & \multicolumn{3}{c}{\textbf{Cut-and-Choose}} & \multicolumn{3}{c}{\textbf{Batched C\&C}} \\
\cmidrule(lr){2-4}\cmidrule(lr){5-7}\cmidrule(lr){8-10}
$\ell$ & Prove & Verify & Size & Prove & Verify & Size & Prove & Verify & Size \\
 & (s) & (s) & (KB) & (s) & (s) & (KB) & (s) & (s) & (KB) \\
\midrule
1  & 5.64 & 0.347 & 2.8 & 0.020 & 0.019 & 42.2  & 0.013 & 0.009 & 27.8  \\
4  & 5.63 & 0.347 & 2.8 & 0.024 & 0.022 & 45.7  & 0.016 & 0.012 & 31.4  \\
8  & 5.62 & 0.346 & 2.8 & 0.083 & 0.075 & 105.1 & 0.076 & 0.065 & 90.8  \\
10 & 5.58 & 0.346 & 2.8 & 0.272 & 0.246 & 295.7 & 0.265 & 0.235 & 281.4 \\
12 & 5.65 & 0.347 & 2.8 & 1.029 & 0.926 & 1057.9& 1.021 & 0.916 & 1043.3\\
16 & 5.54 & 0.347 & 2.8 & 16.19 & 14.53 & 16{,}297 & 16.17 & 14.53 & 16{,}284 \\
\midrule
128 & 5.65 & 0.348 & 2.8 & --- & --- & --- & --- & --- & --- \\
\bottomrule
\end{tabular}

\end{table}

\mypar{Lightning HTLC variant.}
We also implemented a local two-node Lightning prototype using \texttt{ldk-node}. This prototype targets the standard HTLC interface exposed by deployed Lightning implementations, and is therefore not the PTLC/adaptor-signature channel construction discussed in~\cref{sec:extensions}. Since a standard Lightning payment reveals a SHA256 preimage rather than an adaptor-signature witness, we bind each payment hash to the corresponding protocol witness with an off-chain Groth16 proof for the relation 
\[
\{
    ((X, h), \proofwitness): X=g^\proofwitness \land h=\mathrm{SHA256}(\proofwitness)
\}.
\]
We use this relation with $X=Z$ for the dealer payment and with $X=Y_{\mathsf{win}}$ for the prize payment. In our benchmark run, the dealer precomputed a single $Y_{\mathsf{win}}$ and three fresh ticket witnesses $z_i$, opened one local Lightning channel, and executed three HTLC rounds against the same $Y_{\mathsf{win}}$; the first two rounds were losing guesses and the final round was winning. The Groth16 link proofs were generated before channel funding. This precomputation took $21.96$s, channel setup took $0.35$s, and the online HTLC execution took $1.30$s total, or roughly $0.43$s per round. Each Groth16 link proof is $128$ bytes, with a $65$ byte public statement consisting of the compressed secp256k1 point $X$ and Lightning payment hash $h$.

\begin{table}
\centering
\ifnum\fullversion=0 
\footnotesize
\fi

\caption{Microbenchmarks for protocol functions outside $\relation_\win$. Times are in milliseconds.}
\label{tab:microbenchmarks}

\begin{tabular}{@{}lr@{\quad}|@{\quad}lr@{\quad}|@{\quad}lr@{}}
\toprule
\textbf{Function} & \textbf{Time} & \textbf{Function} & \textbf{Time} & \textbf{Function} & \textbf{Time} \\
\midrule
$\Gamma_\mathsf{DL}.\proofprove$ & $0.060$ & $\oprf.\eval$ & $0.118$ & $\Gamma_\pSign$ & $0.120$ \\
$\Gamma_\mathsf{DL}.\proofverify$ & $0.103$ & $\Pi_\enc.\enc$ & $0.091$ & $\pVrfy$ & $0.145$ \\
$\oprf.\Request$ & $0.064$ & $\Pi_\enc.\dec$ & $0.050$ & $\adapt$ & $0.287$ \\
$\oprf.\BlindEval$ & $0.043$ & $\proofsys.\proofprove(\relation_\enc)$ & $0.300$ & $\extract$ & $0.043$ \\
$\oprf.\Finalize$ & $0.109$ & $\proofsys.\proofverify(\relation_\enc)$ & $0.359$ & $\vrfy$ & $0.101$ \\
\bottomrule
\end{tabular}

\end{table}

\input{figure_proofs.tex}

\mypar{On-chain cost and fees.}
A complete winning probabilistic swap execution involves exactly four on-chain transactions: two funding transactions (one per party), one payment-claim transaction by the dealer, and one reward-claim transaction by the party. Each transaction uses a standard locking script $\Lambda(\mathsf{pk}_\mathsf{tmp}, T, \mathsf{pk})$ as defined in~\cref{fig:lockingscripts}, which requires only a single signature for the claim path and a timelock check for the refund path. These scripts are already supported natively by Bitcoin and Litecoin Taproot via \texttt{OP\_CHECKSIG} and \texttt{OP\_CHECKLOCKTIMEVERIFY}~(BIP~65), and do not require any additional opcodes or custom spending conditions beyond what standard atomic swaps already use~\cite{SP:ThyMalMor22}. Consequently, probabilistic swap transactions have the same on-chain structure as ordinary script-path Taproot spends.

In our public Bitcoin testnet4 run, the Bitcoin reward funding transaction paid a fee of $1550$ sat and had weight $616$WU, while the Bitcoin reward claim paid a fee of $500$ sat and had weight $545$WU.

\mypar{Cross-chain swap: Bitcoin $\times$ Litecoin.}
To test the cross-chain variant from~\cref{sec:extensions}, we ran the protocol on Bitcoin testnet4 and Litecoin testnet. Both chains support Taproot Schnorr spends over secp256k1, so the same adaptor-signature code applies on both sides. The public runs used $\ell=8$, hence winning probability $1/256$, and the batched cut-and-choose proof with $\lambda=80$. In the winning run, the dealer locked the Bitcoin reward and the party locked the Litecoin payment. The dealer claimed the Litecoin payment by revealing the adaptor witness $z$; the party extracted $z$, finalized the OPRF, and claimed the Bitcoin reward using $y_{\mathsf{win}}$.

We repeated the experiment with a losing guess on the same public test networks. The dealer still claimed the Litecoin payment, which revealed $z$. The party extracted $z$ and finalized the OPRF, but the result was not the winning adaptor witness. Adapting the Bitcoin reward signature therefore failed, and the dealer later spent the Bitcoin refund path after the timeout.

\begin{table}
\centering
\footnotesize

\caption{On-chain transaction IDs for Bitcoin $\times$ Litecoin cross-chain probabilistic swaps.}
\label{tab:txids}

\begin{tabular}{llll}
\toprule
\textbf{Run} & \textbf{Event} & \textbf{Chain} & \textbf{Transaction ID} \\
\midrule
Win  & Dealer funding   & Bitcoin testnet4 & \href{https://mempool.space/testnet4/tx/c69fc5c174fa9013367fca5870e333150c7928432c8532857e3015e8fedcec33}{\texttt{c69fc5c1...fedcec33}} \\
Win  & Party funding    & Litecoin testnet & \href{https://blockexplorer.one/litecoin/testnet/tx/32aa28fc6b66cf419ddf974189c0d7fa514598cfad049c8ecc0c6e453fc4f32e}{\texttt{32aa28fc...3fc4f32e}} \\
Win  & Dealer claim     & Litecoin testnet & \href{https://blockexplorer.one/litecoin/testnet/tx/e71cae45ab0d68ebaa6b6f314b0a8e808857b2bdb39c4cba7cde7c740e96c04f}{\texttt{e71cae45...740e96c04f}} \\
Win  & Party claim      & Bitcoin testnet4 & \href{https://mempool.space/testnet4/tx/7ad17bae032b5e9579582a83a1139301ea119d8cd86808716ede50a127e0fff8}{\texttt{7ad17bae...127e0fff8}} \\
Lose & Dealer funding   & Bitcoin testnet4 & \href{https://mempool.space/testnet4/tx/0f42a6e66756d053c1e4f92018cae71ac3d9aa3d3e6b5c2a70d608e2aab1ffd1}{\texttt{0f42a6e6...aab1ffd1}} \\
Lose & Party funding    & Litecoin testnet & \href{https://blockexplorer.one/litecoin/testnet/tx/08c9e270c522aa78da5b5f6494ec0556e959578052e7cc7da3d990280381fbde}{\texttt{08c9e270...0381fbde}} \\
Lose & Dealer claim     & Litecoin testnet & \href{https://blockexplorer.one/litecoin/testnet/tx/43a5c8648f42974a91b80e1de0b252700da316b4ea59ab14ee2036fb2cf78e6d}{\texttt{43a5c864...2cf78e6d}} \\
Lose & Dealer refund    & Bitcoin testnet4 & \href{https://mempool.space/testnet4/tx/6014382a6fdfe7326db91e7691f5e962ee2be322a0b7354ba5aad58806955d89}{\texttt{6014382a...06955d89}} \\
\bottomrule
\end{tabular}
\end{table}

%% file: figure_proofs.tex
\begin{figure}
\centering

\definecolor{bpblue}{RGB}{44,123,182}
\definecolor{ccgreen}{RGB}{31,160,120}

\ifnum\fullversion=0
  \pgfplotsset{mywidth/.style={width=0.2\textwidth, height=0.2\textwidth},
               myverifyplot/.style={xshift=-0.2cm}}
  \def\myhsep{1cm}
\else
  \pgfplotsset{mywidth/.style={width=0.33\textwidth, height=0.33\textwidth},
               myverifyplot/.style={}}
  \def\myhsep{1.2cm}
\fi

\begin{tikzpicture}
\begin{groupplot}[
    group style={group size=3 by 1, horizontal sep=\myhsep},
    mywidth,
    xmin=1, xmax=15,
    xtick={1,3,5,7,9,11,13,15},
    xlabel={$\ell$},
    grid=major,
    grid style={draw=gray!25},
    tick style={black!60},
    axis line style={black!70},
    label style={font=\small},
    tick label style={font=\small},
    legend style={
        fill=none,
        font=\scriptsize,
        legend columns=2,
        /tikz/every even column/.append style={column sep=0.8cm},
        at={(2,1.1)},
        anchor=south
    },
    unbounded coords=jump,
        ylabel style={yshift=-0.2cm, font=\scriptsize},
]

\nextgroupplot[
    ylabel={Prove time (s)},
    ymin=0, ymax=10,
    ytick={0,2,4,6,8,10},
]

\addplot[
    very thick,
    color=bpblue,
    restrict expr to domain={\thisrow{ell}}{1:15}
]
table[
    x=ell,
    y=bp_prove_s,
    col sep=comma
] {results.csv};
\addlegendentry{BP}

\addplot[
    very thick,
    color=ccgreen, dashed,
    restrict expr to domain={\thisrow{ell}}{1:15}
]
table[
    x=ell,
    y=batched_cut_and_choose_prove_s,
    col sep=comma
] {results.csv};
\addlegendentry{Cut-and-choose}

  \nextgroupplot[
      ylabel={Verify time (s)},
      ymin=0, ymax=8,
      ytick={0,2,4,6,8},
      myverifyplot
  ]

\addplot[
    very thick,
    color=bpblue,
    restrict expr to domain={\thisrow{ell}}{1:15}
]
table[
    x=ell,
    y=bp_verify_s,
    col sep=comma
] {results.csv};

\addplot[
    very thick,
    color=ccgreen, dashed,
    restrict expr to domain={\thisrow{ell}}{1:15}
]
table[
    x=ell,
    y=batched_cut_and_choose_verify_s,
    col sep=comma
] {results.csv};

\nextgroupplot[
    ylabel={Proof size (KB)},
    ymode=log,
    ymin=1, ymax=1e4,
    log basis y=10,
]

\addplot[
    very thick,
    color=bpblue,
    restrict expr to domain={\thisrow{ell}}{1:15},
]
table[
    x=ell,
    col sep=comma,
    create on use/bp_proof_kb/.style={
        create col/expr={\thisrow{bp_proof_bytes}/1024}
    },
    y=bp_proof_kb
] {results.csv};

\addplot[
    very thick,
    color=ccgreen, dashed,
    restrict expr to domain={\thisrow{ell}}{1:15},
]
table[
    x=ell,
    col sep=comma,
    create on use/cc_proof_kb/.style={
        create col/expr={\thisrow{batched_cut_and_choose_proof_bytes}/1024}
    },
    y=cc_proof_kb
] {results.csv};
\end{groupplot}
\end{tikzpicture}

\caption{Performance and proof size for different instantiations of the well-formedness proof of $\paymentStmt$. The parameter $\ell$ corresponds to the winning party's probability via $\winningProb = 2^{-\ell}$. }
\label{fig:results}
\end{figure}
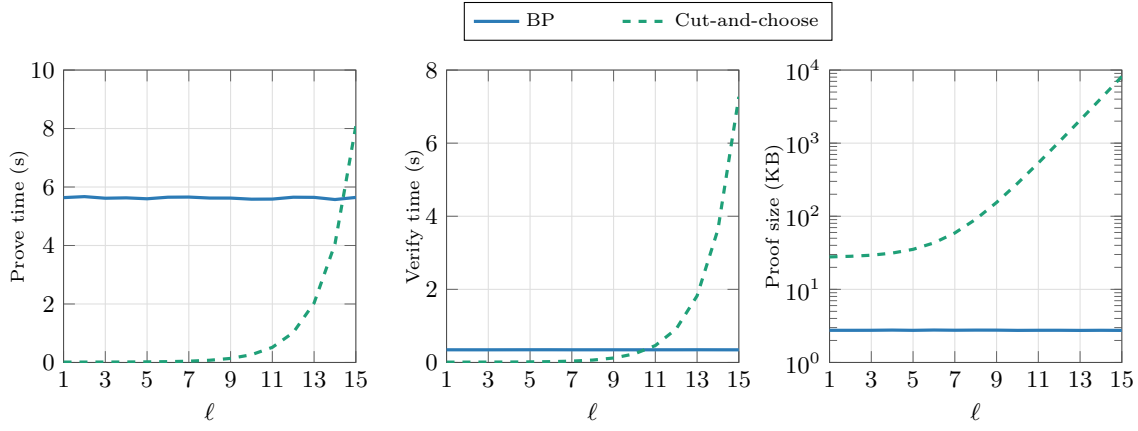

%% file: appdx_prelims.tex
\section{Preliminaries} \label{sec:appdx:prelims}
We recall preliminaries on oblivious PRFs, adaptor signatures, and zero-knowledge proofs.

\mypar{Oblivious PRFs.}
We recall the security properties of OPRFs: pseudorandomness, unlinkability, and uniqueness.
We base our definitions on the respective definitions for partially-oblivious PRFs from~\cite{EC:TCRSTW22}. These definitions imply the same notions for OPRFs, since an OPRF is a POPRF, simply by allowing just a single tag $t=\bot$. 
We provide the security experiments in~\ifnum\fullversion=0
\cref{fig:def:poprf:popriv2}%
\else%
\cref{fig:def:poprf:popriv2,fig:def:poprf:prf}%
\fi%
~and refer to~\cite{EC:TCRSTW22} for high-level explanations and the full-fledged definitions. 
We include uniqueness here, but do not require it in~\cref{thm:main}, since in our construction, well-formedness of the encrypted request is proven using an NIZK, which implies uniqueness of the decrypted value.

\newcommand{\pparam}{\ensuremath{pp}\xspace}
\newcommand{\poprf}{\oprf}
\newcommand{\randFn}{\ensuremath{\mathsf{RandFn}}\xspace}
\newcommand{\fun}{\mathsf{F}}
\newcommand{\init}{\ensuremath{\mathsf{Init}}\xspace}

\newcommand{\POPRIV}{\ensuremath{\mathsf{PoPRIV}}\xspace}
\newcommand{\UNIQUE}{\ensuremath{\mathsf{UNIQUE}}\xspace}

\begin{figure}
\centering
\ifnum\fullversion=0

	\scalebox{0.8}{%
	\parbox{\textwidth}{%
		\begin{pcvstack}
			\input{fig_poprif.tex}     
		\pcvspace
			\input{fig_oprf.tex}
		\end{pcvstack}
	}}
\else
\input{fig_poprif.tex}     
\fi

\ifnum\fullversion=0 
	\caption{Request privacy, uniqueness, and pseudorandomness experiments for OPRFs.}
	\label{fig:def:poprf:popriv2}\label{fig:def:poprf:unique}\label{fig:def:poprf:prf}
\else
\caption{Request privacy, uniqueness, and pseudorandomness experiments for OPRFs.}
	\label{fig:def:poprf:popriv2}\label{fig:def:poprf:unique}
\fi
\end{figure}

\ifnum\fullversion=1 
\begin{figure}
\centering
	\input{fig_oprf.tex}

	\caption{Pseudorandomness experiment for OPRFs.}
	\label{fig:def:poprf:prf}
\end{figure}
\fi

For the security proof of our theorem, we require non-blackbox access to the underlying structure of the OPRF. 
In particular, our simulator needs to be able to program the OPRF output to a target value. 
To allow this, we restrict our theorem to the family of two-hash OPRFs, which all share the same structure. 
We follow the abstraction from~\cite{EC:BDFH25}. 

\begin{definition}[Two Hash \oprf]
We say that an \( \oprf \) is \emph{two-hash} if it can be written as
$
\oprf(\mathsf{sk}, x) := H_2\big( F(\sk, H_1(x)) \big),
$
where $H_1$ and $H_2$ are hash functions, and $F$ can be computed efficiently. 
\end{definition}

\lightpar{Our modified 2HashDH OPRF.}
To enable cut-and-choose proofs, we slightly modify the 2HashDH OPRF construction, such that the dealer can provide openings \emph{without} revealing the OPRF evaluation's message. 
In particular, our OPRF has the following form: 
Let $\GG$ be a cyclic group of prime order $p$, and let $\hash_p$ and $\hash_\GG$ be hash functions mapping to $\ZZ_p$ and $\GG$, respectively. 
The public key of the server has the form $\pk = \left( X= g^\sk, \{g^{\alpha_i}\}, \{g^{\alpha_j}\}, \{\alpha_i \}\right)$.

To form a request for input $x$, the client runs $\Request(x)$ by sampling $r \sample \ZZ_p$ and computing $\request = \hash_\GG(x)^r$. 
Upon receiving $\request$, the server computes $\response = \BlindEval(\sk, \request) = \request^\sk$. 
Requests and responses are identical to the original 2HashDH~\cite{7467360}. 
However, to finalize, the client runs $\Finalize(x, r, \response)$ to remove the blinding and obtains $\eval(\sk, x) = \{\hash_p(\pk, \response^{\alpha_i/r})\} = \{\hash_p(\pk, \hash_\GG(x)^{\sk\cdot \alpha_i})\}$ for all elements $\alpha_i$ part of the public key.
The main difference between this OPRF and 2HashDH is two things:
First, the first input to $\hash_p$, is $\pk$ in our case and $x$ in the case of 2HashDH. The second difference is that the OPRF finalization is run on each partial public key element, instead of a single one. 

This change does not change the security of the OPRF. 
To see this, we first recall the pseudorandomness proof of 2HashDH. 
The proof of pseudorandomness of 2HashDH is based on the one-more-gap DH assumption~\cite{7467360}. The simulator handles queries to the RO $\hash_\GG$ oracle using the challenge oracle of the Gap-OMDH assumption, which outputs CDH challenges; it also sets $X$ to a CDH challenge from the help oracle. 
When the adversary queries the oracle $\hash_p$ on input $(x, Y)$, the simulator checks, if $Y$ is a valid CDH solution by querying $(\hash_GG(x), \pk, Y)$ to the DDH oracle of the Gap-OMDH assumption. 
If this is the case, and the quota $q$ of the pseudorandomness experiment is greater than or equal $0$, i.e., the adversary did more $\BlindEv$ queries than queries to $\hash_p$ on such valid inputs, the simulator simply queries the $\mathsf{LimEv}$ oracle, and can program the output of $\hash_p$ accordingly.
However, if the adversary comes up with such an input $(x, Y)$ that denotes a solution for CDH, but the quota $q$ is exhausted, the simulator cannot use the $\mathsf{LimEv}$ oracle to answer this query. 
Yet, this is not a problem for the simulator, since in this case the adversary provided one more solution to the CDH problem, which the simulator can output to break the Gap-OMDH assumption. 
As the Gap-OMDH assumption is hard, this edge case occurs only with negligible probability, and the simulator suffices to show pseudorandomness. 

To show the pseudorandomness of our scheme in a similar fashion, we simulate the RO $\hash_p$ and the $\BlindEv$ oracle exactly as in the proof of 2HashDH. 
However, when the adversary provides a value $(\pk, Y)$ to the $\hash_p$ oracle, we handle the situation differently. 
This is, the simulator checks for all pairs $(x, H_p(x))$ the adversary received from $\hash_p$ so far, and all values $\alpha_i \in \pk$, if the triple $(\pk, H_p(x), Y^{1/\alpha_i})$ denotes a valid DDH triple.
This introduces a runtime overhead of $64 \cdot |mathcal{Q}|$, where $\mathcal{Q}$ is the size of queries to $H_p$, but is still efficient. 
In addition, the domain of the OPRF is just 64x bigger, so for each message $x$, the simulator partitions the output of the $\mathsf{LimEv}$ oracle in 64 parts, and if $(H(x), \alpha_i, pk)$ denotes a valid CDH triple, hands out the $i-$th part of this partition. This resolves the issue that for a single CDH solution, the adversary can query the oracle $\hash_p$ 64 times. 
Therefore, whenever the simulator cannot answer a query to the adversary as the $\mathsf{LimEv}$ oracle did not provide a solution, the simulator finds a solution for the Gap-OMDH assumption.

\mypar{Adaptor signatures.}
We recall the two main security guarantees of adaptor signatures, namely extractability and pre-signature adaptability following the notion of~\cite{EC:GSST24}.

\begin{definition}[Extractability] \label{def:FExt}
    An adaptor signature scheme $\adaptorScheme$ is extractable, if for every $\ppt$ adversary $\adv$ there exists a negligible function $\negl[]$ such that for every $\lambda \in \NN$
    \[
     \prob{\fext_{\adv, \adaptorScheme}(\secpar) = 1} \leq \negl[\secpar] \ ,
    \]
    where the experiment $\fext_{\adv, \adaptorScheme}$ is described in \cref{fig:fullExtractability}, and the probability is taken over the random choices of all probabilistic algorithms.
\end{definition}

\begin{figure}
\centering
\small

\ifnum\fullversion=0 
\begin{pcvstack}[boxed]
	\pcsetargs{linenumbering}
	\procedure{$\fext_{\adv, \adaptorScheme}(\secpar)$}{
		(\sk, \pk) \gets \kgen(\secparam); b \gets 1;
		\setS, \setT \gets \emptyset \\
		(m^*, \sigma^*) \gets \adv(\pk)^{\newY(\secpar), \sign(\sk,\cdot), \pSign(\sk,\cdot,\cdot)}\\
		\pcassert \vrfy(\pk, m^*, \sigma^*)\\
		\pcassert (m^* \notin \setS) \\
		\pcfor (Y, \pSig) \in \setT[m^*] \\
		\hspace{1em} \pcif (Y, \extract(Y, \pSig, \sigma^*)) \in \relation \pcthen \\
		\hspace{2em} b \gets 0 \\
		\pcreturn b
	}	
	\pcvspace
	\procedure{$\newY(\secpar)$}{
	(Y, y) \gets \relation.\genRelation(\secparam);\;
	\pcreturn Y
	}

	\pcvspace
	\begin{pchstack}
		\procedure{$\sign(\sk,m)$}{
			\sigma \gets \Sigma.\sign(\sk, m) \\
			\setS \gets \setS \cup \{m\}\\
			\pcreturn \sigma
		}	
		\pchspace
		\procedure{$\pSign(\sk,m, Y)$}{
			\pSig \gets \adaptorScheme.\pSign(\sk, m, Y)\\
			\setT[m] \gets \setT[m] \cup \{(Y, \pSig)\}\\
			\pcreturn \pSig
		}
	\end{pchstack}
\end{pcvstack}
\else

\begin{pchstack}[boxed]
	\pcsetargs{linenumbering}
	\procedure{$\fext_{\adv, \adaptorScheme}(\secpar)$}{
		(\sk, \pk) \gets \kgen(\secparam); b \gets 1;
		\setS, \setT \gets \emptyset \\
		(m^*, \sigma^*) \gets \adv(\pk)^{\newY(\secpar), \sign(\sk,\cdot), \pSign(\sk,\cdot,\cdot)}\\
		\pcassert \vrfy(\pk, m^*, \sigma^*)\\
		\pcassert (m^* \notin \setS) \\
		\pcfor (Y, \pSig) \in \setT[m^*] \\
		\hspace{1em} \pcif (Y, \extract(Y, \pSig, \sigma^*)) \in \relation \pcthen \\
		\hspace{2em} b \gets 0 \\
		\pcreturn b
	}	
	\pchspace
	
	\begin{pcvstack}

		\procedure{$\sign(\sk,m)$}{
			\sigma \gets \Sigma.\sign(\sk, m) \\
			\setS \gets \setS \cup \{m\}\\
			\pcreturn \sigma
		}	
		\pcvspace
		\procedure{$\pSign(\sk,m, Y)$}{
			\pSig \gets \adaptorScheme.\pSign(\sk, m, Y)\\
			\setT[m] \gets \setT[m] \cup \{(Y, \pSig)\}\\
			\pcreturn \pSig
		}
		\pcvspace
			\procedure{$\newY(\secpar)$}{
	(Y, y) \gets \relation.\genRelation(\secparam);\;
	\pcreturn Y
	}
	\end{pcvstack}
\end{pchstack}
\fi

	\caption{The security game $\fext_{\adv,\adaptorScheme}(\secpar)$.}
	\label{fig:fullExtractability}
\end{figure}

\begin{definition}[Pre-signature adaptability]
    \label{def:PreSignatureAdaptability}
    An adaptor signature scheme $\adaptorScheme$ satisfies
    \emph{pre-signature adaptability}, if for all $\secpar \in \NN$, messages $m \in \bin^*$,
    statement/witness pairs $(Y, y) \in \relation$, public keys $\pk$ and pre-signatures 
    $\pSig \in \bin^*$ we have $\pVrfy(\pk, m, Y, \pSig) = 1$, 
    then $\vrfy(\pk, m, \adapt(\pk, \pSig, y)) = 1$.
\end{definition}

	\subsection{Non-interactive Arguments} 
	
\newcommand{\defcal} [1]{\expandafter\newcommand\csname c#1\endcsname{\mathcal{#1}}}
\newcounter{ct}
\forloop{ct}{1}{\value{ct} < 27}{
    \edef\letter{\Alph{ct}}
    \expandafter\defcal\letter
}
\label{sec:nizk}
We follow the notion of~\cite{EC:GLRS26}.
A \emph{non-interactive argument system} (NARG) for relation $\cR$ in the random oracle model,
denoted by $\proofsys$, consists of a tuple of algorithms $(\proofsetup, \cP, \cV)$ having black-box access to a random oracle
$\hash:\bin^* \to \bin^\secpar$, with the following syntax:
\begin{itemize}[leftmargin=*,nosep]
  \item $\nizkpp \gets \setup(\secparam)$: takes as input the security parameter $\secparam$ and outputs public parameters $\nizkpp$.
  \item $\pi \gets \cP^\hash(\nizkpp,\proofstatement,\proofwitness)$: takes as input parameters $\nizkpp$, a statement $\proofstatement$ and witness $\proofwitness$, and outputs a proof $\pi$ if $(\proofstatement,\proofwitness)\in\cR$.
  \item $b \gets \cV^\hash(\nizkpp,\proofstatement,\pi)$: takes as input parameters $\nizkpp$, a statement $\proofstatement$ and proof $\pi$, and it accepts ($b = 1$) or rejects ($b = 0$).
\end{itemize}
In this work, we consider NARGs with \emph{transparent} setup, i.e. $\nizkpp$ can be generated with a call to the random oracle.
For this reason, we may omit $\nizkpp$ in the description of the prover and verifier algorithms as it is implicit.

\begin{definition}[Completeness]
  A NARG $\proofsys$ satisfies \emph{completeness} if for every $(\proofstatement,\proofwitness)\in\cR$, it holds that
\begin{align*}
\prob{
     \begin{array}{l}
		b=1: \nizkpp\gets\setup(\secparam); \pi \gets \cP^\hash(\nizkpp,\proofstatement,\proofwitness); \\ b \gets \cV^\hash(\nizkpp,\proofstatement,\pi)
\end{array}
} = 1.
\end{align*}
\end{definition}
\begin{definition}[Soundness]
  A NARG $\proofsys$ satisfies \emph{soundness} if for every PPT adversary $\adv$, it holds that
    \[
    \prob{
     \begin{array}{l}
	    b=1 \land \proofstatement\not\in\cL_\cR : \nizkpp\gets\setup(\secparam); \\
	    (\proofstatement,\pi) \gets \adv^\hash(\nizkpp); b \gets \cV^\hash(\nizkpp,\proofstatement,\pi)
      \end{array}  
    }
\]
is negligible, where $\cL_\cR\coloneqq \set{\proofstatement \,|\, \exists \proofwitness: (\proofstatement,\proofwitness)\in\cR}$ is the set of true-statements.
\end{definition}
If soundness holds with respect to a computationally unbounded adversary, we say that the argument has \emph{statistical} (or \emph{information-theoretic}) soundness, and we call the argument a \emph{proof} system.

\mypar{Zero-knowledge.}
Informally, an argument is zero-knowledge if a proof reveals no information about the witness~\cite{STOC:GolMicRac85}.
We formalize this property by following the syntax of~\cite{INDOCRYPT:FKMV12}.
A zero-knowledge simulator $\cS$ is defined as a stateful algorithm, whose initial state $\st = \nizkpp$, that operates in two modes.
The first mode, $(y, \st') \gets \cS(1, \st, a)$ takes care of handling calls to the oracle $\hash$ on input a query $a$.
The second mode, $(\pi, \st') \gets \cS(2, \st, \proofstatement)$ simulates a proof for the input statement $\proofstatement$.
We define the following \emph{wrapper} oracles, that are stateful and share their internal state:
\begin{itemize}[leftmargin=*,nosep]
\item $\cS_1(a)$ returns the first output of $\cS(1,\st,a)$;
\item $\cS_2(\proofstatement,\proofwitness)$ returns the first output of $\cS(2,\st,\proofstatement)$ if $(\proofstatement,\proofwitness) \in \cR$ and $\bot$ otherwise;
\item $\cS_2'(\proofstatement)$ returns the first output of $\cS(2,\st,\proofstatement)$.
\end{itemize}
\begin{definition}[Zero-knowledge]\label{def:zk}
  A NARG $\proofsys$ is zero-knowledge if there exists a simulator $\cS$, with wrapper oracles $\cS_1,\cS_2$, such that for all PPT adversaries $\cA$ it holds that:
  \[
    \prob{
      \begin{array}{l}
        \nizkpp \gets \setup(\secparam) \\
        \cA^{\cP,\hash}(\nizkpp) = 1
      \end{array}} \approx
    \prob{
      \begin{array}{l}
        \nizkpp \gets \setup(\secparam) \\
        \cA^{\cS_1,\cS_2}(\nizkpp) = 1
      \end{array}
    }
  \]
\end{definition}
Notice that zero-knowledge is a security property that is only guaranteed for valid statements in the language, hence the above definition uses $\cS_2$ as a proof simulation oracle.
Also, a zero-knowledge NARG with statistical soundness is simply referred to as a non-interactive zero-knowledge proof (NIZK)~\cite{STOC:BluFelMic88}.

\mypar{Simulation exctractability.}
A strong security property of NARGs is \emph{simulation extractability}~\cite{C:DDOPS01,C:GroMal17}.
Informally, from a prover producing a valid proof after seeing polynomially many proofs provided by a simulator (even for \emph{false} statements), it is possible to extract a valid witness.
This property can be formalized in many different ways, depending on the model and the constraints of the adversary and the extractor.
We refer to~\cite{C:BCCRS25} for a recent work that summarizes the main differences.
\begin{definition}[Simulation extractability]\label{def:simext}
  A NARG $\proofsys$ is simulation-extractable with respect to a zero-knowledge simulator $\cS$, with wrapper oracles $\cS_1,\cS_2'$, if for all PPT adversaries $\cA$ there exists a PPT extractor $\cE$ such that:
  \[
    \condprob{
      \begin{array}{l}
        \cV^{\cS_1}(\proofstatement,\pi) = 1 \\
        \land\;(\proofstatement,\pi)\not\in\cQ \\
        \land\; (\proofstatement,\proofwitness) \not\in\cR
      \end{array}
    }{
      \begin{array}{l}
        \nizkpp \gets \setup(\secparam)\\
        (\proofstatement,\pi) \gets \cA^{\cS_1,\cS_2'}(\nizkpp)\\
        \proofwitness \gets \cE(\st,\proofstatement,\pi)
      \end{array}
    } \in \negl
  \]
  where $\st$ is the final internal state of $\cS$, and $\cQ$ is the set of statement-proof pairs obtained by $\cA$ when interacting with $\cS_2'$.
\end{definition}

The previous definition implies the weaker notion of \emph{knowledge soundness}~\cite{FOCS:Sahai99}.
Informally, a NARG is knowledge-sound if there exists a PPT extractor, but without simulating proofs to the adversary first.

%% file: fig_poprif.tex
\begin{pchstack}     
\begin{pcvstack}[boxed]
\pcsetargs{linenumbering}
\procedure{$\POPRIV _{\adv, \poprf}^b(\secpar)$}{
   \pparam \gets \poprf.\setup(\secpar)\\
	i \gets 0\\
    b' \gets \adv^{\Request, \Finalize}(\pparam)\\
   \pcreturn b'
}


\procedure[lnstart=4,linenumbering]{$\Request(\pk, x_0, x_1)$}{
	i \gets i +1\\
	(\st_{i, 0}, \request_0) \gets \Request(\pk, x_0)\\
	(\st_{i, 1}, \request_1) \gets \Request(\pk, x_1)\\
	\pcreturn (\request_b, \request_{1-b})
}
\procedure[lnstart=8,linenumbering]{$\Finalize(j, \response, \response')$}{
	\pcassert j \leq i\\
	y_b \gets \Finalize(\st_{j, b}, \response)\\
	y_{1-b} \gets \Finalize(\st_{j, 1-b}, \response')\\
	\pcif y_0 = \bot \lor y_1 = \bot \\ \t \pcreturn \bot\\
	\pcreturn (y_0, y_1)
}

\end{pcvstack}
\pchspace
\begin{pcvstack}[boxed]
\pcsetargs{linenumbering}
\procedure{$\UNIQUE_{\adv, \poprf}(\secpar)$}{
	\queue \gets \emptyset; \;   	q \gets 0\\
   \pparam \gets \poprf.\setup(\secpar)\\
   	(i, j, \response_i, \response_j) \gets \adv^{\Request}(\pparam)\\
   	\pcassert 1 \leq i \leq j \leq q\\
   	(\st_i, \pk_i,  x_i) \gets \queue[i]\\
   	(\st_j, \pk_j,  x_j) \gets \queue[j]\\
   	\pcassert (\pk_i, x_i) = (\pk_j, x_j)\\
   	y_i \gets \Finalize(\st_i, \response_i)\\
   	y_j \gets \Finalize(\st_j, \response_j)\\
	\pcif y_0 = \bot \lor y_1 = \bot\\ \t  \pcreturn 0\\
   \pcreturn y_i \neq y_j
}


\procedure[lnstart=12,linenumbering]{$\Request(\pk, x)$}{
	q \gets q+1\\
	(\st, \request) \gets \Request(\pk, x)\\
	\queue[q] \gets (\st, \pk, x)\\
	\pcreturn  \request
}
\end{pcvstack}

\end{pchstack}

%% file: fig_oprf.tex
\begin{pchstack}[boxed]

\begin{pcvstack}
\procedure[linenumbering]{$\poprf^b_{\adv,\poprf, \simulator}(\secpar)$}{
   \randFn \sample \fun, q \gets 0 \\
   \st_\mathsf{P} \gets \mathsf{P}.\init(\secpar) \\
   \pparam \gets \poprf.\setup(\secpar)\\
   (\sk, \pk_1) \gets \poprf.\kgen^\mathsf{P}(\pparam)\\
   (\st_\simulator, \pk_0) \gets \simulator.\init(\pparam)\\
   b' \gets \adv^{\eval, \BlindEval, \Prim}(\pparam, \pk_b)\\
   \pcreturn b'
}

\pcvspace

\procedure[lnstart=7,linenumbering]{$\BlindEval(\request)$}{
	q \gets q + 1\\
	\response_1 \gets \oprf.\BlindEval(\sk, \request)\\
	\response_0 \gets \simulator.\BlindEval^\mathsf{LimEv}(\st_\simulator, \request)\\
	\pcreturn \response_b
}

\end{pcvstack}
\pchspace
\begin{pcvstack}
\procedure[lnstart=11,linenumbering]{$\eval(x)$}{
	y_1 \gets \oprf.\eval^\mathsf{P}(\pk, \sk, x)\\
	y_0 \gets \randFn(x)\\
	\pcreturn y_b
}

\pcvspace

\procedure[lnstart=14,linenumbering]{$\Prim(x)$}{
	y_1 \gets \mathsf{P}.\eval(x, \st_\mathsf{P})\\
	(y_0, \st_\simulator) \gets \simulator.\eval^{\mathsf{LimEv}}(x, \st_\simulator)\\
	\pcreturn y_b
}

\pcvspace

\procedure[lnstart=17,linenumbering]{$\mathsf{LimEv}(x)$}{
	q \gets q - 1\\
	\pcif q \geq 0 \pcreturn \eval(x)\\
	\pcreturn \bot
}
\end{pcvstack}
\end{pchstack}

%% file: instantiation.tex
\section{Instantiation of Proofsystems} \label{sec:instances}

Our protocol requires proofs of well-formedness and proofs of knowledge. 
In this section, we provide instantiations of the relations to be proven, along with sigma protocols for proving them. 
We keep the sigma protocols interactive, but they can be made non-interactive using the Fiat-Shamir transform~\cite{C:FiaSha86}.

\mypar{Proof of well-formed winning statement $\paymentStmt$.} The relation $\relation_\win$ is satisfied, if a winning statement $\paymentStmt$ is well-formed w.r.t. some OPRF public key $\pk$, and a length parameter $\ell$. For this, we provide two relations. The first one, $\relation_\win$, we use when proving well-formedness using Bulletproofs. 
\[
	\relation_\win := \left\{
        \begin{array}{l}
          \proofstatement = (\pk, \paymentStmt, \ell) \\
          \proofwitness = (\sk, \witnessTarget)
        \end{array}:
        \begin{array}{l}
        	|\witnessTarget| \leq 2^{\ell},\\
        	\paymentWit = \oprf.\eval(\sk, \witnessTarget),\\
        	\paymentStmt = g^{\paymentWit} \land \pk = g^{\sk}
        \end{array}
        \right\}
\]   
To prove this relation, we can instantiate the RO and prove well-formedness of $\proofstatement$ using, e.g., Bulletproofs~\cite{SP:BBBPWM18} and a constraint system like R1CS~\cite{SP:PHGR13}. 
If we do not want to instantiate the RO for our proof, we can use the relation $\relation_\win'$ for the well-formedness of $\paymentStmt$.
\[
	\relation_\win' := \left\{
        \begin{array}{l}
          \proofstatement = (\pk, \ciphertext, \paymentStmt, \ell) \\
          \proofwitness = (\sk, \witnessTarget)
        \end{array}:
        \begin{array}{l}
        	|\witnessTarget| \leq 2^{\ell}\\ 
            (\paymentStmt, \paymentWit) \gets \genRelation(\secpar)\\
        	\rho = \oprf.\eval(\sk, \witnessTarget),\\
        	\ciphertext = \rho + \paymentWit \land \pk = g^{\sk}
        \end{array}
        \right\}
\] 
 
To prove well-formedness using $\relation_\win'$, we can use the cut-and-choose strategy, as depicted in~\cref{fig:proofs:cutandchoose,fig:proofs:orSchnorr}. 
 
\mypar{Proof of well-formed OPRF response encryption.}
This proof is the technical cornerstone of the OPRF-as-a-service abstraction introduced in~\cref{sec:overview}: it lets the party verify that the dealer has honestly computed a blind OPRF evaluation on the party's request and locked the result in a well-formed ciphertext, without yet revealing the output.
The relation $\relation_\enc$ is satisfied if a ciphertext $\ciphertext$ is well-formed w.r.t. some OPRF public key $\pk$, a statement $Z$, and an OPRF request $\request$.
\[
	\relation_{\enc} := \left\{
        \begin{array}{l}
          \proofstatement = (\pk, Z, \request, \ciphertext) \\
          \proofwitness = (\sk, z)
        \end{array}:
        \begin{array}{l}
			\pk = g^{\sk} \land Z = g^{z},\\
			\response = \oprf.\BlindEv(\sk, \request),\\
			\ciphertext = \Pi_\enc.\enc(Z, \response)
        \end{array}
        \right\}
\]  
If $\Pi_\enc$ is instantiated with ElGamal encryption, we show this well-formedness using a Chaum-Pedersen sigma protocol, depicted in~\cref{fig:proofs:schnorrEnc}.
\begin{figure}
	\begin{textdef}
		Public inputs: a group $(\GG,g,p)$ and 
		$
			\proofstatement = (\pk, Z, \request, \ciphertext).
		$\\ 
		The prover provides a witness
		$
			\proofwitness = (\sk,\alpha)
		$
		such that
		$
			\pk = g^{\sk},\quad
			Z = g^z,\quad
			\ciphertext = (c_1,c_2) = (g^\alpha,\; Z^\alpha \cdot \request^{\sk}).
		$\\
		The verifier has no additional input.

		\textbf{Proving:}
		\begin{enumerate}[leftmargin=*]
			\item The prover samples
			$
				r_{\sk}, r_{\alpha} \sample \ZZ_p
			$
			and computes the commitments
			$
				R_1 := g^{r_{\sk}},		
				R_2 := g^{r_{\alpha}},
				R_3 := Z^{r_{\alpha}} \cdot \request^{r_{\sk}}.
			$
			It sends $(R_1,R_2,R_3)$ to the verifier.

			\item The verifier samples a challenge
			$
				c \sample \ZZ_p
			$ 
			and sends it to the prover.

			\item The prover computes the responses
			$
				z_{\sk} := r_{\sk} + c \cdot \sk,
				z_{\alpha} := r_{\alpha} + c \cdot \alpha,
			$
			and sends $(z_{\sk}, z_{\alpha})$ to the verifier.
		\end{enumerate}

		\textbf{Verification:}
		Parse $\ciphertext$ as $(c_1,c_2)$. Upon receiving $(R_1,R_2,R_3,z_{\sk},z_{\alpha})$, the verifier accepts iff
		\[
			g^{z_{\sk}} \stackrel{?}{=} R_1 \cdot \pk^c \land 
			g^{z_{\alpha}} \stackrel{?}{=} R_2 \cdot c_1^c \land
			Z^{z_{\alpha}} \cdot \request^{z_{\sk}}
			\stackrel{?}{=}
			R_3 \cdot c_2^c.
		\]
	\end{textdef}

	\caption{Interactive Schnorr proof for the relation $\relation_{\enc}$.}
	\label{fig:proofs:schnorrEnc}
\end{figure}

\begin{lemma}\label{lem:schnorrEnc}
The sigma protocol in~\cref{fig:proofs:schnorrEnc} is complete, $2$-special sound with knowledge error $1/p$, and honest-verifier zero-knowledge.
\end{lemma}
\begin{proof}
\lightpar{Completeness.} For any valid witness $(\sk, \alpha)$ the verification equations follow directly from expanding $R_1, R_2, R_3$ and the responses.

\lightpar{Special soundness.} Given two accepting transcripts with identical first message $(R_1,R_2,R_3)$ but distinct challenges $c \neq c'$, let $e = c-c' \neq 0$. From the first equation $g^{z_\sk - z'_\sk} = \pk^e$, so $\sk = (z_\sk - z'_\sk)/e$. From the second, $\alpha = (z_\alpha - z'_\alpha)/e$. Substituting into the third equation yields $Z^\alpha \cdot \request^\sk = c_2$, so $(\sk,\alpha) \in \relation_\enc$.

\lightpar{Honest-verifier zero-knowledge.} The simulator samples $c, z_\sk, z_\alpha \sample \ZZ_p$ and sets
$
  R_1 := g^{z_\sk}\pk^{-c}, R_2 := g^{z_\alpha}c_1^{-c}, R_3 := Z^{z_\alpha}\cdot\request^{z_\sk}\cdot c_2^{-c}.
$
The resulting transcript satisfies all verification equations and is identically distributed to an honest one, since the honest first-round values are uniformly random in $\GG$ given the challenge.
\end{proof}

\mypar{Proofs of knowledge for commitment-opening, and secret keys.}      
The relation $\relation_\com$ is satisfied, if a party knows an opening $(\sk, \comrand)$ for a commitment $\com = g^{\sk}h^{\comrand}$.
The relation $\relation_\mathsf{DL}$ is satisfied if a party knows the discrete logarithm $\sk$ of a group element $\pk$.
\begin{align*}        
	\relation_\com &:= \left\{
        \begin{array}{l}
          \proofstatement = (\com, \GG, g, h) \\
          \proofwitness = (\sk, \comrand)
        \end{array}:
        \begin{array}{l}
			\com = g^{\sk}h^{\comrand}
        \end{array}
        \right\}\\
	\relation_\mathsf{DL} &:= \left\{
        \begin{array}{l}
          \proofstatement = (\pk, \GG, g) \\
          \proofwitness = \sk
        \end{array}:
        \begin{array}{l}
			\pk = g^{\sk}
        \end{array}
        \right\}\\
\end{align*}
We depict sigma protocols for these PoKs in~\cref{fig:proofs:schnorrCom,fig:proofs:schnorrDL}.
\begin{figure}
	\begin{textdef}
		Public inputs: a group $(\GG,g,p)$, an additional generator $h \in \GG$, and a statement
		$
			\proofstatement = (\com,\GG,g,h).
		$\\ 
		The prover provides a witness
		$
			\proofwitness = (\sk,\comrand)
		$
		such that
		$
			\com = g^{\sk} h^{\comrand}.
		$\\
		The verifier has no additional input.

		\textbf{Proving:}
		\begin{enumerate}[leftmargin=*]
			\item The prover samples
			$
				r_{\sk}, r_{\comrand} \sample \ZZ_p
			$
			and computes the commitment
			$
				R := g^{r_{\sk}} h^{r_{\comrand}}.
			$ 
			It sends $R$ to the verifier.

			\item The verifier samples a challenge
			$
				c \sample \ZZ_p
			$ 
			and sends it to the prover.

			\item The prover computes the responses
			$
				z_{\sk} := r_{\sk} + c \cdot \sk ,
				z_{\comrand} := r_{\comrand} + c \cdot \comrand ,
			$
			and sends $(z_{\sk}, z_{\comrand})$ to the verifier.
		\end{enumerate}

		\textbf{Verification:}
		Upon receiving $(R,z_{\sk},z_{\comrand})$, the verifier accepts iff
		$
			g^{z_{\sk}} h^{z_{\comrand}}
			\stackrel{?}{=}
			R \cdot \com^c.
		$
	\end{textdef}

	\caption{Interactive Schnorr proof for the relation $\relation_\com$.}
	\label{fig:proofs:schnorrCom}
\end{figure}
\begin{figure}
	\begin{textdef}
		Public inputs: a group $(\GG,g,p)$ and a statement
		$
			\proofstatement = (\pk,\GG,g).
		$\\ 
		The prover provides a witness
		$
			\proofwitness = \sk
		$
		such that
		$
			\pk = g^{\sk}.
		$\\
		The verifier has no additional input.

		\textbf{Proving:}
		\begin{enumerate}[leftmargin=*]
			\item The prover samples
			$
				r \sample \ZZ_p
			$
			and computes
			$
				R := g^r.
			$
			It sends $R$ to the verifier.

			\item The verifier samples a challenge
			$
				c \sample \ZZ_p
			$
			and sends it to the prover.

			\item The prover computes the response
			$
				z := r + c \cdot  \sk 
			$
			and sends $z$ to the verifier.
		\end{enumerate}

		\textbf{Verification:}
		Upon receiving $(R,z)$, the verifier accepts iff
		$
			g^z \stackrel{?}{=} R \cdot \pk^c.
		$
	\end{textdef}

	\caption{Interactive Schnorr proof for the relation $\relation_\mathsf{DL}$.}
	\label{fig:proofs:schnorrDL}
\end{figure}

\begin{figure*}

	\begin{textdef}
		Public inputs: a group $(\GG, g, p)$, and hash functions $\hash_\GG: \bin^{*} \to \GG$, $\hash_p: \bin^{*} \to \ZZ_p$, and a statement $\proofstatement = (\pk, \paymentStmt, \ell)$.\\
		The prover provides $\proofwitness = (\sk, \witnessTarget)$ as additional input, and the verifier has no additional input. \\
		\textbf{Proving:}
		\begin{enumerate}[leftmargin=*]
			\item The prover samples scalars $(\alpha_1, \ldots, \alpha_\secpar, r_1, \ldots, r_\secpar) \sample \ZZ_p$, computes the $\secpar$ commitments
			\[
				\left\{ \hash_p\big(\hash_\GG(\witnessTarget)^{\sk \cdot \alpha_i}\big) + r_i, R_i = g^{r_i}, A_i = g^{\alpha_i} \right\}_{1\leq i \leq \secpar},
			\]
			and sends them to the verifier. 
			\item The verifier selects $\secpar/2$ random indices and sends them to the prover.
			\item For each index $j$ selected by the verifier, the prover outputs $(r_j, \hash_\GG(y)^{\sk \cdot \alpha_j})$, and a proof of well-formedness of $\hash_\GG(y)^{\sk \cdot \alpha_j}$ w.r.t. $\pk$, $g^{\alpha_j}$, and all possible values $\hat y \in \bin^{\ell}$ (depicted in~\cref{fig:proofs:orSchnorr}). 
			In addition, for each unselected commitment with index $k$, the prover outputs $(\alpha_k, s_k = r_k + \paymentWit)$.
		\end{enumerate}
		\textbf{Verification:} The verifier obtains $r_j, \hash_\GG(y)^{\sk \cdot \alpha_j}$ and a proof of well-formedness for each selected index $j$, and the tuple $(\alpha_k, s_k)$ for each unselected index $k$. 
		\begin{enumerate}[leftmargin=*]	
			\item For each selected index $j$, the verifier verifies the opening by checking the proof of well-formedness, and by asserting that $r_j$ is a correct opening for $R_j$. 
			\item For each unselected index $k$, the verifier asserts that the value $s_k$ is well-formed by checking if $\paymentStmt \cdot R_k \stackrel{?}{=} g^{s_k}$ and  $A_k = g^{\alpha_k}$.
		\end{enumerate}

	\end{textdef}
	\caption{Interactive cut-and-choose proof for $\relation_\win$.}
	\label{fig:proofs:cutandchoose}
\end{figure*}

\begin{lemma}\label{lem:cutandchoose}
The cut-and-choose protocol in~\cref{fig:proofs:cutandchoose} is complete, computationally sound under the DDH assumption, and honest-verifier zero-knowledge, assuming the HVZK of Lemma~\ref{lem:orSchnorr}.
\end{lemma}
\begin{proof}
\lightpar{Completeness.} An honest prover with valid witness $(\sk, \witnessTarget)$ always passes all checks: for selected index $j$, the OR proof verifies that $T_j = \hash_\GG(\witnessTarget)^{\sk\cdot\alpha_j}$ is well-formed, the opening $c_j = \hash_p(T_j) + r_j$ is consistent, and $R_j = g^{r_j}$. For unselected index $k$, the value $s_k = r_k + \paymentWit$ satisfies $g^{s_k} = \paymentStmt \cdot R_k$ and $A_k = g^{\alpha_k}$.

\lightpar{Soundness.} Suppose no $\witnessTarget \in [0,2^\ell]$ exists such that $\paymentWit$ can be recovered evaluating the OPRF with key $\sk$ matching $\pk$. 
					  Since all $\secpar$ committed instances are constructed using the \emph{same} $\witnessTarget$, no valid witness exists for any of the selected instances' OR statements either.
					  Hence for each selected index $j$ the prover must either produce an accepting transcript of~\cref{lem:orSchnorr} without a valid witness, or guess the set of opened elements correctly. 
					  By the $2$-special soundness of the sigma protocol, forging an OR proof succeeds with probability at most $1/p$. The probability of correctly guessing the opening is at most $\frac{((n/2)!)^2}{n!}$, which is also negligible. Soundness holds. 

\lightpar{Honest-verifier zero-knowledge.} The simulator draws the challenge set $\mathcal{I} \subset [\secpar]$ with $|\mathcal{I}|=\secpar/2$ uniformly and sets $\paymentWit := \log_g(\paymentStmt)$.
\begin{itemize}[leftmargin=*,nosep]
  \item For each $j \in \mathcal{I}$: sample $r_j, \alpha_j \sample \ZZ_p$, pick arbitrary $T_j \in \GG$, set $c_j = \hash_p(T_j)+r_j$, $R_j = g^{r_j}$, $A_j = g^{\alpha_j}$, and simulate the OR proof for $T_j$ via Lemma~\ref{lem:orSchnorr}.
  \item For each $k \notin \mathcal{I}$: sample $s_k, \alpha_k \sample \ZZ_p$, set $R_k = g^{s_k}/ \paymentStmt$, $A_k = g^{\alpha_k}$.
\end{itemize}
The opened instances provide no information about the target value $y_\target$ by the DDH assumption in $\GG$. 
Furthermore, all commitments are uniformly distributed, and the revealed values are identically distributed in the honest execution. 
\end{proof}

\begin{figure*}

	\begin{textdef}
		Public inputs: a group $(\GG,g,p)$, a hash function $\hash_\GG:\bin^* \to \GG$, and a statement $\proofstatement_{\mathsf{or}} = (\pk, A, T, \ell).$\\
		The prover provides a witness
		$
			\proofwitness_{\mathsf{or}} = (\sk,\alpha,y)
		$ 
		such that $y \in \bin^\ell$, $\pk = g^{\sk}$, $A = g^\alpha$, and $T = \hash_\GG(y)^{\sk \alpha}$.\\
		The verifier has no additional input.

		For every candidate $\hat y \in \bin^\ell$, define $h_{\hat y} := \hash_\GG(\hat y)$.
		The protocol proves that
		\[
			\exists \hat y \in \bin^\ell,\ \exists (\sk,\alpha) \in \ZZ_p^2 :
			\pk = g^{\sk},\ A = g^\alpha,\ T = h_{\hat y}^{\sk\alpha}.
		\]

		\textbf{Proving:}
		\begin{enumerate}[leftmargin=*]
			\item For the real branch $\hat y = y$, the prover computes
			$
				B_y := h_y^\alpha.
			$
			It samples $t_{\alpha}, t_{\sk} \sample \ZZ_p$ and computes
			\[
				a_{1,y} = g^{t_\alpha}, \qquad
				a_{2,y} = h_y^{t_\alpha}, \qquad
				a_{3,y} = g^{t_\sk}, \qquad
				a_{4,y} = B_y^{t_\sk}.
			\]

			\item For every simulated branch $\hat y \in \bin^\ell \setminus \{y\}$, the prover samples
			$
				c_{\hat y}, z_{\alpha,\hat y}, z_{\sk,\hat y}, \beta_{\hat y} \sample \ZZ_p,
			$ 
			defines
			$
				B_{\hat y} := h_{\hat y}^{\beta_{\hat y}},
			$
			and computes
			\[
				a_{1,\hat y} = g^{z_{\alpha,\hat y}} A^{-c_{\hat y}}, \qquad
				a_{2,\hat y} = h_{\hat y}^{z_{\alpha,\hat y}} B_{\hat y}^{-c_{\hat y}}, \qquad
				a_{3,\hat y} = g^{z_{\sk,\hat y}} \pk^{-c_{\hat y}}, \qquad
				a_{4,\hat y} = B_{\hat y}^{z_{\sk,\hat y}} T^{-c_{\hat y}}.
			\]

			\item The prover sends, for every $\hat y \in \bin^\ell$, the first message
			$
				\Big(B_{\hat y}, a_{1,\hat y}, a_{2,\hat y}, a_{3,\hat y}, a_{4,\hat y}\Big)
			$ 
			to the verifier.

			\item The verifier samples a challenge
			$
				c \sample \ZZ_p
			$ 
			and sends it to the prover.

			\item The prover sets
			$
				c_y := c - \sum_{\hat y \in \bin^\ell \setminus \{y\}} c_{\hat y} ,
			$ 
			and computes the real responses
			\[
				z_{\alpha,y} := t_\alpha + c_y \alpha ,
				\qquad
				z_{\sk,y} := t_\sk + c_y \sk .
			\]

			\item The prover sends, for every $\hat y \in \bin^\ell$, the tuple
			$
				(c_{\hat y}, z_{\alpha,\hat y}, z_{\sk,\hat y})
			$
			to the verifier.
		\end{enumerate}

		\textbf{Verification:}
		Upon receiving
		$
			\Big(B_{\hat y}, a_{1,\hat y}, a_{2,\hat y}, a_{3,\hat y}, a_{4,\hat y},
			c_{\hat y}, z_{\alpha,\hat y}, z_{\sk,\hat y}\Big)_{\hat y \in \bin^\ell},
		$ 
		the verifier accepts iff:
		\begin{enumerate}[leftmargin=*]
			\item The branch challenges sum to the global challenge:
			$
				\sum_{\hat y \in \bin^\ell} c_{\hat y} \equiv c \pmod p.
			$

			\item For every $\hat y \in \bin^\ell$, the following four equations hold:
			\[
				g^{z_{\alpha,\hat y}} \stackrel{?}{=} a_{1,\hat y} \cdot A^{c_{\hat y}},
				\qquad
				h_{\hat y}^{z_{\alpha,\hat y}} \stackrel{?}{=} a_{2,\hat y} \cdot B_{\hat y}^{c_{\hat y}}, \qquad
				g^{z_{\sk,\hat y}} \stackrel{?}{=} a_{3,\hat y} \cdot \pk^{c_{\hat y}},
				\qquad
				B_{\hat y}^{z_{\sk,\hat y}} \stackrel{?}{=} a_{4,\hat y} \cdot T^{c_{\hat y}}.
			\]
		\end{enumerate}
	\end{textdef}

	\caption{Interactive Schnorr OR proof for well-formedness of cut-and-choose openings (\cref{fig:proofs:cutandchoose}).}
	\label{fig:proofs:orSchnorr}
\end{figure*}

\begin{lemma}\label{lem:orSchnorr}
The sigma protocol in~\cref{fig:proofs:orSchnorr} is complete, $2$-special sound with knowledge error $1/p$, and honest-verifier zero-knowledge.
\end{lemma}
\begin{proof}
The protocol follows the standard Cramer--Damg{\aa}rd--Schoenmakers OR-composition.

\lightpar{Completeness.} For the real branch $\hat y = y$, the four verification equations hold by expanding $a_{1,y},\ldots,a_{4,y}$ and the responses $z_{\alpha,y}, z_{\sk,y}$. For every simulated branch $\hat y \neq y$, the equations hold by construction.

\lightpar{Special soundness.} Given two accepting transcripts with identical first messages $\{(B_{\hat y}, a_{1,\hat y},\ldots,a_{4,\hat y})\}_{\hat y}$ but distinct global challenges $c \neq c'$, there exists $y^* \in \bin^\ell$ with $c_{y^*} \neq c'_{y^*}$, since $\sum_{\hat y} c_{\hat y} = c$ and $\sum_{\hat y} c'_{\hat y} = c'$. Let $e = c_{y^*} - c'_{y^*} \neq 0$. From the first two equations for branch $y^*$:
\[
  \alpha = \frac{z_{\alpha,y^*} - z'_{\alpha,y^*}}{e},\qquad\text{satisfying}\quad A = g^\alpha\;\text{ and }\;B_{y^*} = h_{y^*}^\alpha.
\]
From the last two:
\[
  \sk = \frac{z_{\sk,y^*} - z'_{\sk,y^*}}{e},\qquad\text{satisfying}\quad \pk = g^\sk\;\text{ and }\;T = B_{y^*}^\sk = h_{y^*}^{\sk\alpha}.
\]
Hence $(\sk,\alpha,y^*)$ is a valid witness.

\lightpar{Honest-verifier zero-knowledge.} All simulated branches use the standard Schnorr simulation (responses and branch challenges uniform). The real branch contributes a pair of composed Chaum--Pedersen transcripts, each perfectly simulatable by choosing responses first. The joint distribution is therefore identical to a simulated transcript.
\end{proof}

\begin{figure*}
	\begin{textdef}
		Public inputs: a group $(\GG,g,p)$, a second generator $u \in \GG$, a hash function $\hash_\GG:\bin^* \to \GG$, and a statement
		$
			\proofstatement_{\mathsf{or}} = (\pk, U, V, \ell).
		$
		For every candidate $\hat y \in \bin^\ell$, define
		$
			h_{\hat y} := \hash_\GG(\hat y).
		$

		The prover provides a witness
		$
			\proofwitness_{\mathsf{or}} = (\sk,\rho,y)
		$
		such that $y \in \bin^\ell$, $\pk = g^{\sk}$, and
		\[
			U = u^\rho,
			\qquad
			V = h_y^{\sk}\cdot g^\rho.
		\]

		The protocol proves that
		$
			\exists \hat y \in \bin^\ell,\ \exists (\sk,\rho)\in\ZZ_p^2 :
			\pk = g^{\sk}
			\ \land\
			U = u^\rho
			\ \land\
			V = h_{\hat y}^{\sk}\cdot g^\rho.
		$

		\textbf{Proving:}
		\begin{enumerate}[leftmargin=*]
			\item For the real branch $\hat y = y$, the prover computes
			$
				M_y := h_y^{\sk}.
			$
			It samples
			$
				t_{\sk}, t_\rho \sample \ZZ_p
			$
			and computes
			\[
				a_{1,y} = g^{t_{\sk}},
				\qquad
				a_{2,y} = h_y^{t_{\sk}},
				\qquad
				a_{3,y} = u^{t_\rho},
				\qquad
				a_{4,y} = g^{t_\rho}.
			\]

			\item For every simulated branch $\hat y \in \bin^\ell \setminus \{y\}$, the prover samples
			$
				c_{\hat y}, z_{\sk,\hat y}, z_{\rho,\hat y}, \mu_{\hat y} \sample \ZZ_p,
			$
			defines
			$
				M_{\hat y} := h_{\hat y}^{\mu_{\hat y}},
			$
			and computes
			\[
				a_{1,\hat y} = g^{z_{\sk,\hat y}} \pk^{-c_{\hat y}},
				\qquad
				a_{2,\hat y} = h_{\hat y}^{z_{\sk,\hat y}} M_{\hat y}^{-c_{\hat y}},
				\qquad 
				a_{3,\hat y} = u^{z_{\rho,\hat y}} U^{-c_{\hat y}},
				\qquad
				a_{4,\hat y} = g^{z_{\rho,\hat y}} (V/M_{\hat y})^{-c_{\hat y}}.
			\]

			\item The prover sends, for every $\hat y \in \bin^\ell$, the first message
			$
				\bigl(M_{\hat y}, a_{1,\hat y}, a_{2,\hat y}, a_{3,\hat y}, a_{4,\hat y}\bigr)
			$
			to the verifier.

			\item The verifier samples a challenge
			$
				c \sample \ZZ_p
			$
			and sends it to the prover.

			\item The prover sets
			$
				c_y := c - \sum_{\hat y \in \bin^\ell\setminus\{y\}} c_{\hat y} \pmod p
			$
			and computes
			\[
				z_{\sk,y} := t_{\sk} + c_y \sk,
				\qquad
				z_{\rho,y} := t_\rho + c_y \rho.
			\]

			\item The prover sends, for every $\hat y \in \bin^\ell$, the tuple
			$
				(c_{\hat y}, z_{\sk,\hat y}, z_{\rho,\hat y})
			$
			to the verifier.
		\end{enumerate}

		\textbf{Verification:}
		Upon receiving
		$
			\Bigl(M_{\hat y}, a_{1,\hat y}, a_{2,\hat y}, a_{3,\hat y}, a_{4,\hat y},
			c_{\hat y}, z_{\sk,\hat y}, z_{\rho,\hat y}\Bigr)_{\hat y \in \bin^\ell},
		$
		the verifier accepts iff:
		\begin{enumerate}[leftmargin=*]
			\item
			$
				\sum_{\hat y \in \bin^\ell} c_{\hat y} \equiv c \pmod p.
			$

			\item For every $\hat y \in \bin^\ell$, the following equations hold:
			\[
				g^{z_{\sk,\hat y}} \stackrel{?}{=} a_{1,\hat y}\cdot \pk^{c_{\hat y}},
				\qquad
				h_{\hat y}^{z_{\sk,\hat y}} \stackrel{?}{=} a_{2,\hat y}\cdot M_{\hat y}^{c_{\hat y}},
			\qquad
				u^{z_{\rho,\hat y}} \stackrel{?}{=} a_{3,\hat y}\cdot U^{c_{\hat y}},
				\qquad
				g^{z_{\rho,\hat y}} \stackrel{?}{=} a_{4,\hat y}\cdot (V/M_{\hat y})^{c_{\hat y}}.
			\]
		\end{enumerate}
	\end{textdef}

	\caption{Interactive Schnorr OR proof that $(U,V)$ is well formed as $U=u^\rho$ and $V=h_y^{\sk}\cdot g^\rho$ for some $y \in \bin^\ell$.}
	\label{fig:proofs:orWF}
\end{figure*}

\begin{lemma}\label{lem:orWF}
The sigma protocol in~\cref{fig:proofs:orWF} is complete, $2$-special sound with knowledge error $1/p$, and honest-verifier zero-knowledge.
\end{lemma}
\begin{proof}
The argument is structurally identical to Lemma~\ref{lem:orSchnorr}. Each branch proves knowledge of $(\sk,\rho)$ such that $\pk = g^\sk$, $M_{\hat y} = h_{\hat y}^\sk$, $U = u^\rho$, and $V/M_{\hat y} = g^\rho$ via a pair of composed Chaum--Pedersen sub-arguments sharing the branch challenge $c_{\hat y}$.

\lightpar{Special soundness.} From two transcripts with distinct global challenges, there exists $y^* \in \bin^\ell$ with $c_{y^*} \neq c'_{y^*}$. From the first two equations, $\sk = (z_{\sk,y^*} - z'_{\sk,y^*})/(c_{y^*}-c'_{y^*})$, satisfying $\pk = g^\sk$ and $M_{y^*} = h_{y^*}^\sk$. From the last two, $\rho = (z_{\rho,y^*} - z'_{\rho,y^*})/(c_{y^*}-c'_{y^*})$, satisfying $U = u^\rho$ and $V = M_{y^*} \cdot g^\rho = h_{y^*}^\sk \cdot g^\rho$.

\lightpar{Honest-verifier zero-knowledge.} Follows by the same OR-simulation argument as Lemma~\ref{lem:orSchnorr}.
\end{proof}

\begin{figure*}

	\begin{textdef}
		Public inputs: a cyclic group $(\GG,g,p)$ of prime order $p$, a second generator $u \in \GG$, and a statement
		$
			\proofstatement_{\mathsf{wf}} = (X,Y,U_\rho,U_\alpha,T_u,T_g).
		$
		The prover provides a witness
		$
			\proofwitness_{\mathsf{wf}} = (h,\rho,\alpha)
		$
		such that
		$
			X = h \cdot g^\rho,
		$
		$
			Y = h^\alpha,
		$
		$
			U_\rho = u^\rho,
		$
		$
			U_\alpha = u^\alpha,
		$
		and
		$
			T_u = u^{\rho\alpha}.
		$
		The verifier has no additional input.

		The protocol proves that
		\[
			\exists (h,\rho,\alpha) \in \GG \times \ZZ_p^2 :
			X = h \cdot g^\rho,\ 
			Y = h^\alpha,\ 
			U_\rho = u^\rho,\ 
			U_\alpha = u^\alpha,\ 
			T_u = u^{\rho\alpha}.
		\]

		The prover first defines the auxiliary value
		$
			T_g := X^\alpha / Y.
		$
		For a valid witness, this satisfies
		$
			T_g = g^{\rho\alpha}.
		$
		The proof consists of three Chaum--Pedersen subproofs:
		\[
			\log_u(U_\alpha) = \log_{U_\rho}(T_u), \qquad
			\log_u(T_u) = \log_g(T_g), \qquad
			\log_u(U_\alpha) = \log_X(YT_g).
		\]

		\textbf{Proving:}
		\begin{enumerate}[leftmargin=*]
			\item The prover computes
			$
				T_g := X^\alpha / Y
			$
			and
			$
				\delta := \rho\alpha \in \ZZ_p.
			$

			\item For the first subproof
			$
				\log_u(U_\alpha) = \log_{U_\rho}(T_u),
			$
			the prover samples
			$
				r_1 \sample \ZZ_p
			$
			and computes
			\[
				a_{1,1} := u^{r_1},
				\qquad
				a_{1,2} := U_\rho^{r_1}.
			\]

			\item For the second subproof
			$
				\log_u(T_u) = \log_g(T_g),
			$
			the prover samples
			$
				r_2 \sample \ZZ_p
			$
			and computes
			\[
				a_{2,1} := u^{r_2},
				\qquad
				a_{2,2} := g^{r_2}.
			\]

			\item For the third subproof
			$
				\log_u(U_\alpha) = \log_X(YT_g),
			$
			the prover samples
			$
				r_3 \sample \ZZ_p
			$
			and computes
			\[
				a_{3,1} := u^{r_3},
				\qquad
				a_{3,2} := X^{r_3}.
			\]

			\item The prover sends
			$
				(T_g, a_{1,1},a_{1,2}, a_{2,1},a_{2,2}, a_{3,1},a_{3,2})
			$
			to the verifier.

			\item The verifier samples a challenge
			$
				c \sample \ZZ_p
			$
			and sends it to the prover.

			\item The prover computes
			\[
				z_1 := r_1 + c\alpha,
				\qquad
				z_2 := r_2 + c\delta,
				\qquad
				z_3 := r_3 + c\alpha
			\]
			and sends
			$
				(z_1,z_2,z_3)
			$
			to the verifier.
		\end{enumerate}

		\textbf{Verification:}
		Upon receiving
		$
			(T_g, a_{1,1},a_{1,2}, a_{2,1},a_{2,2}, a_{3,1},a_{3,2}, z_1,z_2,z_3),
		$
		the verifier accepts iff all six equations hold:
		\[
			u^{z_1} \stackrel{?}{=} a_{1,1} \cdot U_\alpha^c,
			\qquad
			U_\rho^{z_1} \stackrel{?}{=} a_{1,2} \cdot T_u^c,
		\qquad
			u^{z_2} \stackrel{?}{=} a_{2,1} \cdot T_u^c,
			\qquad
			g^{z_2} \stackrel{?}{=} a_{2,2} \cdot T_g^c,
		\qquad
			u^{z_3} \stackrel{?}{=} a_{3,1} \cdot U_\alpha^c,
			\ifnum\fullversion=0\qquad\else\;\fi
			X^{z_3} \stackrel{?}{=} a_{3,2} \cdot (Y T_g)^c.
		\]

		If all checks pass, then the verifier concludes that $T_u = u^{\rho\alpha}$, $T_g = g^{\rho\alpha}$, and $YT_g = X^\alpha$; hence
		\[
			Y \cdot g^{\rho\alpha} = X^\alpha = (h g^\rho)^\alpha = h^\alpha g^{\rho\alpha},
		\]
		which implies $Y = h^\alpha$ while $h$ remains hidden.
	\end{textdef}

	\caption{Interactive Chaum--Pedersen proof for well-formedness of $(X,Y)$ with hidden base $h$.}
	\label{fig:proofs:hiddenbaseWF}
\end{figure*}

\begin{lemma}\label{lem:hiddenbase}
The sigma protocol in~\cref{fig:proofs:hiddenbaseWF} is complete, $2$-special sound with knowledge error $1/p$, and honest-verifier zero-knowledge.
\end{lemma}
\begin{proof}
The protocol is a conjunction of three Chaum--Pedersen sub-arguments sharing a single challenge $c$.

\lightpar{Completeness.} For a valid witness $(h,\rho,\alpha)$, set $\delta := \rho\alpha$. Then $T_g = X^\alpha/Y = (hg^\rho)^\alpha/h^\alpha = g^{\rho\alpha} = g^\delta$. The responses $z_1 = r_1 + c\alpha$, $z_2 = r_2 + c\delta$, $z_3 = r_3 + c\alpha$ satisfy all six verification equations by direct computation.

\lightpar{Special soundness.} Given two accepting transcripts with the same first message but distinct challenges $c \neq c'$, let $e = c-c'$. From the first sub-argument, $\alpha = (z_1 - z'_1)/e$. From the second, $\delta = (z_2-z'_2)/e$, hence $\rho = \delta/\alpha$. Setting $h = X \cdot g^{-\rho}$, the third sub-argument confirms $h^\alpha = Y$, so $(h,\rho,\alpha)$ is a valid witness.

\lightpar{Honest-verifier zero-knowledge.} The simulator samples $c, z_1, z_2, z_3 \sample \ZZ_p$, sets $T_g \sample \GG$, and computes
\begin{align*}
a_{1,1} &:= u^{z_1}U_\alpha^{-c},\quad& a_{1,2} &:= U_\rho^{z_1}T_u^{-c},\\
a_{2,1} &:= u^{z_2}T_u^{-c},\quad& a_{2,2} &:= g^{z_2}T_g^{-c},\\
a_{3,1} &:= u^{z_3}U_\alpha^{-c},\quad& a_{3,2} &:= X^{z_3}(YT_g)^{-c}.
\end{align*}
Each pair is a perfect Chaum--Pedersen simulation. Since the commitment pairs are uniformly distributed in $\GG$, the simulated transcript is identically distributed to an honest one.
\end{proof}

%% file: appdx_figures.tex
\section{Additional Figures} \label{sec:appdx:figures}
In this section, we provide additional figures.

\input{2pdlogfuncfig.tex}
\input{bip65fig.tex}